%% file: all.tex
\documentclass[12pt]{article}
\usepackage{bbm,fullpage}
\usepackage{bm}
\usepackage{amsmath,amssymb,amsthm}
\usepackage{algorithm, pseudocode}
\usepackage{mathrsfs}
\usepackage{enumitem}
\usepackage[titles]{tocloft}
\usepackage{yfonts}
\usepackage{xcolor}
\usepackage[pdftex,                %
    bookmarks         = true,%     % Signets
    bookmarksnumbered = true,%     % Signets numérotés
    pdfpagemode       = None,%     % Signets/vignettes fermé à l'ouverture
    pdfstartview      = FitH,%     % La page prend toute la largeur
    pdfpagelayout     = SinglePage,% Vue par page
    colorlinks        = true,%     % Liens en couleur
   citecolor         =blue,
    urlcolor          = magenta,%  % Couleur des liens externes
    pdfborder         = {0 0 0}%   % Style de bordure : ici, pas de bordure
    ]{hyperref}%                   % Utilisation de HyperTeX

\usepackage{amsmath}
\usepackage{easybmat}
\usepackage{multirow,bigdelim}

\usepackage[indexonlyfirst,ucmark,toc]{glossaries}

\input{macros.tex}

\title{Solving determinantal systems using \\
homotopy techniques}

\author{J.D. Hauenstein$^{1}$, M. {Safey El Din}$^{2}$, \'E. Schost$^{3}$, T. X. Vu$^{2,3}$}

\newtheorem{pbm}{Problem}
\newtheorem{definition}{Definition}
\newtheorem{theorem}[definition]{Theorem}
\newtheorem{corollary}[definition]{Corollary}
\newtheorem{proposition}[definition]{Proposition}
\newtheorem{lemma}[definition]{Lemma}
\newtheorem{remark}[definition]{Remark}

\begin{document}

\maketitle
\footnotetext[1]{Department of Applied and Computational Mathematics and Statistics,
University of Notre Dame, USA}
\footnotetext[2]{Sorbonne Universit\'e, CNRS, INRIA, Laboratoire d'Informatique de Paris 6, PolSys, Paris, France}
\footnotetext[3]{David Cheriton School of Computer Science, University of Waterloo, ON, Canada}

\begin{abstract}
  Let $\KK$ be a field of characteristic zero and $\KKbar$ be an
  algebraic closure of $\KK$. Consider a sequence of polynomials
  $G=(g_1,\dots,g_s)$ in $\KK[X_1,\dots,X_n]$, a polynomial matrix
  $\mF=[f_{i,j}] \in \KK[X_1,\dots,X_n]^{p \times q}$, with $p \leq q$,
  and the algebraic set $\VpFG{p}{\mF}{G}$ of points in $\KKbar$ at
  which all polynomials in $\mG$ and all $p$-minors of $\mF$
  vanish. Such polynomial systems appear naturally in e.g. polynomial
  optimization, computational geometry.

  We provide bounds on the number of isolated points in
  $\VpFG{p}{\mF}{G}$ depending on the maxima of the degrees in rows
  (resp. columns) of $\mF$. Next, we design homotopy algorithms for
  computing those points. These algorithms take advantage of the
  determinantal structure of the system defining
  $\VpFG{p}{\mF}{G}$. In particular, the algorithms 
  run in time that is 
  polynomial in the bound on the number of isolated points.
\end{abstract}

\section{Introduction}\label{sec:intro}

Throughout, $\KK$ is a field of characteristic zero with algebraic
closure $\KKbar$, $(X_1, \ldots, X_n)$ is a set of $n$ variables, and
$\KK[X_1,\dots,X_n]$ is the multivariate polynomial ring in $n$
variables with coefficients in $\KK$.  With this setup, let
$\mF=[f_{i,j}] \in \KK[X_1,\dots,X_n]^{p \times q}$ be a polynomial
matrix, with $p \leq q$. The first question which will interest us in
this paper is to describe the set of points $\x \in \KKbar{}^n$ at
which the evaluation of the matrix $\mF$ has rank less than $p$.  In
the particular case $p=1$, this simply means finding all common
solutions of $f_{1,1},\dots,f_{1,q}$.

For any matrix $\mF$ over a ring $R$, and for any integer $r$,
$M_r(\mF)$ will denote the set of $r$-minors of $\mF$, and $I_r(\mF)$
will denote the ideal they generate in $R$. For any subset $I$ in
$\KK[X_1,\dots,X_n]$, $V(I)$ will denote the zero-set of $I$ in
$\KKbar{}^n$, and for a matrix $\mF$ with entries in
$\KK[X_1,\dots,X_n]$, we will write $V_r(\mF)=V(I_r(\mF))$. In
particular, for $\mF$ of size $p \times q$, with $p \le q$, the set of
points introduced in the previous paragraph is
$$\VpF{p}{\mF}=\{\bx \in \KKbar{}^n \mid \mathrm{rank}(\mF({\bx})) < p\}.$$
This is an algebraic set, since it is defined by the vanishing of
all maximal minors of $\mF$. 

We will discuss below dimension properties of $\VpF{p}{\mF}$.  Recall
that any algebraic set $V$ is the finite union of its
\emph{irreducible components}: these are the maximal irreducible
algebraic sets contained in it (an algebraic set is irreducible if it
is not the union of two proper algebraic sets). The {\em dimension} of
an algebraic set $V$ is the largest integer $d$ such that intersecting
$V$ with $d$ generic hyperplanes yields finitely many points; those
algebraic sets with all irreducible components of the same dimension
are called {\em equidimensional}. We refer to
e.g.~\cite[Chap.\ I and II]{Shafarevich77} for these notions.

For the problem above, it is natural to consider the case where $n =
q-p+1$.  Indeed, results due to Macaulay~\cite{Macaulay16} and Eagon
and Northcott~\cite{EN62} imply that all irreducible components of
$\VpF{p}{\mF}$ have dimension at least $n-(q-p+1)$; furthermore, in
the case $n = q-p+1$, $\VpF{p}{\mF}$ has dimension zero for a generic
choice of the entries of $\mF$ (this is proved for instance
in~\cite{Spa14}). Of course, even if we assume $n = q-p+1$,
$\VpF{p}{\mF}$ may have components of positive
dimension; in this case, we will be interested in describing only its
{\em isolated points}, that is, the points in the irreducible
components of $\VpF{p}{\mF}$ of dimension zero (this notion makes
sense for any field $\KK$; when $\KK=\mathbb{R}$, these points are
indeed isolated for the metric topology).

%% \begin{pbm} \label{problem} 
%%   Given a field $\KK$, a matrix $\mF \in \KK[X_1,\dots,X_n]^{p \times q}$ with $p
%%   \leq q$ and $n = q-p+1$, compute the isolated points of $S_\mF$.
%% \end{pbm}

Studying the set $\VpF{p}{\mF}$ is a particular case of a slightly
more general question. In addition to matrix $\mF$, we may indeed take
into account further equations of the form $g_1 =\cdots=g_s=0$, for
some $G=(g_1,\dots,g_s)$ in $\KK[X_1,\dots,X_n]$. In this setting, the
natural relation between the number $n$ of variables, the size of
$\mF$ and the number $s$ of polynomials in $G$ is now
$n=q-p+s+1$. Then, we define the algebraic set
$$\VpFG{p}{\mF}{G} = \{\bx \in \KKbar{}^n \mid
\mathrm{rank}(\mF({\bx})) < p \text{~and~} g_1(\bx)=\cdots=g_s(\bx)=0
\};$$ this is thus the zero-set of the ideal
$I_p(\mF) + \langle g_1,\dots,g_s\rangle$ (here
$\langle g_1,\dots,g_s\rangle$ denotes the ideal generated by
$g_1, \ldots, g_s$).  Our main problem is the following.
\begin{pbm} \label{problem2} 
  For a field $\KK$, a matrix $\mF \in \KK[X_1,\dots,X_n]^{p \times q}$ and
  polynomials $G=(g_1,\dots,g_s)$ in $\KK[X_1,\dots,X_n]$ such that 
  $p \leq q$ and   $n = q-p+s+1$, compute the isolated points of $\VpFG{p}{\mF}{G}$.
\end{pbm}
This problem appears in a variety of context; prominent examples are
optimization problems~\cite{GSZ10,JP14,BGHS14,GS14,NDS06}, and related
questions in real algebraic
geometry~\cite{ARS,BaGiHeMb01,BaGiHePa05,BGHSS,BRSS,RealDecompICMS,
  CellDecompSurface,BertiniReal,RealNumerical, SaSc03,SaSc11,SaSc17},
where $\mF$ consists of the Jacobian matrix of $G$, together with one
extra row, corresponding to the gradient of a function that we want to
optimize on $V(G)$. Because they show up several times in this introduction, we will refer to this particular class of inputs 
as systems {\em coming from optimization}.

In several of these situations, we are only interested in the
solutions of the system made of minors $M_p(\mF)$ and
$G=(g_1,\dots,g_s)$ at which the associated Jacobian matrix has full
rank. This subset of solutions is finite and is always a subset of the
set of isolated points of $\VpFG{p}{\mF}{G}$ \cite[Theorem
  16.19]{Eisenbud95}; we call these points {\em simple 
  points}. The set of simple  points coincides with
$\VpFG{p}{\mF}{G}$ when the system $M_p(\mF),G$ generates a radical
ideal of dimension zero; this case appears frequently in the context of
algorithms in real algebraic geometry~\cite{BGHSS}.  

Hence, it also makes sense to look at the following slight variant of
Problem~\eqref{problem2}.

\begin{pbm} \label{problem3} For a field $\KK$, a matrix
  $\mF \in \KK[X_1,\dots,X_n]^{p \times q}$ and polynomials
  $G=(g_1,\dots,g_s)$ in $\KK[X_1,\dots,X_n]$ with $p \leq q$ and
  $n = q-p+s+1$, compute the simple  points of~$\VpFG{p}{\mF}{G}$.
\end{pbm}

We will represent the output of our algorithm using univariate
polynomials. Let $V \subset \KKbar{}^n$ be a zero-dimensional
algebraic set defined over $\KK$. A \emph{zero-dimensional
  parametrization} $\scrR = ((w,v_1, \ldots, v_n), \lambda)$ of $V$
consists of polynomials $(w,v_1, \ldots, v_n)$ such that $w \in
\KK[Y]$ is monic and squarefree, all $v_i$'s are in $\KK[Y]$ and satisfy
$\deg(v_i) < \deg(w)$, and $\lambda$ is a $\KK$-linear form in $n$
variables, such that
\begin{itemize}
\item $\lambda(v_1, \ldots, v_n) = Yw'$ mod $w$ with $w'=\frac{\partial w}{\partial Y}$;
\item we have $V = Z(\scrR)$, with $$Z(\scrR)= \left\{\left(\frac{v_1(\tau)}{w'(\tau)}, \ldots, \frac{v_n(\tau)}{w'(\tau)}\right) \ | \ w(\tau) = 0\right\}.$$
\end{itemize}
The constraint on $\lambda$ then says that the root of $w$ are the
values taken by $\lambda$ on $V$. This representation was introduced
in~\cite{Kronecker82,Macaulay16}, and has been used in a variety of
algorithms, such as those
in~\cite{GiMo89,GiHeMoPa95,ABRW,GiHeMoMoPa98,Rouillier99,GiLeSa01}.
The reason why we use a rational parametrization, with $w'$ as a
denominator, goes back to~\cite{ABRW, Rouillier99, GiLeSa01}: when
$\KK=\Q$, this allows us to control precisely the bit-size of the
coefficients, using bounds such as those
in~\cite{Schost03,DaSc04}. The same phenomenon holds with $\KK=k(T)$,
for a field $k$, in which case we want to control degrees in $T$ of
the numerators and denominators of the coefficients of $\scrR$.

Our first result gives a bound on the multiplicities of the solutions
of $\VpFG{p}{\mF}{G}$. To state it, we need the following notation.
Take $\mF=[f_{i,j}]_{1 \le i \le p, 1 \le j \le q}$ in
$\KK[X_1,\dots,X_n]^{p \times q}$.  We will consider two degree
measures for matrix $\mF$; these have been used before for
determinantal ideals, see for instance~\cite{NieRan09,MiSt04}. For
$i=1,\dots,p$, we will write $\rdeg(\mF,i)$ for the degree of the
$i$th row of $\mF$, that is,
$\rdeg(\mF,i)=\max(\deg(f_{i,j}))_{1 \le j \le q}$; similarly, for
$j=1,\dots,q$, we write $\cdeg(\mF,j)$ for the degree of the $j$th
column of $\mF$, that is,
$\cdeg(\mF,j)=\max(\deg(f_{i,j}))_{1 \le i \le p}$. For $k \ge 0$,
$$E_k(\delta_1,\dots,\delta_q)=\sum_{1\leq i_1 < \cdots < i_k \leq
  n}\delta_{i_1} \cdots \delta_{i_k}$$ is the elementary symmetric
polynomial of degree $k$ in $(\delta_1, \ldots, \delta_q)$ and
$$S_k(\alpha_1,\dots,\alpha_p) = \sum_{i_1+\cdots+i_p=k, i_j \geq
  0}\alpha_1^{i_1}\cdots\alpha_p^{i_p}$$ is the $k$th complete
symmetric polynomial in $(\alpha_1,\dots,\alpha_p)$.

Finally, we recall the notion of multiplicity of a point $\bx$ with
respect to an ideal $I$ in $\KKbar[X_1,\dots,X_n]$; this notion
extends to ideals in $\KK[X_1, \ldots, X_n]$ by considering their
extension in $\KKbar[X_1, \ldots, X_n]$. We refer
to~\cite[Chap.\ 3]{Eisenbud95} and~\cite[Chap.\ 4]{CLO_UAG} for more
details on the following notions.

The ideal $I$ can be written as the intersection of finitely many
primary components, that is, $I=Q_1\cap\cdots \cap Q_r$ for some
primary ideals $Q_1,\dots,Q_r$; this decomposition is said to be
minimal when $V(Q_i)\neq V(Q_j)$ for $i\neq j$. Take $\bx$ isolated in
$V(I)$; then there exists a unique primary component $Q_i$, which must
has dimension zero, such that $\bx$ is in $V(Q_i)$; because we take
a primary decomposition over $\KKbar$, we actually have
$V(Q_i)=\{\bx\}$. Although minimal primary decompositions are not
unique, the fact that $\bx$ is isolated implies that $Q_i$ does not
depend on the primary decomposition of $I$ we consider; then, the
\emph{multiplicity} of $\bx$ is defined as the dimension of
$\KKbar[X_1,\dots,X_n]/Q_i$. When $\bx=0\in\KKbar{}^n$,
the dimension of $\KKbar[X_1,\dots,X_n]/Q_i$ is the same as that of
$\KKbar[[X_1, \ldots, X_n]]/I$, where $\KK[[X_1, \ldots, X_n]]$ denotes
the formal power series ring in $X_1, \ldots, X_n$ with coefficients
in $\KKbar$ (this follows from~\cite[Theorem 4.2.2]{CLO_UAG}). 

%%  (to see this we use $\KK[X_1,\dots,X_n]/Q_i = \KKbar[[X_1,
%%     \ldots, X_n]]/Q_i$ \cite[Theorem 2.2, Chap. 4]{CLO_UAG} and the
%% straightforward equality $Q_i. \KKbar[[X_1, \ldots, X_n]] =
%% I. \KKbar[[X_1, \ldots, X_n]]$).

The following is our first~result.

\begin{theorem}\label{theo:1}
  Let $\mF$ be in $\KK[X_1,\dots,X_n]^{p \times q}$ and let
  $G=(g_1,\dots,g_s)$ be in $\KK[X_1,\dots,X_n]$, with $p \le q$ and
  $n=q-p+s+1$. Then, the sum of the multiplicities of the isolated
  points of $I_p(\mF) + \langle g_1,\dots,g_s \rangle$ is at most
  $\min(c,c')$ with
$$c=\deg(g_1) \cdots \deg(g_s) E_{n-s}(\cdeg(\mF,1), \ldots, \cdeg(\mF,q))$$
and
$$c'=\deg(g_1) \cdots \deg(g_s) S_{n-s}(\rdeg(\mF,1), \ldots, \rdeg(\mF,p)).$$
\end{theorem}
When $\rdeg(\mG,i)=\cdeg(\mF,j)=d$ for all $i,j$, the two bounds given
above coincide, with common value $\deg(g_1) \cdots \deg(g_s) {q \choose {p-1}} d^{n-s}$; otherwise, either of the two expressions
$E_{n-s}(\cdeg(\mF,1), \ldots, \cdeg(\mF,q))$ and
$S_{n-s}(\rdeg(\mF,1), \ldots, \rdeg(\mF,p))$ can be the minimum. For 
instance, consider the case $s=0$ (so there are no equations $G$),
and where the degrees of the entries in $\mF$ are 
$$ \begin{bmatrix}
    2 & 1 & 5 & 7 \\
    2 & 1 & 5 & 7 \\
    2 & 1 & 5 & 7 
  \end{bmatrix}.$$
Here, we have $p=3, q=4, s=0$ and $n=2$. Then, 
the quantity $c$ is $c=E_2(2,1,5,7) = 2\cdot1+2\cdot5+2\cdot7+1\cdot5+1\cdot7+5\cdot7 = {73}$,
whereas $c'=6 \cdot 7^2=294.$ On the other hand, if we 
take $\mF$ with degree profile
$$ \begin{bmatrix}
    2 & 2 & 2 & 2 \\
    1 & 1 & 1 & 1 \\
    5 & 5 & 5 & 5 
  \end{bmatrix},$$
with the same values of $p,q,s,n$, we get $c=6 \cdot 7^2=294$ and
$c'=S_2(2,1,5) = 2^2+2\cdot 1 + 2\cdot 5 + 1^2 + 1 \cdot 5 + 5^2 =
{47}$.  For systems coming from optimization, where $F$ is a
Jacobian matrix, we are in a
situation similar to the second example, where the $i$th row degree of
$F$ is simply the degree of the corresponding equation, minus one.

Previous work by Miller and Sturmfels~\cite[Chapter~15]{MiSt04} proved
very general results on the multi-degrees of determinantal ideals
built from matrices with indeterminate entries (in which case we have
$s=0$, but the assumption $n=q-p+1$ does not hold); in particular,
they obtain analogues (and generalizations) of the result in
Theorem~\ref{theo:1} in that context.

Nie and Ranestad proved in~\cite{NieRan09} that the bounds in
Theorem~\ref{theo:1} are tight for two families of polynomials
(in a similar context where the polynomials are homogeneous in
$n+1$ variables):
\begin{itemize}
\item when entries of $\mF$ are generic and homogeneous, and
 such that $\deg(f_{i,j}) = \cdeg(\mF,j)$ for all $i,j$, the ideal
 $I_p(\mF)$ has degree $E_{n}(\cdeg(\mF,1), \ldots, \cdeg(\mF,q))$;
\item when entries of $\mF$ are  generic and homogeneous, and
  such that $\deg(f_{i,j}) = \rdeg(\mF,i)$ for all $i,j$, the
  ideal  $I_p(\mF)$ has degree $S_{n}(\rdeg(\mF,1), \ldots, \rdeg(\mF,p))$.
\end{itemize}
From this, they deduce that the degree of the ideal $I_p(\mF) +
\langle g_1,\dots,g_s \rangle$ is at most \sloppy $\deg(g_1) \cdots
\deg(g_s) S_{n-s}(\rdeg(\mF,1), \ldots, \rdeg(\mF,p))$, for systems
coming from optimization problems, assuming that this ideal has
dimension zero. In this context, Spaenlehauer gave in~\cite{Spa14} an
explicit expression for the Hilbert function of the ideal $I_p(\mF) +
\langle g_1,\dots,g_s \rangle$, for a generic input.

\medskip

Our second result gives bounds on the cost of computing a
zero-dimensional parametrization of the isolated solutions of
$\VpFG{p}{\mF}{G}=V(I_p(\mF) + \langle g_1,\dots,g_s \rangle)$. Our
algorithms take as input a \emph{straight-line program} (that is, a
sequence of elementary operations $+, -, \times$) that computes the
entries of $\mF$ and $G$ from the input variables $X_1,\dots,X_n$; the
\emph{length $\sigma$} of the input is the number of operations it
performs. This assumption is not restrictive, since any matrix $\mF$
and polynomials $G$ can be computed by a straight-line program (a
naive solution would consist in computing and adding all monomials in
$\mF$ and $G$).

\begin{theorem}\label{theo:2}
  Suppose that matrix $\mF \in \KK[X_1,\dots,X_n]^{p \times q}$ and
  polynomials $G=(g_1,\dots,g_s)$ in $\KK[X_1,\dots,X_n]$ are given by
  a straight-line program of length $\sigma$. Assume that
  $\deg(g_1),\dots,\deg(g_s)$, as well as
  $\cdeg(\mF,1), \ldots, \cdeg(\mF,q)$ and
  $\rdeg(\mF,1), \ldots, \rdeg(\mF,p)$ are all at least equal to $1$.

  Then, there exist randomized algorithms that solve
  Problem~\eqref{problem2} in either
   $$\softO\left (
     {q \choose p} c(e+c^5 )\big(\sigma + q \delta + \gamma  \big )
   \right)$$
  operations in $\KK$, with
  \begin{align*}
    c&=\deg(g_1)\cdots\deg(g_s)\ E_{n-s}(\cdeg(\mF,1), \ldots, \cdeg(\mF,q))\\
    e&=(\deg(g_1)+1)\cdots(\deg(g_s)+1)\ E_{n-s}(\cdeg(\mF,1)+1, \ldots, \cdeg(\mF,q)+1),\\
    \gamma&= \max(\deg(g_i), 1\leq i \leq s)\\
    \delta &= \max(\cdeg(\mF,i), 1\leq i \leq q)
  \end{align*}
  or 
   $$\softO\left (
     {q \choose p} c'(e'+{c'}^5 )\big(\sigma + p \alpha  +\gamma \big )
   \right)$$
  operations in $\KK$, with 
\begin{align*}
  c'&=\deg(g_1)\cdots\deg(g_s)\ S_{n-s}(\rdeg(\mF,1), \ldots, \rdeg(\mF,p))\\
  e'&=(\deg(g_1)+1)\cdots(\deg(g_s)+1)\ S_{n-s}(\rdeg(\mF,1)+1, \ldots, \rdeg(\mF,p)+1),\\
    \gamma&= \max(\deg(g_i), 1\leq i \leq s)\\
    \alpha &= \max(\rdeg(\mF,j), 1\leq j \leq p).
\end{align*}
\end{theorem}
The assumption that all degrees are at least $1$ is not a restriction.
If $\deg(g_i)=0$ for some $i$, $g_i$ is a constant, so either the
system is inconsistent (if $g_i \ne 0$) or $g_i$ can be
discarded. Similarly, if say $\cdeg(\mF,i)=0$, the $i$th column of
$\mF$ consists of constants; after applying linear combinations with
coefficients in $\KK$ to the rows of $\mF$, we may assume that all
entries in the $i$th column, except at most one, are non-zero without
changing the column degrees. The $i$th column of $\mF$ (and the row of
the non-zero entry, if there is one) can then be discarded.

Remark further that in the common situation where all degrees 
$\deg(g_i)$, $\rdeg(\mF,i)$ and $\cdeg(\mF,j)$ involved
in the formulas above are at least equal to $2$, we have the
inequalities $e \le c^2$, $e' \le {c'}{}^2$ and $\binom{q}{p}\leq c$,
$\binom{q}{p}\leq c'$; as a result, the runtimes become {\em
  polynomial} in respectively $c,\sigma$ and $c',\sigma$. This is to
be compared with Theorem~\ref{theo:1}, which shows that $\min(c,c')$
is a natural upper bound for the output size of such~algorithms.

For solving Problem~\eqref{problem3}, one obtains slightly better
complexity estimates. 

\begin{theorem}\label{theo:3}
  Suppose that the matrix $\mF \in \KK[X_1,\dots,X_n]^{p \times q}$
  and polynomials $G=(g_1,\dots,g_s)$ in $\KK[X_1,\dots,X_n]$ are
  given by a straight-line program of length $\sigma$. Assume that
  $\deg(g_1),\dots,\deg(g_s)$, as well as
  $\cdeg(\mF,1), \ldots, \cdeg(\mF,q)$ and
  $\rdeg(\mF,1), \ldots, \rdeg(\mF,p)$ are all at least equal to $1$.

  Then, there exist randomized algorithms that solve
  Problem~\eqref{problem3} in either
   $$\softO\left (
     {q \choose p} ce\big(\sigma + q \delta + \gamma  \big )
   \right)$$
or 
   $$\softO\left (
     {q \choose p} c'e'\big(\sigma + p \alpha  +\gamma \big )
   \right)$$
  operations in $\KK$, 
  all notation being as in Theorem~\ref{theo:3}.
\end{theorem}
As above, in the common situation where all degrees involved are at
least $2$, the runtimes become {\it polynomial} in $c, \sigma$ and
$c',\sigma'$.

The probabilistic aspects are as follows: at several steps, the
algorithms on which Theorems~\ref{theo:2} and~\ref{theo:3} rely will
draw elements from the base field at random. In all cases, there
exists an algebraic hypersurface $\cal H$ of the parameter space such
that success is guaranteed for all choices of parameters not
in~$\cal H$.

%% The number ${q \choose p}$ is the number of elements in the
%% determinantal system of $\mF$. In several cases, we can improve the
%% algorithm and replace this by the number $q$ of columns of $\mF$, for
%% instance, when all isolated solutions of $I_\mF$ are known to have
%% multiplicity one.

As already said, our algorithms are based on a {\em symbolic homotopy
  continuation}.  Homotopy continuation algorithms have become a
foundational tools for numerical algorithms, either in continuation of
Shub and Smale's early work~\cite{ShSm93}, or along the lines of work
by Morgan, Sommese, Wampler (as summarized, for instance,
in~\cite{BertiniBook,SoWa05}), with an emphasis on the algebraic
geometry underlying these techniques. In this context, dedicated
numerical homotopy algorithms has also been developed to take into
account sparsity in polynomial systems (see
e.g. \cite{Ver94,Ver09,AdVe13}).

By contrast, their usage in symbolic contexts is more
recent. Early references are~\cite{HeKrPuSaWa99,BoMaWaWa04}, which
deal with systems with no particular structure; further work extended
this idea to sparse systems (in the polyhedral
sense)~\cite{JeMaSoWa09,HeJeSa10,HeJeSa13,HeJeSa14} and
multihomogeneous systems~\cite{HeJeSaSo02,SaSc16}.  In~\cite{SaSc16},
these techniques are used to solve Problem~\eqref{problem3}, but the
complexity estimates obtained there depend on multi-homogeneous
B\'ezout bounds involving the maxima of $\rdeg(\mF, 1), \ldots,
\rdeg(\mF, p)$ or $\cdeg(\mF,1), \ldots, \cdeg(\mF, q)$.

Most algorithms in the previous references have in common that they
solve {\em square} systems, that is, systems with as many equations as
unknowns; extensions of these methods can deal with systems of
positive dimension by essentially using variants of the algorithm for
square systems.  One notable exception is given in \cite{SVV10} where
dedicated homotopy algorithms are given to solve Schubert problems
which consist in determining linear spaces of prescribed dimension
which meet a set of fixed linear subspaces in specified
dimensions. Observe that such problems are formulated with rank
conditions on some special matrices (see e.g. \cite{LDSVV18}). These
algorithms strongly exploit and are dedicated to the structure of the
Schubert problem through the Littlewood Richardson rule and an
associated combinatorial construction. Hence, as far as we know, they
cannot be used to solve determinantal systems of equations expressing
that a given matrix with polynomial entries is rank deficient. 

One of the contributions in this paper is to deal with determinantal
systems of equations, which are in essence over-determined; this is
made possible by the algebraic properties of determinantal ideals.

It is well-known that Gr\"obner bases behave rather well on
over-determined systems. Starting from the determination of the
Hilbert function of a determinantal ring due to Conca and
Herzog~\cite{CH94}, complexity estimates are given in \cite{FSS13,
  FSS12} for computing Gr\"obner basis of ideals generated by either
$M_r(\mF)$ when $r\leq p\leq q$, or $G, M_{p}(\mF)$ (for inputs coming
from optimization problems), but under some genericity assumptions on
the entries of $\mF$ or $G$, which are also assumed to all have the
same degree. This series of works culminated with the result obtained
by Spaenlehauer in \cite{Spa14}, where he removes this latter degree
assumption and provides sharp complexity statements, still under
genericity~assumptions.

Systems encoding rank defects in polynomial matrices have also been
studied in the scope of the so-called geometric resolution algorithm
in \cite{BaGiHeLeMaSo15}, with a slight generalization in \cite{SaSp16}
computing simple  solutions (those isolated solutions which
are not simple are not considered in this line of work). As our
algorithm here, these algorithms take as input straight-line programs
but instead of using deformation techniques to build a global
homotopy, determinantal systems are solved incrementally in some
chart.  Hence, the complexity of these algorithms depends here on the
maximum degrees of the varieties defined by the considered
intermediate systems. Even without taking into account the dimension
assumption, additional results are needed to compare these
intermediate degrees with the quantities involved in our complexity
estimates.

\medskip In the following paragraphs, we describe our results in more
detail.  As a preliminary, we will need an algorithm which takes as
input polynomials $\bC=(c_1,\dots,c_m)$ and a point $\bx$ in the
zero-set of $\bC$, and which decides whether $\bx$ is an isolated
points of $V(\bC)$. This  will be used to solve
Problem~\eqref{problem2}.

Without any other information, this decision problem is difficult to
solve in a good complexity. However, when a bound $\mu$ is known on
the multiplicity of $\bx$ as a root of $\bC$, it becomes possible to
solve this problem in time polynomial in the number of equations $m$,
the number of variables $n$, the bound $\mu$, and the complexity of
evaluation $\sigma$ of $\bC$. This is detailed in
Section~\ref{sec:isolated}, where we explain how to modify an
algorithm by Mourrain~\cite{Mourrain97} and adapt it to our context.

In Section~\ref{sec:homotopy}, we give an algorithm which takes as
input a sequence of polynomials $\bC$ and computes a zero-dimensional
parametrization of the isolated points of $V(\bC)$, assuming the
existence of a suitable homotopy deformation. Explicitly, we suppose
that $\bC$ involves variables $\bX=(X_1,\dots,X_n)$, we let $T$ be a
new variable, and we suppose that we know a family of polynomials
$\bB$ in $\KK[T,\bX]$ such that $\bB(1,\bX)=\bC$. Let then $\bA$ be
the polynomials $\bB(0,\bX)$, and suppose that $V(\bA)$ is finite, and
that we are able to find a zero-dimensional parametrization of it
efficiently. We will actually need a few further conditions (for
instance, at all points in $V(\bA)$, the Jacobian matrix of these
polynomials must have rank $n$).

We will see in Section~\ref{sec:homotopy} that when all these
conditions hold, we can divise a homotopy algorithm that lifts the
points of $V(\bA)$ (that correspond to $T=0$) into a curve $\cal C$
parametrized by $T$. The isolated points of $V(\bC)$ all belong to the
fiber of $\cal C$ above $T=1$, but some points in this fiber can
actually lie in positive dimensional components of $V(\bC)$; the
algorithm of Section~\ref{sec:isolated} will filter out such
points. The complexity we obtain depends linearly on the complexity of
evaluating $\bC$ and polynomially on a bound on the sum of the
multiplicites of the isolated points of $V(\bC)$ and the degree of the
homotopy curve. When one only wants to compute simple solutions, a
variant of the homotopy algorithm is given: this is actually
simpler since we replace the algorithm of
Section~\ref{sec:isolated} with a simple criterion allowing us to
identify the simple solutions.

We will apply these results to our determinantal problems as
follows. Given $\mF \in \KK[\bX]^{p \times q}$ and $G=(g_1,\dots,g_s)$,
we will build a matrix
\[\mU = (1-T)\cdot \mL + T \cdot \mF \in \KK[T, \X]^{p \times q}\]
that connects a \emph{start matrix} $\mL$ to the target matrix $\mF$,
together with a homotopy of the form
\[V = (1-T) \cdot K + T \cdot G,\]
that connects a start system $K=(k_1,\dots,k_s)$ to the target system
$G$.  In Section~\ref{sec:check}, we prove that several assumptions of
the algorithm of Section~\ref{sec:homotopy} are satisfied for such
systems, independently of the choice of $\mL$ and $K$.

The actual construction of the system $K$ will be rather
straightforward; the difficulty lies in the definition of a matrix
$\mL$ that will respect either the column-degree or the row-degree of
$\mF$ (while satisyfing all assumptions for the algorithm of
Section~\ref{sec:homotopy}).  The column-degree case is treated in
Section~\ref{sec:columndegree} in a rather straightforward way,
whereas the row-degree case is more delicate, and is treated in
Sections~\ref{sec:prel-row} and~\ref{sec:rowdegree}. In both cases, we
bound the sum of the multiplicities of the isolated points in
$\VpFG{p}{\mF}{G}$ (thereby establishing Theorem~\ref{theo:1}), as
well as the degree of the homotopy curve.

%%%%%%%%%%%%%%%%%%%%%%%%%%%%%%%%%%%%%%%%%%%%%%%%%%%%%%%%%%%%
%%%%%%%%%%%%%%%%%%%%%%%%%%%%%%%%%%%%%%%%%%%%%%%%%%%%%%%%%%%%
%%%%%%%%%%%%%%%%%%%%%%%%%%%%%%%%%%%%%%%%%%%%%%%%%%%%%%%%%%%%

\section{A local dimension test} \label{sec:isolated}

Let $\LL$ be a field containing the field $\KK$ and $\LLbar$ be an
algebraic closure of $\LL$.  Let $\bC=(c_1,\dots,c_m)$ be polynomials
in $\KK[\bX]$, with $\bX=(X_1,\dots,X_n)$. Given a point $\bx$ with
coordinates in $\LL$ that belongs to the zero-set $V(\bC)\subset
\LLbar{}^n$, we discuss here how to decide whether $\bx$ is an
isolated point in $V(\bC)$. We make the following assumption in the
rest of this section:
\begin{description}[leftmargin=*]
\item [$\assH.$] We are given as input an integer $\mu$ such that
  either $\bx$ is isolated in $V(\bC)$, with multiplicity at most
  $\mu$ with respect to the ideal $\langle \bC\rangle$, or $\bx$
  belongs to a positive-dimensional component of $V(\bC)$.
\end{description}
Without loss of generality, we also assume that $m\ge n$ (otherwise, $\bx$
cannot be an isolated solution). 
\begin{proposition}\label{prop:testisolated}
  Suppose that $\bC$ is given by a straight-line program of length $\sigma$.
  If assumption $\assH$ is satisfied, we can decide whether $\bx$ is an
  isolated point of $V(\bC)$ using 
$$O(n^4 \mu^4 + n^2 m \mu^3 + n \sigma \mu^4) \subset (\mu\,\sigma\,m)^{O(1)}$$ operations in~$\LL$.
\end{proposition}
Reference~\cite{BaHaPeSo09} gives an algorithm to compute the
dimension of $V(\bC)$ at $\bx$, but its complexity is not known to us,
as it relies on linear algebra with matrices of potentially large size
(not necessarily polynomial in $\mu,\sigma,m$).  Instead, we use an
adaptation of a prior result by Mourrain~\cite{Mourrain97}, which
allows us to control the size of the matrices we handle. We only give
detailed proofs for new ingredients that are specific to our context,
a key difference being the cost analysis in the straight-line program
model: Mourrain's original result depends on the number of monomials
appearing when we expand $\bC$, which would be too high for the
applications we will make of this result. Remark that the assumption
that $\KK$ (and thus $\LL$) have characteristic zero is needed for 
Mourrain's algorithm.

We assume henceforth that $\bx=0$; this is done by replacing $\bC$ by
the polynomials $\bC(\bX+\bx)$, which have complexity of evaluation
$\sigma'=\sigma+n$.  The basis of our algorithm is the following
remark.

\begin{lemma}
  Let $I$ be the zero-dimensional ideal
  $\langle \bC \rangle + \m^{\mu+1}$, where
  $\m=\langle X_1,\dots,X_n\rangle$ is the maximal ideal at the
  origin. Then, $0$ is isolated in $V(\bC)$ if and only if the
  multiplicity $d$ of $I$ at the origin is at most $\mu$.
\end{lemma}
\begin{proof}
  This follows from the following
  result~\cite[Theorem~A.1]{BaHaPeSo09}.  For $k \ge 1$, let $I_k$ be
  the zero-dimensional ideal $\langle \bC \rangle + \m^{k}$, and let
  $\nu_k$ be the multiplicity of the origin with respect to this
  ideal. Then, the reference above proves that the sequence
  $(\nu_k)_{k \ge 1}$ is non-decreasing, and that $0$ is isolated in
  $V(\bC)$ if and only if there exists $k\ge 1$ such that
  $\nu_k=\nu_{k+i}$ for any $i\geq 0$.
  \begin{itemize}
  \item If $0$ is isolated in $V(\bC)$, then by assumption $\assH$ 
    its multiplicity with respect to $\langle \bC\rangle$ is at most $\mu$,
    and its multiplicity $d$ with respect to $I$ cannot be larger.
  \item Otherwise, by the result above, $\nu_{k+1} > \nu_k$ holds for
    all $k \ge 1$, so that $\nu_k \ge k$ holds for all such $k$ (since
    $\nu_1=1$). In particular, the multiplicity $d$ of 
 $I$ at the origin, which is $\nu_{\mu+1}$, is at least $\mu+1$.
    \qedhere
  \end{itemize}
\end{proof}

Hence, we are left with deciding whether the multiplicity $d$ of the
ideal $I$ at the origin is at most $\mu$; remark that this
multiplicity is equal to the dimension of $\LL[\bX]/I$, since $I$ is
$\m$-primary.  We do this by following and slightly modifying
Mourrain's algorithm for the computation of the orthogonal
$I^{\perp}$, that is, the set of $\LL$-linear forms $\LL[\bX] \to \LL$
that vanish on $I$; this is a $\LL$-vector space naturally identified
with the dual of $\LL[\bX]/I$, so it has dimension $d$, the
multiplicity of $I$ at the origin.

We do not need to give all details of the algorithm, let alone proof
of correctness; we just mention the key ingredients for the cost
analysis in our setting. 

The algorithm represents the elements in $I^{\perp}$ by means of {\em
  multiplication matrices}. An important feature of $I^{\perp}$ is
that it admits the structure of a $\LL[\bX]$-module: for $k$ in
$\{1,\dots,n\}$ and $\beta$ in $I^{\perp}$, the $\LL$-linear form $X_k
\cdot \beta: f \mapsto \beta(X_k f)$ is easily seen to still lie in
$I^{\perp}$.  In particular, if
$\bbeta=(\beta_1,\dots,\beta_d)$ is an $\LL$-basis of
$I^{\perp}$, then for all $k$ as above, and all $i$ in
$\{1,\dots,d\}$, $X_k \cdot \beta_i$ is a linear combination of
$\beta_1,\dots,\beta_d$. Mourrain's algorithm computes a basis
$\bbeta=(\beta_1,\dots,\beta_d)$ with the following features:
\begin{itemize}
\item for $i$ in $\{1,\dots,d\}$ and $k$ in $\{1,\dots,n\}$, we have
  $X_k \cdot \beta_i=\sum_{0 \le j < i} \lambda^{(k)}_{i,j} \beta_j$
  (hence $\lambda^{(k)}_{i,j}$ may be non-zero 
  only for $j<i$);
\item $\beta_1$ is the evaluation at $0$, $f \mapsto f(0)$;
\item for $i$ in $\{2,\dots,d\}$, $\beta_i(1)=0$.
\end{itemize}
The following lemma shows that the coefficients $(\lambda^{(k)}_{i,j})$
are sufficient to evaluate  the linear forms $\beta_i$ at any $f$ in
$\LL[\bX]$. More precisely, knowing only their values for $j < i \le s$,
for any $s \le d$, allows us to evaluate $\beta_1,\dots,\beta_s$ at such an $f$.
The following lemma follows~\cite{Mourrain97} in its description
of the matrices $\bM_{k,s}$; the (rather straightforward) complexity analysis 
in the straight-line program model is new.
\begin{lemma}\label{lemma:evalbeta}
   Let $s$ be in $1,\dots,d$, and suppose that the coefficients
  $\lambda^{(k)}_{i,j}$ are known for $i=1,\dots,s$, $j=0,\dots,i-1$
  and $k=1,\dots,n$. Given a straight-line program $\Gamma$ of length
  $\sigma$ that computes $\h=(h_1,\dots,h_R)$, one can compute
  $\beta_i(h_r)$, for all $i=1,\dots,s$ and $r=1,\dots,R$, using
  $O(s^3\,\sigma)$ operations in $\LL$.
\end{lemma}
\begin{proof}
  By definition, for $h$ in $\LL[\bX]$ and $k=1,\dots,n$, the following equality
  holds:
$$
  \begin{bmatrix}
    \beta_1(X_k h)\\
    \vdots\\
    \beta_s(X_k h)
  \end{bmatrix}=
\bM_{k,s}
  \begin{bmatrix}
    \beta_1(h)\\
    \vdots\\
    \beta_s(h)
  \end{bmatrix},
\quad\text{with}\quad
\bM_{k,s}= \begin{bmatrix}
    \lambda^{(k)}_{1,1} & \cdots & \lambda^{(k)}_{s,1}\\
    \vdots && \vdots \\
    \lambda^{(k)}_{1,s} & \cdots & \lambda^{(k)}_{s,s}
  \end{bmatrix}.
$$ 
 Remark that the matrices $\bM_{k,s}$ all commute with each other. Indeed, 
for any $k,k'$ in $\{1,\dots,n\}$, and $h$ as above, the relation above implies
that 
$$
\Delta_{k,k',s}
  \begin{bmatrix}
    \beta_1(h)\\
    \vdots\\
    \beta_s(h)
  \end{bmatrix} =
  \begin{bmatrix}
0\\ \vdots \\ 0 
  \end{bmatrix},
$$
where $\Delta_{k,k',s} = \bM_{k,s}\bM_{k',s}-\bM_{k',s}\bM_{k,s}.$ Because 
the linear forms $\beta_1,\dots,\beta_s$ are linearly independent, this implies
that all rows of $\Delta_{k,k',s}$ must be zero, as claimed.
We then deduce that for any polynomial $h$ in $\LL[\bX]$, we have
the equality
$$  \begin{bmatrix}
    \beta_1(h)\\
    \vdots\\
    \beta_s(h)
  \end{bmatrix} =
h(\bM_{1,s},\dots,\bM_{n,s})   \begin{bmatrix}
    \beta_1(1)\\
    \vdots\\
    \beta_s(1)
  \end{bmatrix}. $$ On the other hand, our assumptions imply that the
  sequence $(\beta_1(1),\dots,\beta_s(1))$ is simply $(1,0,\dots,0)$.
  To prove the lemma, it is then enough to note that the evaluations \sloppy
  $h_1(\bM_{1,s},\dots,\bM_{n,s}),\dots,h_R(\bM_{1,s},\dots,\bM_{n,s})$
  can be computed using the straight-line program doing
  $O(s^3\,\sigma)$ operations.
\end{proof}

Mourrain's algorithm proceeds in an iterative manner, starting from
$\bbeta^{(1)}=(\beta_{1})$ (and setting $e_1=1$), and computing
successively $\bbeta^{(2)}=(\beta_{e_1+1},\dots,\beta_{e_2})$,
$\bbeta^{(3)}=(\beta_{e_2+1},\dots,\beta_{e_3})$, \dots for some
integers $e_1 \le e_2 \le e_3 \dots$ Mourrain's algorithm stops when
$e_{\ell+1}=e_{\ell}$, in which case $\beta_1,\dots,\beta_{e_\ell}$ is
an $\LL$-basis of $I^\perp$, and $e_\ell=d$. In our case, we
are not interested in computing this multiplicity, but only in
deciding whether it is less than or equal to the parameter $\mu$. We do it as follows: assume that we have
computed $\bbeta^{(1)},\bbeta^{(2)},\dots,\bbeta^{(\ell)}$, together
with the corresponding integers $e_1,e_2,\dots,e_\ell$, with $e_1 <
\cdots < e_\ell \le \mu$. We compute $\bbeta^{(\ell+1)}$ and $e_{\ell+1}$,
and continue according to the following:
\begin{itemize}
\item if $e_{\ell+1}=e_{\ell}$, we conclude that the multiplicity
  $d$ of $I$ at the origin is $e_\ell \le \mu$; we stop the
  algorithm;
\item if $e_{\ell+1} > \mu$, we conclude that this multiplicity is greater 
  than $\mu$; we stop the algorithm;
\item else, when $e_\ell < e_{\ell+1} \le \mu$, we do another loop.
\end{itemize}
Because the $e_\ell$'s are an increasing sequence of integers, they
satisfy $e_\ell \ge \ell$; hence, every time we enter the loop above we
have $\ell \le \mu$. To finish the analysis of the algorithm, it
remains to explain how to compute $\bbeta^{(\ell+1)}$ from
$(\bbeta^{(1)},\bbeta^{(2)},\dots,\bbeta^{(\ell)})=(\beta_{1},\dots,\beta_{e_\ell})$.

As per our description above, at any step of the algorithm,
$\beta_{1},\dots,\beta_{e_\ell}$ are represented by means of the
coefficients $\lambda^{(k)}_{i,j}$, for $0 \le j < i \le e_{\ell}$ and
$1 \le k \le n$.  At step $\ell$, Mourrain's algorithm solves a homogeneous linear system
$T_\ell$ with $n(n-1) e_\ell/2+m'$ equations and $n e_\ell$ unknowns,
where $m'$ is the number of generators of the ideal $I= \langle \bC
\rangle + \m^{\mu+1}$. Remark that $m'$ is not polynomial in $\mu$ 
and $n$, so the size of $T_\ell$ is {\em a priori} too large to 
fit our cost bound; we will explain below how to resolve this issue.

The nullspace dimension of this linear system gives us the cardinality
$e_{\ell+1}-e_{\ell}$ of $\bbeta^{(\ell+1)}$. Similarly, the coordinates of
the $e_{\ell+1}-e_{\ell}$ vectors in a nullspace basis are precisely
the coefficients $\lambda^{(k)}_{i,j}$ for
$i=e_{\ell}+1,\dots,e_{\ell+1}$, $j=1,\dots,e_\ell$ and $k=1,\dots,n$
(we have $\lambda^{(k)}_{i,j}=0$ for $j=e_{\ell}+1,\dots,i-1$). For
all $\ell \ge 2$, all linear forms $\beta$ in $\bbeta^{(\ell)}$ are
such that for all $k$ in $\{1,\dots,n\}$, $X_k \cdot \beta$ belongs to
the span of $\bbeta^{(1)},\dots,\bbeta^{(\ell-1)}$; in particular, a
quick induction shows that all linear forms in
$\bbeta^{(1)},\dots,\bbeta^{(\ell)}$ vanish on all monomials of degree
at least $\ell$.

There remains the question of setting up the system $T_\ell$. For $k$
in $\{1,\dots,n\}$ and an $\LL$-linear form $\beta$, we denote by
$X_k^{-1} \cdot \beta$ the $\LL$-linear form defined by $\LL$-linearity as follows:
\begin{itemize}
\item $(X_k^{-1} \cdot \beta)(X_k f) = \beta(f)$ for any monomial $f$ in $\LL[\bX]$,
\item $(X_k^{-1} \cdot \beta)(f)=0$ if $f\in \LL[\bX]$ is a monomial which does not depend on $X_k$.
\end{itemize}
In other words,
$(X_k^{-1} \cdot \beta)(f)=\beta(\delta_k(f))$ holds for all $f$,
where $\delta_k:\LL[\bX] \to \LL[\bX]$ is the $k$th divided difference
operator
$$f\mapsto \frac
{f(X_1,\dots,X_n)-f(X_1,\dots,X_{k-1},0,X_{k+1},\dots,X_n)}{X_k}.$$
One verifies that, as the notation suggests, $X_k \cdot (X_k^{-1}
\cdot \beta)$ is equal to $\beta$. This being said, we can then
describe what the entries of $T_\ell$ are:
\begin{itemize}
\item the first $n(n-1) e_\ell/2$ equations involve only the coefficients 
  $\lambda^{(k)}_{i,j}$ previously computed (we refer to~\cite[Section~4.4]{Mourrain97} for details of how exactly 
these entries are distributed in $T_\ell$, as we do not need such details here).
\item each of the other $m'$ equations has coefficient vector
{\small
$$v_f = \big (\
 (X_k^{-1} \cdot \beta_1)(f(X_1,\dots,X_k,0,\dots,0)),\dots,\ (X_k^{-1} \cdot \beta_{e_\ell})(f(X_1,\dots,X_k,0,\dots,0))\
\big )_{1 \le k \le n},$$}
where $f$ is a generator of $I=\langle \bC \rangle +\m^{\mu+1}$.
\end{itemize}
We claim that only those equations corresponding to generators
$c_1,\dots,c_m$ of the input system $\bC$ are useful, as all others are identically
zero.

We pointed out above that any linear form $\beta_i$ in
$\beta_1,\dots,\beta_{e_\ell}$ vanishes on all monomials of degree at
least $\ell$. Since we saw that we must have $\ell \le \mu$, all
$\beta_i$ as above vanish on monomials of degree $\mu$; this implies
that $X_k^{-1}\cdot \beta_i$ vanishes on all monomials of degree
$\mu+1$. The generators $f$ of $\m^{\mu+1}$ have degree $\mu+1$, and
for any such $f$, $f(X_1,\dots,X_k,0,\dots,0)$ is either zero, or of
degree $\mu+1$ as well. Hence, for any $k$, $\beta_i$ in
$\beta_1,\dots,\beta_{e_\ell}$ and $f$ as above, $(X_k^{-1} \cdot
\beta_i)(f(X_1,\dots,X_k,0,\dots,0))$ vanishes. This implies that the
vector $v_f$ is identically zero for such an $f$, and that the
corresponding equation can be discarded.

Altogether, as claimed above, we see that we have to compute the
values
$$(X_k^{-1} \cdot \beta_i)(c_j(X_1,\dots,X_k,0,\dots,0)),$$ for
$k=1,\dots,n$, $i=1,\dots,e_\ell$ and $j=1,\dots,m$.  Fixing $k$, we
let $\bC_k = (c_{j,k})_{1 \le j \le m}$, where $c_{j,k}$ is the
polynomial $c_j(X_1,\dots,X_k,0,\dots,0)$; note that the system
$\bC_k$ can be computed by a straight-line program of length
$\sigma'= \sigma+n$. Then, applying the following lemma with
$s=e_\ell \le \mu$ and $\h = \bC_k$, we deduce that the values
$(X_k^{-1} \cdot \beta_i)(c_j(X_1,\dots,X_k,0,\dots,0))$, for $k$
fixed, can be computed in time $O(\mu^3 (\sigma+n))$.

\begin{lemma}
  Let $s$ be in $1,\dots,d$, and suppose that the coefficients
  $\lambda^{(k)}_{i,j}$ are known for $i=1,\dots,s$, $j=0,\dots,i-1$
  and $k=1,\dots,n$. Given a straight-line program $\Gamma$ of length
  $\sigma$ that computes $\h=(h_1,\dots,h_R)$ and given $k$ in
  $\{1,\dots,n\}$, one can compute $(X_k^{-1}\cdot \beta_i)(h_r)$, for
  all $i=1,\dots,s$ and $r=1,\dots,R$, using $O(s^3 (\sigma+n))$
  operations in $\LL$.
\end{lemma}
\begin{proof}
  In view of the formula $(X_k^{-1} \cdot
  \beta)(f)=\beta(\delta_k(f))$, and of Lemma~\ref{lemma:evalbeta}, it is
  enough to prove the existence of a straight-line program of length
  $O(\sigma+n)$ that computes $(\delta_k(h_1),\dots,\delta_k(h_R))$.

  To do this, we replace all polynomials
  $\gamma_{-n+1},\dots,\gamma_\sigma$ computed by $\Gamma$ by terms
  $\eta_{-n+1},\dots,\eta_\sigma$ and $\nu_{-n+1},\dots,\nu_\sigma$,
  with
  $\eta_\ell=\gamma_\ell(X_1,\dots,X_{k-1},0,X_{k+1},\dots,X_n)$
  and $\nu_\ell$ in $\LL[\bX]$ such that
  $\gamma_\ell= \eta_\ell+X_k \nu_\ell$ holds for all $\ell$, so
  that in particular $\nu_\ell=\delta_k(\gamma_\ell)$.  To compute
  $\eta_\ell$ and $\nu_\ell$, assuming all previous
  $\eta_{\ell'}$ and $\nu_{\ell'}$ are known, we proceed as
  follows:
  \begin{itemize}
  \item if $\gamma_\ell=X_k$, we set $\eta_\ell=0$ and $\nu_\ell=1$;
  \item if $\gamma_\ell=X_{k'}$, with $k' \ne k$, we set $\eta_\ell=X_{k'}$ and $\nu_\ell=0$;
  \item if $\gamma_\ell =c_\ell$, with $c_\ell \in \LL$,
    then we set $\eta_\ell=c_\ell$ and  $\nu_\ell=0$;
  \item if $\gamma_\ell = \gamma_{a_\ell} \pm \gamma_{b_\ell}$,
    for some indices $a_\ell,b_\ell < \ell$, 
    then we set $\eta_\ell=\eta_{a_\ell}\pm\eta_{b_\ell}$
    and $\nu_\ell=\nu_{a_\ell}\pm\nu_{b_\ell}$;
\item if $\gamma_\ell = \gamma_{a_\ell} \gamma_{b_\ell}$,
      for some indices $a_\ell,b_\ell < \ell$,
    then we set $\eta_\ell=\eta_{a_\ell} \eta_{b_\ell}$
    and $$\nu_\ell=
\eta_{a_\ell} \nu_{b_\ell}
+
\nu_{a_\ell} \eta_{b_\ell}
+
X_k\nu_{a_\ell} \nu_{b_\ell}.$$
\end{itemize}
One verifies that in all cases, the relation $\gamma_\ell=
\eta_\ell+X_k \nu_\ell$ still holds. Since the previous
construction allows us to compute $\eta_\ell$ and $\nu_\ell$ in
$O(1)$ operations from the knowledge of all previous $\eta_{\ell'}$
and $\nu_{\ell'}$, we deduce that all $\eta_\ell$ and $\nu_\ell$,
for $\ell=-n+1,\dots,\sigma$, can be computed by a straight-line program of
length $O(\sigma+n)$.
\end{proof}

Taking all values of $k$ into account, we see that we can compute all
entries we need to set up the linear system $T_\ell$ using $O(\mu^3
n(\sigma+n))$ operations in $\LL$. After discarding the useless equations
described above, the numbers of equations and unknowns in the system
$T_\ell$ are respectively at most $n^2 \mu+m$ and $n \mu$; this
implies that we can find a nullspace basis of it in time $O(n^4 \mu^3
+ n^2 m \mu^2)$. Altogether, the time spent to find
$\bbeta^{(\ell+1)}$ from
$(\bbeta^{(1)},\bbeta^{(2)},\dots,\bbeta^{(\ell)})=(\beta_{1},\dots,\beta_{e_\ell})$
is $O(n^4 \mu^3 + n^2 m \mu^2 + n \sigma \mu^3)$.

Since we saw that we do at most $\mu$ such loops, the cumulated time
is $O(n^4 \mu^4 + n^2 m \mu^3 + n \sigma \mu^4)$, and
Proposition~\ref{prop:testisolated} is proved.

%%%%%%%%%%%%%%%%%%%%%%%%%%%%%%%%%%%%%%%%%%%%%%%%%%%%%%%%%%%%
%%%%%%%%%%%%%%%%%%%%%%%%%%%%%%%%%%%%%%%%%%%%%%%%%%%%%%%%%%%%
%%%%%%%%%%%%%%%%%%%%%%%%%%%%%%%%%%%%%%%%%%%%%%%%%%%%%%%%%%%%

\section{Symbolic homotopies}\label{sec:homotopy}

In this section, we work over our field $\KK$, still using $n$ variables
$\bX=(X_1,\dots,X_n)$. Given polynomials $\bC=(c_1,\dots,c_m)$ in
$\KK[\bX]^m$, we give algorithms to compute a zero-dimensional
parametrization of the isolated points (or simple  points) of
$V(\bC)$, assuming the existence of a suitable {\em homotopy
  deformation} of $\bC$. We assume $m\ge n$, otherwise no isolated
point exists in $V(\bC)$.

Let $T$ be a new variable and consider polynomials
$\bB=(b_1,\dots,b_m)$ in $\KK[T,\bX]$; for $\tau$ in $\KKbar$,
we write $\bB_\tau=(b_{\tau,1},\dots,b_{\tau,m})=\bB(\tau,\bX)\subset
\KKbar[\bX]$ and we assume that $\bB$ is such that $\bB_1=\bC$.
Define further the ideal $J=\langle \bB \rangle \subset \KKbar[T,\bX]$ and
consider the folllowing assumptions.
\begin{description}[leftmargin=*]
\item[$\assA_1.$] Any irreducible component of $V(J) \subset
  \KKbar{}^{n+1}$ has dimension at least one.
\item[$\assA_2.$] For any maximal ideal $\m \subset\KKbar[T,\bX]$, if the
  localization $J_\m \subset \KKbar[T,\bX]_\m$ has height $n$, then it is
  unmixed (that is, all associated primes have height $n$).
\end{description}
An obvious example where such properties hold is when $m=n$. Then,
$\assA_1$ is Krull's theorem, and $\assA_2$ is Macaulay's unmixedness
theorem in the Cohen-Macaulay ring
$\KKbar[T,\bX]_\m$~\cite[Corollary~18.14]{Eisenbud95}. More generally,
these properties hold when $\bB$ is the sequence of $p$-minors of a $p
\times q$ matrix with entries in $\KK[T,\bX]$, with $p \le q$ and $n=q-p+1$; we
discuss this, and a slightly more general situation, in Section~\ref{sec:check}.  

For $\tau$ in $\KKbar$, we further denote by $\assG(\tau)$ the
following three properties.
\begin{description}[leftmargin=*]
\item[$\assG_1(\tau).$] For $k=1,\dots,m$,
  $\deg_\bX(b_k)=\deg_\bX(b_{\tau,k})$ (where $\deg_\bX$ denotes the
  degree in $\bX$).
\item[$\assG_2(\tau).$] The only common solution to
  $b_{\tau,1}^H(\tau,\bX)=\cdots=b_{\tau,m}^H(\tau,\bX)=0$ is
  $(0,\dots,0)\in\KKbar{}^n$, where for $k=1,\dots,m$, $b_{\tau,k}^H$ is
  the polynomial in $\KKbar[X_0,\bX]$ obtained by homogenizing
  $b_{\tau,k}$ using a new variable $X_0$. In particular, $V(\bB_\tau)
  \subset \KKbar{}^n$ is finite.
\item[$\assG_3(\tau).$] The ideal $\langle \bB_\tau \rangle$ is
  radical in $\KKbar[\bX]$.
\end{description}

The first result in this section is the following.
\begin{proposition}\label{prop:degree_fiber}
  Suppose that assumptions $\assA_1$ and $\assA_2$ hold. Then, there
  exists an integer $c$ such that for all $\tau$ in $\KKbar$, the sum
  of the multiplicities of the isolated solutions of $\bB_\tau$ is at
  most $c$, and is equal to $c$ if $\assG(\tau)$ holds.
\end{proposition}

We next give our algorithms for
\begin{itemize}
\item computing the isolated solutions of the polynomial system
  $\bC=(c_1,\dots,c_m)$;
\item computing the simple solutions of the polynomial system
  $\bC$.
\end{itemize}
In order to control the cost of the algorithm, we introduce the
following assumptions.
\begin{description}[leftmargin=*]
\item[${\assD}_1$.] We are given $\tau$ in $\KK$ such that $\assG(\tau)$
  holds; without loss of generality, we assume that $\tau=0$. We also
  suppose that we know a description of $V(\bB_{0})$ by means of a
  zero-dimensional parametrization  $\scrR_0 =((w_{0},v_{0,1},\dots,v_{0,n}),\lambda)$ with
  coefficients in $\KK$. The linear form $\lambda$ needs to satisfy
  some genericity requirements, that are described in
  Subsection~\ref{proof:prop2}.
\item[${\assD}_2$.] We know an integer $e$ such that the union of the
  one-dimensional components of $V(J)$ in $\KKbar{}^{n+1}$ has degree
  at most $e$ (we prove that $e \ge c$ in Lemma~\ref{lemma:e-geq-c}).
\item[${\assD}_3$.] We can compute $\bB$ using a straight-line program
  of length $\sigma$.
\end{description}
Then, the second main result in this section is the following.
\begin{proposition}\label{prop:compute_isolated}
  Assume that ${\assD}_1, {\assD}_2$ and ${\assD}_3$ hold. Let $c$
  be as in Proposition~\ref{prop:degree_fiber}. There exists a
  randomized algorithm ${\sf Homotopy}$ which computes a zero-dimensional
  parametrization of the isolated points of $V(\bC)$ using
  $$\softO(c^5 m n^2  + c(e+c^5) n(\sigma + n^3)) \subset (e\,\sigma\,m)^{O(1)}$$
  operations in~$\KK$. 
\end{proposition}
The variant below focuses on the computation of simple 
points. We reuse the notations introduced above.
\begin{proposition}\label{prop:compute_regular}
  Under the assumptions of Proposition~\ref{prop:compute_isolated},
  there exists a randomized algorithm ${\sf Homotopy\_simple}$  which computes a
  zero-dimensional parametrization of the simple  points 
  of~$V(\bC)$ using
  $$\softO( c^2\, m \,n^2 + \,c\,e\,n (\sigma+n^2) )\subset (e\,\sigma\,m)^{O(1)}$$
  operations in~$\KK$.
\end{proposition}

%%%%%%%%%%%%% TODO: check that
%% \begin{remark}
%%   These algorithms rely on the choice of a {\em generic} $\KK$-linear
%%   form. As already said, this form is used in the returned rational
%%   parametrization and must satisfy some properties which are described
%%   further (see Subsection~\ref{proof:prop2}). When the chosen form
%%   does not fulfill these requirements, either the algorithm fails or
%%   the output parametrization encodes a subset of the simple isolated
%%   $V(\bC)$.
%% \end{remark}

%%%%%%%%%%%%%%%%%%%%%%%%%%%%%%%%%%%%%%%%%%%%%%%%%%%%%%%%%%%%

\subsection{Proof of Proposition~\ref{prop:degree_fiber}}

This subsection is devoted to prove
Proposition~\ref{prop:degree_fiber}. In the course of the proof, we
will give a precise characterization of the integer $c$ mentioned in
the proposition, although the statement given in the proposition will
actually be enough for our further purposes. {\em In all the rest of
  this subsection, we assume that $\assA_1$ and $\assA_2$ hold.}

Consider an irredundant primary decomposition of the ideal $J=\langle \bB\rangle$ in
$\KKbar[T,\bX]$, of the form $J=Q_1 \cap \cdots
\cap Q_r$, and let $P_1,\dots,P_r$ be the associated primes, that is,
the respective radicals of $Q_1,\dots,Q_r$. We assume that
$P_1,\dots,P_s$ are the minimal primes, for some $s \le r$, so that
$V(P_1),\dots,V(P_s)$ are the (absolutely) irreducible components of
$V(J)\subset \KKbar{}^{n+1}$. By $\assA_1$, these irreducible
components all have dimension at least one. Refining further, we
assume that $t \le s$ is such that $V(P_1),\dots,V(P_t)$ are the
 irreducible components of $V(J)$ of dimension one whose
image by $\pi_T: (\tau,x_1,\dots,x_n) \mapsto \tau$ is Zariski dense in
$\KKbar$.

\begin{lemma}\label{lemma:vPi}
  Let $\tau$ be in $\KKbar$ and let $\bx \in \KKbar{}^n$ be an isolated
  solution of the system $\bB_\tau$. Then, $(\tau,\bx)$ belongs to $V(P_i)$
  for at least one index $i$ in $\{1,\dots,t\}$, and does not belong
  to $V(P_i)$ for any index $i$ in $\{t+1,\dots,r\}$.
\end{lemma}
\begin{proof}
  Because $(\tau,\bx)$ cancels $\bB$, it belongs to at least one of
  $V(P_1),\dots,V(P_r)$. It remains to rule out the possibility that
  $(\tau,\bx)$ belongs to $V(P_i)$ for some index $i$ in
  $\{t+1,\dots,r\}$.

  We first deal with indices $i$ in $\{t+1,\dots,s\}$. These are those
  primary components with minimal associated primes $P_i$ that either
  have dimension at least two, or have dimension one but whose image
  by $\pi_T$ is a single point. In both cases, all irreducible
  components of the intersection $V(P_i)\cap V(T-\tau)$ have dimension
  at least one. Since $\bx$ is isolated in $V(\bB_\tau)$, $(\tau,\bx)$ is
  isolated in $V(\bB)\cap V(T-\tau)$, so it cannot belong to
  $V(P_i)\cap V(T-\tau)$ for any $i$ in $\{t+1,\dots,s\}$.
  
  We conclude by proving that $(\tau,\bx)$ does not belong to $V(P_i)$,
  for any of the embedded primes $P_{s+1},\dots,P_r$. We proceed by
  contradiction, assuming for definiteness that $(\tau,\bx)$ belongs to
  $V(P_{s+1})$. Because $P_{s+1}$ is an embedded prime, $V(P_{s+1})$
  is contained in (at least) one of $V(P_1),\dots,V(P_s)$. In view of
  the previous paragraph, it cannot be one of
  $V(P_{t+1}),\dots,V(P_s)$.  Now, all of $V(P_1),\dots,V(P_t)$ have
  dimension one, so $V(P_{s+1})$ has dimension zero (so it is the point $\{(\tau,\bx)\}$). For the same
  reason, if $(\tau,\bx)$ belonged to another $V(P_i)$, for some $i >
  s+1$, $V(P_i)$ would also be zero-dimensional, and thus equal to $\{(\tau,\bx)\}$; as a result, $V(P_i)$
  would be equal to $V(P_{s+1})$, and this would contradict the
  irredundancy of our decomposition.
  
  To summarize, $(\tau,\bx)$ belongs to $V(P_{s+1})$, together with
  $V(P_i)$ for some indices $P_i$ in $\{1,\dots,t\}$ (say
  $P_1,\dots,P_u$, up to reordering, for some $u \ge 1$), and avoids
  all other associated primes.  Let us localize the decomposition
  $J=Q_1 \cap \cdots \cap Q_r$ at
  $P_{s+1}$. By~\cite[Proposition~4.9]{AtMc},
  $J_{P_{s+1}}={Q_1}_{P_{s+1}} \cap \cdots \cap {Q_u}_{P_{s+1}}\cap
  {Q_{s+1}}_{P_{s+1}}$ is an irredundant primary decomposition of
  $J_{P_{s+1}}$ in $\KKbar[T,\bX]_{P_{s+1}}$; the minimal primes are
  ${P_1}_{P_{s+1}},\dots,{P_u}_{P_{s+1}}$.

  By Corollary~4 p.24 in~\cite{Matsumura86}, for any prime
  ${P_i}_{P_{s+1}}$, $i=1,\dots,u$ or $i=s+1$, the localization of
  $\KKbar[T,\bX]_{P_{s+1}}$ at ${P_i}_{P_{s+1}}$ is equal to
  $\KKbar[T,\bX]_{P_{i}}$. In particular, the height of ${P_i}_{P_{s+1}}$
  in $\KKbar[T,\bX]_{P_{s+1}}$ is equal to that of $P_i$ in
  $\KKbar[T,\bX]_{P_{i}}$, that is, $n$ if $i=1,\dots,u$, since then
  $V(P_i)$ has dimension $1$, or $n+1$ if $i=s+1$. Since $u \ge 1$,
  this proves that $J_{P_{s+1}}$ has height $n$. As a result, $\assA_2$ implies that $J_{P_{s+1}}$ is unmixed, a contradiction.
\end{proof}

Let us write $J=J' \cap J''$, with $J'=Q_1 \cap \cdots \cap Q_t$ and
$J''=Q_{t+1} \cap \cdots \cap Q_r$. For $\tau$ in $\KKbar$, we denote
by $J_\tau \subset \KKbar[T,\bX]$ the ideal $J + \langle T-\tau \rangle$,
and similarly for $J'_\tau$ and $ J''_\tau$.

\begin{lemma}\label{lemma:JJprime}
  Let $\tau$ and $\bx$ be as in Lemma~\ref{lemma:vPi}. Then, the
  multiplicities of the ideals $J_\tau$ and $J'_\tau$ at $(\tau,\bx)$
  are the same.
\end{lemma}
\begin{proof}
  Without loss of generality, assume that $\tau=0 \in \KKbar$ and
  $\bx=0 \in \KKbar{}^n$. We start from the equality $J=J' \cap J''$,
  which holds in $\KKbar[T,\bX]$, and we see it in the formal power
  series ring $\KKbar[[T,\bX]]$.  The previous lemma implies that
  there exists a polynomial in $J''$ that does not vanish at
  $(\tau,\bx)=0 \in \KKbar{}^{n+1}$.  This polynomial is a unit in
  $\KKbar[[T,\bX]]$, which implies that the extension of $J''$ in
  $\KKbar[[T,\bX]]$ is the trivial ideal $\langle 1 \rangle$, and
  finally that the equality of extended ideals $J=J'$ holds in
  $\KKbar[[T,\bX]]$. This implies the equality
  $J+\langle T \rangle =J'+\langle T \rangle $ in $\KKbar[[T,\bX]]$,
  and the conclusion follows.
\end{proof}

Our goal is now to give a bound on the sum of the multiplicites of
$\bB_\tau$ at all its isolated roots, for any $\tau$ in $\KKbar$.

To achieve this, we consider the Puiseux series field
$\SS = \KKbar\langle \langle T\rangle \rangle$ in $T$ with
coefficients in $\KKbar$.  Since $\KKbar$ is algebraically closed and
of characteristic $0$, $\SS$ is algebraically closed (actually, it is
an algebraic closure of $\KKbar(T)$) and hence a perfect field.

Next, we consider the extension $\frak{J}$ of $J$ in $\SS[\bX]$, and
similarly $\frak{J}'$ and ${\frak J}''$ denote extensions of $J'$ and
$J''$ in $\SS[\bX]$.

\begin{lemma}\label{lemma:dimJprime}
  The ideal $\frak{J}'$ has dimension zero and $V(\frak{J}') \subset
  \SS^n$ is the set of isolated solutions of
  $V(\frak{J}) \subset \SS^n$.
  % The ideal $\frak{J}'$ has dimension zero and $V(\frak{J}') \subset
  % \overline{\KK(T)}{}^n$ is the set of isolated solutions of
  % $V(\frak{J}) \subset \overline{\KK(T)}{}^n$.
\end{lemma}
\begin{proof}
 From the equality $J=J' \cap J''$ and Corollary~3.4 in~\cite{AtMc},
 we deduce that $\frak{J}=\frak{J'} \cap \frak{J''}$. The properties
 of $J'$ (that the irreducible components of $V(J')$ are precisely those
 irreducible components of $V(J)$ that have dimension one and with a
 dense image by $\pi_T$) imply our claim.
\end{proof}

Let us write $c=\dim_{\SS}(\SS[\bX]/{\frak J}')$. Because $\SS$ is an
algebraic closure of $\KKbar(T)$, one has
$\dim_{\KKbar(T)} ( \KKbar(T)[\bX]/\tilde{J'} ) = c$ where
$\tilde{J'}$ is the extension of $J'$ in $\KKbar(T)[\bX]$.

The following lemma relates this quantity to the multiplicities of the
solutions in any fiber $\bB_\tau$. This proves the first statement in
Proposition~\ref{prop:degree_fiber}.

\begin{lemma}\label{lemma:19}
  Let $\tau$ be in $\KKbar$. The sum of the multiplicities of the
  isolated solutions of $\bB_\tau$ is at most equal to $c$.
\end{lemma}
\begin{proof}
  The sum in the lemma is also the sum of the multiplicities of the
  ideal $J_\tau$ at all $(\tau,\bx)$, for $\bx$ an isolated solution
  of $\bB_\tau$.  By Lemma~\ref{lemma:JJprime}, this is also the sum
  of the multiplicities of $J'_\tau$ at all $(\tau,\bx)$, for $\bx$ an
  isolated solution of $\bB_\tau$. We prove below that the sum of the
  multiplicities of $J'_\tau$ at all $(\tau,\bx)$, for $\bx$ such that
  $(\tau,\bx)$ cancels $J'_\tau$, is at most $c$; this will be enough
  to conclude (for any isolated solution $\bx$ of $\bB_\tau$,
  $(\tau,\bx)$ is a root of $J'_\tau$, though the converse may not be
  true). Remark that the latter sum is simply the dimension of
  $\KKbar[T,\bX]/J'_\tau$.
  
  Let $m_1,\dots,m_k$ be monomials that form a $\KKbar$-basis of
  $\KKbar[T,\bX]/J'_\tau$; since $T-\tau$ is in $J'_\tau$, these
  monomials can be assumed not to involve $T$.  We will prove that
  they are still $\KKbar(T)$-linearly independent in
  $\KKbar(T)[\bX]/\tilde{J'}$; this will imply that $k \le c$,
  and finish the proof.
  
  Suppose that there exists a linear combination $A_1 m_1 + \cdots +
  A_k m_k$ in $\tilde{J}'$, with all $A_i$'s in $\KKbar(T)$, not
  all of them zero. Thus, we have an equality $a_1/d_1\, m_1 + \cdots
  + a_k/d_k\, m_k = a/d$, with $a_1,\dots,a_k$ and
  $d,d_1,\dots,d_k$ in $\KKbar[T]$ and $a$ in the ideal
  $J'$. Clearing denominators, we obtain a relation of the form $e_1
  m_1 +\cdots+ e_k m_k \in J'$, with not all $e_i$'s zero. Let
  $(T-\tau)^u$ be the highest power of $T-\tau$ that divides all
  $e_i$'s (this is well-defined, since not all $e_i$'s vanish) so that
  we can rewrite the above as $(T-\tau)^u (f_1 m_1 +\cdots+ f_k
  m_k) \in J'$, with $f_i=e_i/(T-\tau)^u \in \KKbar[T]$ for all $i$.
  In particular, our definition of $e_i$ implies that the values
  $f_i(\tau)$ are not all zero.

  Recall that the ideal $J'$ has the form $J'=Q_1 \cap \cdots \cap
  Q_t$. For $i=1,\dots,t$, since $Q_i$ is primary, the membership
  equality $(T-\tau)^u (f_1 m_1 +\cdots +f_k m_k) \in J'$ implies
  that either $f_1 m_1 +\cdots +f_k m_k$ or some power
  $(T-\tau)^{uv}$, for some $v > 0$, is in $Q_i$. Since $Q_i$ does not
  contain non-zero polynomials in $\KKbar[T]$, $f_1 m_1 +\cdots+ f_k
  m_k$ belongs to all $Q_i$'s, that is, to $J'$. We can then
  evaluate this relation at $T=\tau$. We saw that the values
  $f_i(\tau)$ do not all vanish on the left, which is a contradiction
  with the independence of the monomials $m_1,\dots,m_k$ modulo
  $J'_\tau$.
\end{proof}

We now take $\tau$ in $\KKbar$ and we discuss the geometry of $V(J)$
near $\tau$; without loss of generality, we suppose that $\tau=0$.  We
already emphasized that the field $\SS$ is an algebraic closure of
$\KKbar(T)$; we thus let $\Phi_1,\dots,\Phi_{c'}$ be the points of
$V(\mathfrak{J}')$, with coordinates taken in $\SS$. In particular, we
see that $c' \le c$; we prove below that if $\assG(0)$ holds, we
actually have $c'=c$ (that is, that $\mathfrak{J}'$ is radical).

Any non-zero series $\varphi$ in $\SS$ admits a well-defined {\em
  valuation} $\nu(\varphi)$, which is the smallest exponent that
appears in its expansion support; we also set $\nu(0)=\infty$. The
valuation $\nu(\Phi)$, for a vector $\Phi=(\varphi_1,\dots,\varphi_s)$
with entries in $\SS$, is the minimum of the valuations of its
exponents. We say that $\Phi$ is {\em bounded} if it has non-negative
valuation; in this case, $\lim_0(\Phi)$ is defined as the vector
$(\lim_0(\varphi_1),\dots,\lim_0(\varphi_s))$, with
$\lim_0(\varphi_i)={\rm coeff}(\varphi_i,T^0)$ for all $i$.

Without loss of generality, we assume that
$\Phi_1,\dots,\Phi_\kappa$ are bounded, and
$\Phi_{\kappa+1},\dots,\Phi_{c'}$ are not, for some $\kappa$ in
$\{0,\dots,c'\}$, and we define $\varphi_1,\dots,\varphi_\kappa$ by
$\varphi_i=\lim_0(\Phi_i)\in\KKbar{}^n$ for
$i=1,\dots,\kappa$.

\begin{lemma}\label{lemma:Z1}
  The equality $V(J' +\langle T \rangle)=\{\varphi_i \mid i=1,\dots,\kappa\}$ holds.
\end{lemma}
\begin{proof}
  Let $(s_1,\dots,s_h)$ be generators of the ideal $J'$ in
  $\KKbar[T,\bX]$; they also generate $\mathfrak{J}'$ in
  $\KKbar(T)[\bX]$. Then, the polynomials
  $s_{0,i}=s_i(0,\bX) \in \KKbar[\bX]$, for $i=1,\dots,h$, are such
  that
  $J'+\langle T\rangle = \langle T,s_{0,1},\dots,s_{0,h} \rangle$.
  Consider $i \le \kappa$, and the corresponding vector of series
  $\Phi_i$. We know that for $j=1,\dots,h$, we have $s_j(\Phi_i)=0$.
  Since all elements involved have non-negative valuation, we can take
  the coefficient of degree $0$ in $T$ in this equality and deduce
  $s_{0,j}(\varphi_i)=0$, as claimed. Hence, each $\varphi_i$, for
  $i \le \kappa$, is in $V(J' + \langle T \rangle)$.

  Conversely, take indeterminates $T_1,\dots,T_n$, and let $\LL$ be
  the algebraic closure of the field $\KKbar(T_1,\dots,T_n)$; let
  ${\cal C}\subset{\LL}{}^{n+1}$ be the zero-set of the ideal $J'\cdot
  \LL[T,\bX]$ and consider the projection ${\cal C} \to {\LL}{}^2$
  defined by $(\tau,x_1,\dots,x_n)\mapsto (\tau,T_1 x_1 + \cdots + T_n
  x_n)$. The Zariski closure ${\cal S}$ of the image of this mapping
  is a hypersurface, that is, a plane curve.  Since the ideal $J'$ is generated by polynomials
  with coefficients in $\KKbar$, one deduces that ${\cal S}$ admits a squarefree
  defining equation in $\KKbar(T_1,\dots,T_n)[T,T_0]$.

  Consider such a polynomial, say $C$, and assume without loss of
  generality that $C$ belongs to 
  $\KKbar[T_1,\dots,T_n][T,T_0]$. Because ${\cal C}$ admits no irreducible
  component lying above $T=\tau$, for any $\tau$ in $\KKbar$, $C$
  admits no factor in $\KKbar[T]$; thus, $C(0,T_0)$ is non-zero.

  Let $\ell \in \KKbar[T_1,\dots,T_n,T]$ be the leading coefficient of
  $C$ with respect to $T_0$. Proposition~1 in~\cite{Schost03} proves
  that $C/\ell$, seen in $\KKbar(T_1,\dots,T_n,T)[T_0] \subset
  \LL(T)[T_0]$, is the minimal polynomial of $T_1 X_1 + \cdots +
  T_n X_n$ in $\LL(T)[\bX]/\sqrt{J'}\cdot \LL(T)[\bX]$. The latter ideal
  is also the extension of $\sqrt{\mathfrak{J}'}$ to $\LL(T)[\bX]$, 
  so $C/\ell$ factors as
  $$\frac C\ell = \prod_{1\le i \le c'}(T_0-T_1 \Phi_{i,1} - \cdots - T_n \Phi_{i,n})$$
  in $\LL' [T_0]$ where $\LL'$ is the generalized Power
  series ring in $T$ with coefficients in $\LL$.  This gives the
  equality
  $$C =\ell \prod_{1\le i \le  c'}(T_0-T_1 \Phi_{i,1} - \cdots - T_n
  \Phi_{i,n})$$ over $\SS[T_1,\dots,T_n,T_0]$. 

  Let us extend the valuation $\nu$ on $\SS$ to
  $\SS[T_1,\dots,T_n,T_0]$ in the direct manner, by setting
  $\nu(\sum_\alpha f_\alpha T_0^{\alpha_0} \cdots T_n^{\alpha_n}) =
  \min_\alpha \nu(f_\alpha)$. The fact that $C$ has no factor in
  $\KKbar[T]$ implies that $\nu(C)=0$. Using Gauss' Lemma, we see that
  the valuation of the right-hand side is $\nu(\ell) + \sum_{\kappa <
    i \le c}\mu_i$, with $\mu_i= \nu(\Phi_i)$ for all $i$; note that
  $\mu_i < 0$ for $i > \kappa$. Thus, we can rewrite
  $$C =\left ({T}^{-\nu(\ell)} \ell\right ) 
  \prod_{1 \le i \le \kappa}(T_0-T_1 \Phi_{i,1} - \cdots - T_n  \Phi_{i,n} )
  \prod_{\kappa < i \le c'} ({T}^{-\mu_i}T_0-{T}^{-\mu_i}T_1 \Phi_{i,1} - \cdots - {T}^{-\mu_i}T_n  \Phi_{i,n} ),$$
  where all terms appearing above have non-negative valuation.
  As a result, we can take the coefficient of ${T}^0$ term-wise,
  and obtain
  $$C(0,T_0) = s \prod_{1 \le i \le \kappa}(T_0-T_1 \varphi_{i,1} -
  \cdots - T_n \varphi_{i,n} ),$$ where $s$ is in $\KKbar[T_1,\dots,T_n]$;
  note that $s \ne 0$, since $C(0,T_0)$ is non-zero.
 By construction of $C$, for any
  $\bx=(x_1,\dots,x_n)$ in $V(J'+\langle T \rangle)$, $T_1 x_1 + \cdots + T_n x_n$
  cancels $C(0,T_0)$, so $\bx$ must be one of
  $\varphi_1,\dots,\varphi_{\kappa}$.
\end{proof}

To conclude the proof of Proposition~\ref{prop:degree_fiber}, 
we now assume that property $\assG(0)$ holds.

\begin{lemma}
   $\Phi_1,\dots,\Phi_{c'}$ are bounded; equivalently, $\kappa=c'$.
\end{lemma}
\begin{proof}
  We want to prove that $\Phi_1,\dots,\Phi_{c'}$ are bounded. Without
  loss of generality, one can assume that they are all non-zero (a
  zero vector is bounded).

  For $i=1,\dots,c'$, write
  $\Phi_i=1/T^{e_i} (\Psi_{i,1},\dots,\Psi_{i,n})$, for a vector
  $(\Psi_{i,1},\dots,\Psi_{i,n})$ of generalized power  series of valuation zero,
  that is, such that all $\Psi_{i,j}$ are bounded and
  $(\psi_{i,1},\dots,\psi_{i,n})=\lim_0(\Psi_{i,1},\dots,\Psi_{i,n})$
  is non-zero. Hence, $e_i=-\nu(\Phi_i)$, and we have to prove that
  $e_i \le 0$.  By way of contradiction, we assume that $e_i > 0$.

  The series $\Phi_i$ cancels $b_1,\dots,b_m$. For $k=1,\dots,m$, let
  $b_k^H \in \KKbar[T][X_0,\bX]$ be the homogenization of $b_k$ with
  respect to $\bX$. From the equality
  $b_k^H(T^{e_i},\Psi_{i,1},\dots,\Psi_{i,n})= T^{e_i}b_k(\Phi_i)$, we
  deduce that $b_k^H(T^{e_i},\Psi_{i,1},\dots,\Psi_{i,n})=0$ for all
  $k$. We can write $b_k = b_{0,k} + T \tilde b_k$, for some
  polynomial $\tilde b_k$ in $\KKbar[T,\bX]$, and $\assG_1(0)$ implies
  that $\deg_\bX(\tilde b_k) \le \deg_\bX(b_{0,k})$. As a result, the
  homogenizations (with respect to $\bX$) of $b_{k},b_{0,k}$ and
  $\tilde b_k$ satisfy a relation of the form
  $b^H_k = b_{0,k}^H + X_0^{\delta_k} T \tilde b^H_k$, for some
  $\delta_k \ge 0$. This implies the equality
  $$b_{0,k}^H(T^{e_i},\Psi_{i,1},\dots,\Psi_{i,n}) + T^{\delta_k
    e_i+1}\tilde b_k^H(T^{e_i},\Psi_{i,1},\dots,\Psi_{i,n})=0.$$
  The second term has positive valuation, so that
  $b_{0,k}^H(T^{e_i},\Psi_{i,1},\dots,\Psi_{i,n})$ has positive
  valuation as well. Taking the coefficient of $T^0$, this means 
  that $b_{0,k}^H(0,\psi_{i,1},\dots,\psi_{i,n})=0$ (since $e_i > 0$), which implies 
  that $(\psi_{i,1},\dots,\psi_{i,n})=(0,\dots,0)$, in view of $\assG_2(0)$.
  This however contradicts the definition of $(\psi_{i,1},\dots,\psi_{i,n})$.
\end{proof}

\begin{lemma}\label{lemma:Jprimerad}
  The ideal $\frak{J}'$ is radical; equivalently, $c'=c$.
\end{lemma}
\begin{proof}
  We know that $\frak{J}'$ has dimension zero
  (Lemma~\ref{lemma:dimJprime}), so it is enough to prove that for
  $i=1,\dots,c'$, the localization of
  $\SS[\bX]/\frak{J}'$ at the
  maximal ideal $\mathfrak{m}_{\Phi_i}$ is a field, or equi\-valently
  that the localization of
  $\SS[\bX]/\frak{J}$ at
  $\mathfrak{m}_{\Phi_i}$ is a field.  Recall that
  $\SS$ is algebraically closed, hence a perfect field. By the Jacobian
  criterion~\cite[Theorem~16.19.b]{Eisenbud95}, this is the case if
  and only if the Jacobian matrix of $\bB$ with respect to $\bX$ has
  full rank $n$ at $\Phi_i$. We know that $\varphi_i=\lim_0(\Phi_i)$
  is a root of $\bB_0$ (Lemma~\ref{lemma:Z1}), and the Jacobian
  criterion conversely implies that since the ideal
  $\langle \bB_0 \rangle$ is radical (by assumption $\assG_3(0)$) and
  zero-dimensional (by assumption $\assG_2(0)$), the Jacobian matrix
  of $\bB_0(\bX)=\bB(0,\bX)$ has full rank $n$ at $\varphi_i$. Since
  this matrix is the limit at zero of the Jacobian matrix of $\bB$
  with respect to $\bX$, taken at $\Phi_i$, the latter must have full
  rank $n$, and our claim that $\frak{J}'$ is radical is proved.
\end{proof}

To finish the proof of Proposition~\ref{prop:degree_fiber}, we have to
establish that $V(\bB_0)$ consists of exactly $c$ solutions. First,
since $V(\bB_0)$ is finite, Lemma~\ref{lemma:vPi} implies that $\bx$
is in $V(\bB_0)$ if and only if $(0,\bx)$ is in $V(J' + \langle
T\rangle)$. Next, remark that the two previous lemma taken together
imply that $c=\kappa$; thus, in view of Lemma~\ref{lemma:Z1}, to
conclude, it is enough to prove that for $i,i'$ in $\{1,\dots,c\}$,
with $i \ne i'$, we have $\varphi_i \ne \varphi_{i'}$.

Suppose to the contrary that $\varphi_i = \varphi_{i'}$. We know that
the Jacobian matrix of $\bB_0$ has full rank $n$ at $\varphi_i$; up to
reindexing, we assume that rows $1,\dots,n$ correspond to a maximal
non-zero minor. Let $\bB'=(b_1,\dots,b_n)$.

Let $z=\nu(\Phi_i-\Phi_{i'})$; since $\varphi_i = \varphi_{i'}$, we
have $z > 0$; it is finite else we would have $\Phi_i=\Phi_{i'}$ which
contradicts $i\neq i'$. We can thus write $\Phi_i=f + T^z \delta_i$
and $\Phi_{i'}=f + T^z \delta_{i'}$, for some vectors of bounded
series $f, \delta_i, \delta_{i'}$ such that all terms in $f$ have
valuation less than~$z$; in addition,
$\lim_0(\delta_i) \ne \lim_0(\delta_{i'})$. Write the Taylor expansion
of $\bB'$ at $f$ as
$$\bB'(\Phi_i) = \bB'(f) + \jac_f(\bB',\bX) T^z \delta_i + T^{2z} r_i =0$$
and
$$\bB'(\Phi_{i'}) = \bB'(f) + \jac_f(\bB',\bX) T^z \delta_{i'} + T^{2z}
r_{i'} =0,$$ for some vectors of bounded series $r_i,r_{i'}$.  By
subtraction and division by $T^z$, we obtain
$\jac_f(\bB',\bX) (\delta_i-\delta_{i'}) = T^z r$, for some vector of
bounded series $r$.  Since $\jac_f(\bB',\bX)$ is invertible, this
further gives $\delta_i-\delta_{i'} = T^z r'$, where again $r'$ is a
vector of bounded series.  However, by construction the left-hand side
has valuation zero, while the right-hand side has positive valuation
(since $z > 0$). Hence, we derived a contradiction to our assumption
that $\varphi_i = \varphi_{i'}$. The proof of
Proposition~\ref{prop:degree_fiber} is complete. (Although we do not
need it now, the linearization used above also implies that all
$\Phi_i$ are actually power series.)

We end this section with the proof that $e\geq c$. 

\begin{lemma}\label{lemma:e-geq-c}
Under the above notations and assumptions, the inequality $e\geq c$ holds.   
\end{lemma}

\begin{proof}
  By definition of the integer $e$ given in ${\assD}_2$, and of the
  ideal $J'$, $e$ is greater than or equal to the degree of $V(J')$,
  which is an algebraic curve.

  The degree of this curve is greater than or equal to the cardinality
  of any fiber $V(J'_\tau)$; in particular, we have
  \[
  \sharp V(J'_0) \leq \deg(V(J')) \leq e.
  \]

  Besides, Proposition~\ref{prop:degree_fiber} establishes that the
  number of isolated points of $V(\bB_0)$ equals $c$ (because the
  ideal generated by $\langle \bB, T\rangle$ is radical,
  multiplicities are equal to $1$). By Lemma~\ref{lemma:vPi}, all
  these points lie in $V(J'_0)$ which allows us to deduce $c\leq e$.
\end{proof}

%%%%%%%%%%%%%%%%%%%%%%%%%%%%%%%%%%%%%%%%%%%%%%%%%%%%%%%%%%%%

\subsection{Proofs of Propositions~\ref{prop:compute_isolated}
  and~\ref{prop:compute_regular}} \label{proof:prop2}

Let $\scrR_0 =((w_{0},v_{0,1},\dots,v_{0,n}),\lambda)$ be a
zero-dimensional parametrization of $V(\bB_{0})$ obtained by means of
assumption $\assD_1$, with $q_0$ and all $v_{0,j}$ in $\KK[Y]$. Note
that the degree of $w_0$ is the integer $c$.

\paragraph{Decomposing $\scrR_0$.}
We start by decomposing $\scrR_0$ into finitely many zero-dimensional
parametrizations
$\scrR_{0,j}=((w_{0,j},v_{0,j,1},\dots,v_{0,j,n}),\lambda)_{1\le j\le
  t}$, all with coefficients in $\KK$, such that for $j$ in
$\{1,\dots,t\}$, there exist $\bi_j=(i_{j,1},\dots,i_{j,n})$ such that
the Jacobian matrix of $(b_{0,i})_{i \in \bi_j}$ has full rank $n$ at
$\bx$, for all $\bx$ in $Z(\scrR_{0,j})$.

If $w_0$ were irreducible, we would simply evaluate the Jacobian
matrix of $\bB_0$ at the point $(v_{0,1}/w_0',\dots,v_{0,n}/w_0')$,
which has coordinates in the field $\LL=\KK[Y]/\langle w_0 \rangle$,
and find a non-zero minor of size $n$ in this matrix. It takes
$O(n \sigma)$ operations in $\LL$ to compute this Jacobian matrix, and
$O(mn^2)$ operations in $\LL$ to find an invertible minor, e.g. using
Gaussian elimination. The total time, under the assumption that $w_0$
is irreducible, is thus $O(mn^2 + n\sigma)$ operations in $\LL$, that is,
$\softO(c (mn^2 + n\sigma) )$ operations in $\KK$.

When $w_0$ is not irreducible, $\LL=\KK[Y]/\langle w_0 \rangle$ is a
product of fields. We can still apply the same process as in the
irreducible case; if the algorithm goes through, we have obtained our
answer. In general, one workaround would be to factor $w_0$, but we do
not want our runtime to depend on the cost of factoring polynomials
(else our analysis would depend on the bit size of the data when
$\KK=\mathbb{Q}$). Hence, we will use {\em dynamic evaluation
  techniques}, as in~\cite{D5}. Indeed, the only issue that may arise
is that we attempt to invert a zero-divisor. If this is the case, it
means we have found a non-trivial factor $r_0$ of $w_0$: we can then
replace $\scrR_0$ by two new zero-dimensional parametrizations,
$\scrR'_0=((r_0,(v_{0,1}/s_0) \bmod r_0,\dots,(v_{0,n}/s_0)\bmod
r_0),\lambda)$
and
$\scrR''_0=((s_0,(v_{0,1}/r_0) \bmod s_0,\dots,(v_{0,n}/ r_0)\bmod
s_0),\lambda)$,
with $s_0=w_0/r_0$, that define a partition of $Z(\scrR_0)$ into the
subsets $Z(\scrR'_0)$ and $Z(\scrR''_0)$ where $r_0$ vanishes, resp.\
is non-zero.

We can then start over again, from $\scrR'_0$ and $\scrR''_0$
independently. Overall, in the worst case, this splitting process
induces a extra factor $O(c)$ in the runtime compared to the case
where $w_0$ is irreducible, for a total of $\softO( c^2 (mn^2 + n\sigma) )$
operations in $\KK$.

\paragraph{Lifting power series and rational reconstruction.}
For $j=1,\dots,t$, we can then apply Newton iteration to the system
$(b_i)_{i \in \bi_j}$ to lift
$\scrR_{0,j}=((w_{0,j},v_{0,j,1},\dots,v_{0,j,n}),\lambda)$ into a
zero-dimensional parametrization
$\scrR_{j}=((w_{j},v_{j,1},\dots,v_{j,n}),\lambda)$ with coefficients
in $\KK[[T]]/\langle T^{2e}\rangle$, for $e$ as in ${\assD}_2$. 

As explained in~\cite[Section~2.2]{SaSc16}, using the algorithm
of~\cite{GiLeSa01}, this can be done using $\softO(c\,e (\sigma+n^2)n)$
operations in $\KK$.  Using the Chinese Remainder Theorem, we can
combine all $\scrR_{j}$ into a single zero-dimensional parametrization
$\scrR$ with coefficients in $\KK[[T]]/\langle T^{2e}\rangle$, since
for $j\ne j'$, $w_{0,j}$ and $w_{0,j'}$ generate the unit ideal in
$\KK[[T]]/\langle T^{2e}\rangle$; this takes time 
$\softO(c\,e\,n)$.

Using the notation of the previous subsection, the zeros of $\scrR$ in
$\KKbar[[T]]/\langle T^{2e}\rangle$ are the truncations of the power
series roots $\Phi_1,\dots,\Phi_c$ of $\mathfrak{J}'$. Since $V(J')$ has
degree at most $e$, knowing $\scrR$ at precision $2e$ allows us to
reconstruct a zero-dimensional parametrization $\scrS$ with
coefficients in $\KK(\tau)$ such that $Z(\scrS)=V(\mathfrak{J}')$,
with all coefficients having numerator and denominator of degree at
most $e$~\cite[Theorem~1]{Schost03}.  This is done by applying
rational function reconstruction to all coefficients of $\scrR$, as
in~\cite{Schost03}, and takes time $\softO(c\,e\,n)$.

All in all, the total cost of this step is $\softO(c\,e\,n (\sigma+n^2))$.

\paragraph{A finite set containing the isolated points of $V(\bC)$.}
As we did in the previous subsection for $T=0$, we let
$\Phi'_1,\dots,\Phi'_c$ be the roots of $\mathfrak{J}'$ in the field
of generalized power series in $T'$ with coefficients in $\KKbar$ at
$T=1$, with $T'=T-1$. Without loss of generality, we assume that
$\Phi'_1,\dots,\Phi'_{\kappa'}$ are bounded, and
$\Phi'_{\kappa'+1},\dots,\Phi'_c$ are not, for some $\kappa'$ in
$\{0,\dots,c\}$, and we define $\varphi'_1,\dots,\varphi'_{\kappa'}$
by $\varphi'_i=\lim_0(\Phi'_i)\in\KKbar{}^n$ for $i=1,\dots,\kappa'$.
By Lemma~\ref{lemma:Z1},
$V(J' + \langle T-1\rangle) = \{ \varphi'_i \mid i=1,\dots,\kappa'\}$.

We can now specify our requirements on the linear form $\lambda$.
Following~\cite{RRS} and~\cite{SaSc16}, we ask that $\lambda$ is a
{\em well-separating element}, that is:
\begin{enumerate}
\item $\lambda$ is separating for $V(\mathfrak{J}')=\{\Phi'_1,\dots,\Phi'_c\}$;
\item $\lambda$ is separating for $V(J' + \langle T-1\rangle) = \{ \varphi'_1,\dots,\varphi'_{\kappa'}\}$.
\item $\nu(\lambda(\Phi_i)) = \mu_i$ for all $i=1,\dots,c$, where $\nu$ denotes
 the $T'$-adic valuation.
\end{enumerate}
Applying Lemma~14 in~\cite[Section 3]{SaSc16}, these conditions are
satisfied for a generic choice of $\lambda$. When this is the case,
Lemma~4.4 in~\cite{RRS} shows how to recover a zero-dimensional
parametrization $\scrR_1=((w_1,v_{1,1},\dots,v_{1,n}),\lambda)$ with
coefficients in $\KK$ for the limit set
$V(J' + \langle T-1\rangle) =\{\varphi'_i \mid i=1,\dots,{\kappa'}\}$
starting from the previously computed rational parametrization
$\scrS$, in time $\softO(c\,e\,n)$.

When the chosen form is not generic enough, the algorithm may fail, or
output a parametrization of a subset of the zero-dimensional set we
aim to compute. We refer to \cite[Remark 14]{SaSc16} for a discussion
on probabilistic aspects.

\paragraph{Cleaning.}
Finally, summing all the previous costs, one performs
$$
\softO(c^2 (mn^2+n\sigma) + c\, e\,n (\sigma+n^2) )
$$
operations in $\KK$ for the first three steps (decomposition of
$\scrR_0$, lifting and rational reconstruction and getting a finite
set containing the isolated points of $V(\bC)$).

Let us first show how to prove Proposition~\ref{prop:compute_isolated}.
Lemma~\ref{lemma:vPi} implies that for any isolated solution $\bx$ of
$\bC$, $(1,\bx)$ is in $V(J' + \langle T-1\rangle)$, so in a second
time, we discard from $V(J' + \langle T-1\rangle)$ those points that
do not correspond to isolated points of $V(\bC)$. All such points
belong to a positive-dimensional component of $V(\bC)$. Hence, we can
use the algorithm of Section~\ref{sec:isolated}. By
Proposition~\ref{prop:degree_fiber}, we can take $c$ as an upper bound
on the multiplicity of isolated solutions of $\bC$.

Using the same dynamic evaluation techniques as in the first paragraph
above, we can use the algorithm of Section~\ref{sec:isolated} as if
$Z(\scrR_1)$ were an irreducible variety, with an overhead $\softO(c)$
to account for the cost of operations in $\KK[Y]/\langle w_1 \rangle$.
Since the number of splittings is bounded by $c$ also, the total
overhead is $\softO(c^2)$.  The runtime deduced from
Proposition~\ref{prop:testisolated} is then
$$\softO(c^6 n^4  +  c^5 m n^2  +c^6 n \sigma )$$ operations in~$\KK$. Adding all
costs seen so far, we prove
Proposition~\ref{prop:compute_isolated}. The resulting algorithm,
which we simply name $\mathsf{Homotopy}$, is described hereafter.

\begin{algorithm}
\caption{$\mathsf{Homotopy}(\Gamma,\scrR)$}
{\bf Input}: a straight-line program $\Gamma$ of length $\sigma$ that computes $\bB \in \KK[T,\bX]^m$\\
\textcolor{white}{{\bf Input}:} a zero-dimensional parametrization $\scrR$ of the system $\bB_0$\\
{\bf Output}: a zero-dimensional parametrization of the isolated points of $V(\bC)$, with $\bC=\bB_1$
\begin{enumerate}
  \setlength\itemsep{0em}
\item decompose $\scrR_0$ into $(\scrR_{0,j})_{1 \leq j \leq t}$\\
$\text{\sf{cost:~}} \softO(c^2 (mn^2 + n\sigma) )$
\item lift $(\scrR_{0,j})_{1 \leq j \leq t}$ to $(\scrR_j)_{1\leq j \leq t}$ with 
  coefficients in $\KK[[T]]/\langle T^{2e}\rangle$\\
$\text{\sf{cost:~}} \softO(c\,e (\sigma+n^2)n)$
\item combine $(\scrR_j)_{1 \leq j \leq t}$ into  $\scrR$ with coefficients in $\KK[[T]]/\langle T^{2e}\rangle$\\
$\text{\sf{cost:~}} \softO(c\,e\,n)$
\item compute a zero-dimensional parametrization $\scrS$ with coefficients in $\KK(T)$ from $\scrR$\\
$\text{\sf{cost:~}} \softO(c\,e\,n)$
\item deduce a zero-dimensional parametrization $\scrR_1$ with coefficients in $\KK$ from $\scrS$\\
$\text{\sf{cost:~}} \softO(c\,e\,n)$
\item\label{step:homot:final} remove from $Z(\scrR_1)$ points that are not isolated in $V(\bC)$ \\
  $\text{\sf{cost:~}} \softO(c^6 n^4  + c^5 m n^2   + c^6 n \sigma )$
\end{enumerate}
\label{DetSys}
\end{algorithm}

The only difference to prove Proposition~\ref{prop:compute_regular} is
that we now need to discard from $V(J' + \langle T-1\rangle)$ those
points at which the Jacobian matrix associated to $\bC$ is not full
rank. Doing that is easier than discarding those points which are not
isolated. It suffices to construct a straight-line program evaluating
that Jacobian matrix; this yields a straight-line program of length
$\sigma'\in O(n \, \sigma)$. Next, one evaluates this matrix modulo
$w_1$, as done previously when we were decomposing $\scrR_0$, and use
Gaussian elimination modulo $w_1$ to identify divisors of $w_1$ that
need to be removed. The overall cost is similar to that of decomposing
$\scrR_0$, that is, $\softO(c^2(mn^2+n\sigma))$ operations in $\KK$.
The final cleaning step is done using Algorithm $\mathsf{Clean}$ of
\cite{GiLeSa01} whose cost is dominated by the previous computations.

All in all, the total cost is
$$
\softO(c^2(mn^2+n\sigma) + c\, e\,n (\sigma+n^2)  )
$$
operations in $\KK$. Taking into account the inequality
$e\geq c$ (Lemma~\ref{lemma:e-geq-c}) this simplifies as 
$$
\softO(c^2 mn^2  + c\, e\,n (\sigma+n^2) ),
$$
which ends the proof of Proposition~\ref{prop:compute_regular}. In the
sequel, the resulting algorithm is called
$\mathsf{Homotopy\_simple}$. It differs from Algorithm
$\mathsf{Homotopy}$ at Step \ref{step:homot:final} where the cleaning step we
just described replaces the one of $\mathsf{Homotopy}$.

%%%%%%%%%%%%%%%%%%%%%%%%%%%%%%%%%%%%%%%%%%%%%%%%%%%%%%%%%%%%
%%%%%%%%%%%%%%%%%%%%%%%%%%%%%%%%%%%%%%%%%%%%%%%%%%%%%%%%%%%%
%%%%%%%%%%%%%%%%%%%%%%%%%%%%%%%%%%%%%%%%%%%%%%%%%%%%%%%%%%%%

\section{Properties of determinantal ideals}\label{sec:check}

The following sections will show how to apply the algorithms of the
previous section to Problems~\eqref{problem2} and~\eqref{problem3}, by
applying Proposition~\ref{prop:compute_isolated}
(resp. Proposition~\ref{prop:compute_regular}) to suitable
deformations of our input systems. This proposition requires several
assumptions to hold: some (noted $\assA_1$ and $\assA_2$; see
Section~\ref{sec:homotopy}) are related to the deformed system as a
whole, while the others ($\assG_1$ to $\assG_3$) involve properties at
the starting point of the homotopy ($T=0$). In this section, we prove
that a large variety of systems satisfy $\assA_1$ and $\assA_2$.

Let $T$ and $\bX=(X_1,\dots,X_n)$ be variables, let $J$ be an ideal in
$\KKbar[T,\bX]$, and let us recall properties $\assA_1$ and $\assA_2$: 
\begin{description}[leftmargin=*]
\item[$\assA_1.$] Any irreducible component of $V(J) \subset
  \KKbar{}^{n+1}$ has dimension at least one.
\item[$\assA_2.$] For any maximal ideal $\m \subset\KKbar[T,\bX]$,
  if the localization $J_\m \subset \KKbar[T,\bX]_\m$ has height $n$,
  then it is unmixed (that is, all associated primes have height $n$).
\end{description}

We pointed out in the previous section that when $J$ is generated by
$n$ polynomials, the fact that these properties hold is well-known. To
study the case of maximal minors of a polynomial matrix, we will use
the following results, taken from~\cite[Section~6]{EN62}. Let $R$ be a
Cohen-Macaulay ring and let $I$ be the ideal generated by all
$p$-minors of a $p\times q$ matrix $\mF \in R^{p\times q}$, with
$p \le q$. Then:
\begin{itemize}
\item if $I \ne R$, then the height of $I$ is at most $q-p+1$;
\item if $I$ has height $q-p+1$, then $I$ is unmixed (all associated
  primes have height $q-p+1$).
\end{itemize}

Let then $G=(g_1,\dots,g_s)$ be polynomials in $\KKbar[T,\bX]$, with $s
\le n$, and let $\mF$ be a polynomial matrix in $\KKbar[T,\bX]^{p \times
  q}$, with $p \le q$. We define $J = I_p(\mF) + \langle
g_1,\dots,g_s\rangle,$ that is $J$ is the ideal in $\KKbar[T,\bX]$
generated by all $p$-minors of $\mF$, together with the polynomials
$G$.

\begin{proposition}\label{prop:KH1H2}
  If $n=q-p+s+1$, the ideal $J$ satisfies $\assA_1$ and $\assA_2$.
\end{proposition}

The proof occupies the rest of this section.  Let $\bB=M_p(\mF)$, the
set of all $p$-minors of~$\mF$, and let $V_1,\dots,V_s$ be the
$\KKbar$-irreducible components of $V(J)\subset \KKbar{}^n$.  We prove
in the next paragraph that $\dim(V_i) \ge (n+1) -(q-p+1)$ holds for
all $i$. Of course, we can assume that $V(J)\ne \emptyset$, so that $J
\ne \KKbar[T,\bX]$, otherwise the proposition itself would be vacuously true.

First, remark that for a point $\bx$ in $V(J)\subset \KKbar{}^{n+1}$,
and writing $\m \subset \KKbar[T,\bX]$ for the maximal ideal at $\bx$,
the height of $J_\m$ in $\KKbar[T,\bX]_\m$ is equal to $(n+1)-\max\{
\dim(V_i) \mid 1 \le i \le s, \bx \in V_i\}$. For $i=1,\dots,s$, let
then $\bx_i$ be a point in $V_i$ that does not belong to any other
$V_{i'}$, $i' \ne i$, and let $\m_i$ be the corresponding maximal
ideal; then, the previous equality becomes ${\rm
  height}(J_{\m_i})=(n+1)-\dim(V_i)$.  Applying the first item
mentioned above in $\KKbar[T,\bX]_{\m_i}$ (which is Cohen-Macaulay),
we deduce that $(n+1)-\dim(V_i) \le q-p+1$, that is, $\dim(V_i) \ge
(n+1) -(q-p+1)$.

Notice that we can rewrite $(n+1)-(q-p+1)$ as $s+1$.  Since $G$
consists of $s$ polynomials, all irreducible components of $V(J)$ must
have dimension at least $1$, by Krull's theorem;  property
$\assA_1$ follows.

We next prove $\assA_2$. Let $J_\m=Q_1 \cap \cdots \cap Q_t$ be an
irredundant primary decomposition of $J_\m$ in $\KKbar[T,\bX]_\m$, and
let $P_1,\dots,P_t$ be the corresponding primes; we assume that the
height of $J_\m$ is $n$, and our goal is to prove that all $P_i$'s
have height~$n$.

Of course, we can restrict to an ideal $\m$ containing $J$; $\m$ is
then the maximal ideal at a point $\x \in \KKbar{}^{n+1}$ that belongs
to $V(J)$. The height of the localization
$J_\m \subset \KKbar[T,\bX]_\m$ can be rewritten as
$(n+1)-\dim(V_\x)$, where $V_\x$ is the union of the irreducible
components of $V(J)$ passing through $\x$. Our assumption in $\assA_2$
is that the height of $J_\m$ is $n$, that is, that
$\dim(V_\x)=1$. Thus, every irreducible component of $V(J)$ containing
$\x$ has dimension~$1$.

Let $W$ be an irreducible component of $V(\bB)$ containing $\x$.  We
claim that $\dim(W)=s+1$. Indeed, we mentioned in the first paragraph
that $\dim(W) \ge s+1$. If $\dim(W) > s+1$, then by Krull's theorem,
every irreducible component of $W \cap V(G)$ has dimension greater
than $1$; since $W \cap V(G)$ is a subset of $V(J)$ and contains $\x$,
we have reached a contradiction. Now, the fact that $\dim(W)=s+1$ for
any irreducible component of $V(\bB)$ containing $\x$ means that
$\langle \bB \rangle_\m$ has height $n-s=q-p+1$.  As a
result,~\cite[Theorem~18.18]{Eisenbud95} shows that
$\KKbar[T,\bX]_\m/\langle \bB \rangle_\m$ is Cohen-Macaulay.

For an ideal $I \subset \KKbar[T,\bX]_\m$, we denote by $\bar I$ its
image modulo $\langle \bB \rangle_\m$.  By the remarks
following~\cite[Theorem~IV.5.9]{ZaSa58},
$\bar Q_1 \cap \cdots \cap \bar Q_t$ is an irredundant primary
decomposition of $\bar J_\m$ in
$\KK[T,\bX]_\m/\langle \bB \rangle_\m$, with associated primes
$\bar P_1,\dots,\bar P_t$. In addition, if we let $P_1,\dots,P_u$ be
the minimal primes of $J_\m$, for some $s \le t$,
$\bar P_1,\dots,\bar P_u$ are the minimal primes of $\bar J_\m$.

Our assumption says that $P_1,\dots,P_u$ have height $n$. Because
$\KKbar[T,\bX]_\m/\langle \bB \rangle_\m$ is local and Cohen-Macaulay, for any
$i \le t$, we have 
$$\dim(\KKbar[T,\bX]_\m/\langle \bB \rangle_\m)=\dim((\KKbar[T,\bX]_\m/\langle \bB \rangle_\m) / \bar P_i) + {\rm height}(\bar P_i)$$
by~\cite[Theorem~17.4(i)]{Matsumura86}.
The factor ring $(\KKbar[T,\bX]_\m/\langle \bB \rangle_\m) / \bar P_i$ is simply
$\KKbar[T,\bX]_\m/P_i$, so this can be rewritten as
$$s+1 = \dim(\KKbar[T,\bX]_\m/P_i) + {\rm height}(\bar P_i).$$ For $i\le
u$, we have $\dim(\KKbar[T,\bX]_\m/P_i)=1$, so that ${\rm height}(\bar
P_i)=s$; for $i > u$, the height of $\bar P_i$ is necessarily
$s+1$. Because $\bar P_1,\dots,\bar P_u$ are the minimal primes of
$\bar J_\m$, the height of $\bar J_\m$ is thus $s$ as well.

The ideal $\bar J_\m$ is generated in $\KKbar[T,\bX]_\m/\langle \bB
\rangle_\m$ by $G=(g_1,\dots,g_s)$. Since $\KKbar[T,\bX]_\m/\langle
\bB \rangle_\m$ is Cohen-Macaulay, $\bar J_\m$ is unmixed, that is,
$u=t$.  As a result, $Q_1 \cap \cdots \cap Q_u$ is an irredundant
primary decomposition of $J_\m$, and $J_\m$ is unmixed.

%%%%%%%%%%%%%%%%%%%%%%%%%%%%%%%%%%%%%%%%%%%%%%%%%%%%%%%%%%%%
%%%%%%%%%%%%%%%%%%%%%%%%%%%%%%%%%%%%%%%%%%%%%%%%%%%%%%%%%%%%
%%%%%%%%%%%%%%%%%%%%%%%%%%%%%%%%%%%%%%%%%%%%%%%%%%%%%%%%%%%%

\section{The column-degree homotopy}\label{sec:columndegree}

We can now prove the first half of our results, dealing with the column
degree structure of our matrices. As input, we are given a matrix
$\mF =[f_{i,j}]\in \KK[X_1,\dots,X_n]^{p \times q}$ and polynomials
$G=(g_1,\dots,g_s)$ in $\KK[X_1,\dots,X_n]$, with $p \leq q$ and
$n = q-p+s+1$.  We want to compute the isolated points (or the simple
 points) of $\VpFG{p}{\mF}{G}$, with
$$\VpFG{p}{\mF}{G} = \{\bx \in \KKbar{}^n \mid  \mathrm{rank}(\mF({\bx})) < p
\text{~and~} g_1(\bx)=\cdots=g_s(\bx)=0\}.$$

In this section, we design an algorithm for these both tasks whose
cost depends on the column degrees
$\delta_1=\cdeg(\mF,1),\dots,\delta_q=\cdeg(\mF,q)$; note in
particular that with this notation, $\deg(f_{i,j}) \leq \delta_j$
holds for all $i,j$.  We will also write
$\gamma_1=\deg(g_1),\dots,\gamma_s=\deg(g_s)$.

We point out that (in the case where there are no polynomials $G$),
the construction used in this section was already in the appendix
of~\cite{NieRan09}, where it was used to bound the number of solutions
of determinantal systems (as we mentioned in the introduction).

Recall that for $k\geq 0$, $E_k(\delta_1,\dots,\delta_q)$ denotes the
elementary symmetric polynomial of degree $k$ in
$(\delta_1, \ldots, \delta_q)$. 

\begin{proposition}\label{prop:coldeg}
  Suppose that the matrix $\mF \in \KK[X_1,\dots,X_n]^{p \times q}$
  and the polynomials $G=(g_1,\dots,g_s)$ in $\KK[X_1,\dots,X_n]$ are
  given by a straight-line program of length $\sigma$. Then, the sum
  of the multiplicities of the isolated points of $\VpFG{p}{\mF}{G}$
  are at most
  $c=\gamma_1\cdots\gamma_sE_{n-s}(\delta_1, \ldots, \delta_q)$.

  Assume that all $\gamma_i$'s and $\delta_j$'s are at least equal to
  $1$, and let
  $e=(\gamma_1+1)\cdots(\gamma_s+1) E_{n-s}(\delta_1+1, \ldots,
  \delta_q+1)$, $\gamma = \max(\gamma_1, \ldots, \gamma_s)$ and
  $\delta = \max(\delta_1, \ldots, \delta_q)$.
  Then, there exists a randomized algorithm that computes these isolated
  points
  $$
  \softO\left ({q \choose p} c(e+c^5 )(\sigma +\gamma + q
    \delta)\right)
  $$
  operations in $\KK$.
\end{proposition}

The next proposition states a better complexity estimate when one only
computes simple points of $\VpFG{p}{\mF}{G}$.
\begin{proposition}\label{prop:coldeg_simple}
  Reusing the notations introduced above, 
  there exists a randomized algorithm that computes the simple points of $\VpFG{p}{\mF}{G}$ using
  $$ 
  \softO\left (   {q \choose p}c \,e (\sigma+\gamma+q\delta)) \right )
  $$ 
  operations in $\KK$.
\end{proposition}

These propositions establish the first half of
Theorems~\ref{theo:1},~\ref{theo:2} and~\ref{theo:3}.

\medskip

We use the algorithms of Section~\ref{sec:homotopy}. To match the
notation of that section, we let $\bC=(c_1,\dots,c_{s},\dots,c_m)$ be
polynomials defined as follows: $(c_1,\dots,c_{s})=(g_1,\dots,g_s)$,
and $(c_{s+1},\dots,c_{m})$ are the $p$-minors of $\mF$, so that
$m=s+{q \choose p}$. Thus, $\VpFG{p}{\mF}{G}$ is the zero-set
of~$\bC$.

Using the degrees $\gamma_1,\dots,\gamma_s$ and
$\delta_1,\dots,\delta_q$, we construct a polynomial matrix
$\mL \in \KK[\bX]^{p \times q}$, and polynomials $M=(m_1,\dots,m_s)$
in $\KK[\bX]$, to use as a starting point for the homotopy
algorithm. For any $1 \leq j \leq q$ and $1 \leq k \leq \delta_j$, let
us define
$$\lambda_{j,k} = \lambda_{j,k,0} + \sum_{\ell =
  1}^{n}\lambda_{j,k,\ell}X_\ell,$$ where all $\lambda_{j,k,\ell}$ are
random elements in $\KK$. Then, for $j=1,\dots,q$, we define
$$\lambda_j = \prod_{k=1}^{\delta_j}\lambda_{j,k},$$
and we let  $\mL$ be the matrix
\begin{align}\label{eqdef:col}
\mL = 
\left( \begin{matrix}
\lambda_1 & 2\lambda_2 & \cdots & q\lambda_{q}\\
\lambda_1 & 2^2\lambda_2 & \cdots & q^2\lambda_q\\
\vdots & \vdots &  & \vdots \\
\lambda_1 & 2^p\lambda_2 & \cdots & q^p\lambda_q
\end{matrix} \right) \in \KK[\bX]^{p\times q}.
\end{align}
For $i=1,\dots,s$ and $k=1,\dots,\gamma_i$, let us further define
$$\mu_{i,k} =  \mu_{i,k,0} + \sum_{\ell = 1}^{n}\mu_{i,k,\ell}X_\ell,$$ where
all $\mu_{i,k,\ell}$ are random elements in $\KK$; then, we let
$$a_i=\prod_{k=1}^{\gamma_i} \mu_{i,k}.$$ We can thus define the
system of equations $\bA=(a_1,\dots,a_s,\dots,a_m)$, with $a_i$ as
above for $i=1,\dots,s$, and where $(a_{s+1},\dots,a_{m})$ are the
$p$-minors of $\mL$ (taken in the same order as those in the system
$\bC$).

Let $T$ be a new variable and define the matrix
\mbox{$\mU=(1-T)\cdot \mL + T \cdot \mF \in \KK[T,\bX]^{p\times q}$}.
We let $\bB$ be the polynomials in
$\KK[T,\bX]$ given by $\bB=(b_1,\dots,b_s,\dots,b_m)$, where
\begin{itemize}
\item $b_i=(1-T) a_i + T g_i$ for $i=1,\dots,s$
\item $(b_{s+1},\dots,b_{m})$ are the $p$-minors of $\mU$, taken in
  the same order as those in $\bC$.
\end{itemize}
We can then define $J$ as the ideal generated by $\bB$ in
$\KKbar[T,\bX]$. Using the notation of Section~\ref{sec:homotopy}, we
see that $\bB_0=\bA$ and $\bB_1=\bC$. Having in mind to apply
Proposition~\ref{prop:compute_isolated}
(resp.\ Proposition~\ref{prop:compute_regular}) to compute the isolated
points (resp. simple points) of $V(\bC)=\VpFG{p}{\mF}{G}$, we
now verify that all required assumptions are satisfied.

\paragraph{Properties $\assA_1$ and $\assA_2$.}
These follow from Proposition~\ref{prop:KH1H2}.

\paragraph{Property $\assG_1(0)$.} We have to prove that for $i=1,\dots,m$,
$\deg_\bX(b_i)=\deg_\bX(a_i)$. 

For $i=1,\dots,s$, this amounts to proving that $\deg_\bX((1-T) a_i +
T g_i)=\deg_\bX(a_i)$. The latter is by construction equal to
$\gamma_i$. The former is at most $\gamma_i$ (since $b_i$ is the sum
of two polynomials of degree $\gamma_i$ in $\bX$), but since
evaluating $T$ at $0$ in $b_i$ gives us $g_i$, its degree in $\bX$
must be exactly $\gamma_i$.

To each index $i=s+1,\dots,m$ corresponds a sequence
$\bj_i=(j_{i,1},\dots,j_{i,p})$ such that $b_i$ and $a_i$ are the
minors built with columns indexed by $\bj_i$ in respectively
$\mU=(1-T)\cdot \mL + T \cdot \mF$ and $\mL$. In view of the shape of
$\mL$, the polynomial $a_i$ is equal to $c_i\lambda_{j_{i,1}}\cdots
\lambda_{j_{i,p}}$, with
$$c_i = \left | 
\begin{matrix}
j_{i,1} & j_{i,2} & \cdots & j_{i,p}\\
j_{i,1}^2 & j_{i,2}^2 & \cdots & j_{i,p}^2\\
\vdots & \vdots &  & \vdots \\
j_{i,1}^p & j_{i,2}^p & \cdots & j_{i,p}^p
\end{matrix}
\right |.$$
Because $\KK$ has characteristic zero,
 $c_i$ is a non-zero constant, so that $a_i$ has degree $\delta_{j_{i,1}} +
\cdots + \delta_{j_{i,p}}$.  Since the columns
$(j_{i,1},\dots,j_{i,p})$ of $U$ have respective degrees at most
$(\delta_{j_{i,1}},\dots,\delta_{j_{i,p}})$, $b_i$ has degree at most
$\delta_{j_{i,1}} + \cdots + \delta_{j_{i,p}}$. However, evaluating
$T$ at $0$ in $b_i$ gives us back the polynomial $a_i$, so $b_i$ must
have degree exactly $\delta_{j_{i,1}} + \cdots + \delta_{j_{i,p}}$.

\paragraph{Property $\assG_2(0)$.} We have to prove that the homogenization
of the system $\bA$ has no root at infinity. Thus, let $X_0$ be a new
variable, and let $\bA^H=(a_1^H,\dots,a_m^H)$ be the homogenization
of $\bA$. For $i=1,\dots,s$, we have
$$a_i^H=\prod_{k=1}^{\gamma_i} \mu^H_{i,k} \quad\text{with}\quad \mu^H_{i,k}=(\mu_{i,k,0}X_0 + \sum_{\ell = 1}^{n}\mu_{i,k,\ell}X_\ell),$$
whereas for $i=s+1,\dots,m$, 
$$a_i^H=c_i \lambda^H_{j_{i,1}}\ldots \lambda^H_{j_{i,p}}, \quad \text{~for~} \bj_i=(j_{i,1},\dots,j_{i,p}) \text{~as above},$$
where for $j=1,\dots,q$ we set 
$\lambda^H_j = \prod_{k=1}^{\delta_j}\lambda^H_{j,k}$,
with
$$\lambda^H_{j,k}=\lambda_{j,k,0}X_0 + \sum_{\ell = 1}^{n}\lambda_{j,k,\ell}X_\ell.$$
To prove  $\assG_2(0)$, we start by writing down all projective
solutions of this system (this will be of use below), before adding
the constraint $X_0=0$.

Since all $a_i^H$ are products of linear forms, we find the solutions
of $\bA^H$ by setting some of these linear forms to zero. In order to
cancel $a_1^H,\dots,a_s^H$, we choose indices $\bu=(u_1,\dots,u_s)$,
with $u_1\in\{1,\dots,\gamma_1\}$, \dots,
$u_s\in\{1,\dots,\gamma_s\}$, and we consider the equations 
$$\mu^H_{i,u_i}=0, \quad \text{~that is,~} \quad \mu_{i,u_i,0}X_0 + \sum_{\ell = 1}^{n}\mu_{i,u_i,\ell}X_\ell =0,$$ for $i=1,\dots,s$.
In what follows, we fix such an $\bu$.
Then, for a generic choice of coefficients $\mu_{i,k,\ell}$, these equations
are equivalent to
$$X_{n-s+1}=\Phi_{n-s+1,\bu}(X_0,\dots,X_{n-s}),\dots,X_{n}=\Phi_{n,\bu}(X_0,\dots,X_{n-s}),$$
for some homogeneous linear forms
$\Phi_{n-s+1,\bu},\dots,\Phi_{n,\bu}$.  After applying this
substitution, for all $j=1,\dots,q$, $\lambda^H_j$ can be rewritten as
$$\lambda^H_{j,\bu}=\prod_{k=1}^{\delta_j}\lambda^H_{j,k,\bu},$$
where 
$$\lambda^H_{j,k,\bu}=\lambda_{j,k,0}X_0 + \sum_{\ell =
  1}^{n-s}\lambda_{j,k,\ell}X_\ell + \sum_{\ell =
  n-s+1}^{n}\lambda_{j,k,\ell}
\Phi_{\ell,\bu}(X_0,\dots,X_{n-s}).$$ Then,
$\bx=(x_0,\dots,x_n)$ cancels $a^H_{s+1},\dots,a^H_m$ if and only if
$\bx'=(x_0,\dots,x_{n-s})$ cancels the product
$\lambda^H_{j_1,\bu}\cdots \lambda^H_{j_p,\bu},$ for any choice of $p$ columns
$\bj=(j_1,\dots,j_p)$.

\begin{lemma}
  For $\bx'$ in $\P^{n-s}(\KKbar)$, the products
  $\lambda^H_{j_1,\bu}(\bx')\cdots \lambda^H_{j_p,\bu}(\bx')$
  vanish for all choices of columns $\bj=(j_1,\dots,j_p)$ if and only
  if there exists $\{j_1,\dots,j_{n-s}\} \subset \{1,\dots,q\}$ such 
  that $\lambda^H_{j_1,\bu}(\bx')=\cdots=\lambda^H_{j_{n-s},\bu}(\bx')=0$.
\end{lemma}
\begin{proof}
  Take an arbitrary representative $\bx^*$ of $\bx'$ in
  $\KKbar{}^{n+1}$, and consider the polynomial 
  $(1+\lambda^H_{1,\bu}(\bx^*)Y_1) \cdots (1+\lambda^H_{q,\bu}(\bx^*)Y_q),$
  for new variables $Y_1,\dots,Y_q$. The products
  $\lambda^H_{j_1,\bu}(\bx^*)\cdots \lambda^H_{j_p,\bu}(\bx^*)$ are all zero
  if and only if this polynomial has degree less than $p$, that is, if
  and only if $q-p+1=n-s$ terms among
  $\lambda^H_{1,\bu}(\bx^*),\dots,\lambda^H_{q,\bu}(\bx^*)$ vanish.
\end{proof}

For a given $\bu$ and generic coefficients $\lambda_{j,k,\ell}$ and $\mu_{i,k,\ell}$,
 the linear forms $\lambda^H_{j,k,\bu}$ are all pairwise distinct, so
the condition of the lemma holds if and only if there exist
$\bj=\{j_1,\dots,j_{n-s}\} \subset \{1,\dots,q\}$ and
$\bv=(v_1,\dots,v_{n-s})$, with $v_k$ in $\{1,\dots,\delta_k\}$ for all
$k$, such that $\lambda^H_{j_k,v_k,\bu}(\bx')=0$ 
for $k=1,\dots,n-s$.

This implies that for a fixed $\bu$, the possible values of $\bx'=(x_0,\dots,x_{n-s}) \in \P^{n-s}(\KKbar)$ are
determined as solutions of a linear system of size $n-s$. For a
generic choice of the coefficients $\lambda_{j,k,\ell}$ and
$\mu_{i,k,\ell}$, none of these points satisfies $X_0=0$, so that
$\assG_2(0)$ holds.

\paragraph{Property $\assG_3(0)$.} From $\assG_2(0)$, we know that
the projective variety defined by $\bA^\mH$ has no point at infinity,
so it is finite; as a result, the affine algebraic set defined by
$\bA$ is finite as well. In addition, all the affine solutions to
$\bA$ are obtained by setting $X_0=1$ in the projective solutions of
$\bA^\mH$. In other words, they are obtained by choosing indices
$\bu=(u_1,\dots,u_s)$ with $u_k$ in $\{1,\dots,\gamma_k\}$ for all $k$,
column indices $\bj=(j_1,\dots,j_{n-s})$, and
$\bv=(v_1,\dots,v_{n-s})$, with $v_k$ in $\{1,\dots,\delta_k\}$
for all $k$, solving the affine linear system
$$\lambda_{j_1,v_1,\bu}(X_1,\dots,X_{n-s})=\cdots=\lambda_{j_{n-s},v_{n-s},\bu}(X_1,\dots,X_{n-s})=0$$ 
and using the expressions
$$X_{n-s+1}=\phi_{n-s+1,\bu}(X_1,\dots,X_{n-s}),\dots,X_{n}=\phi_{n,\bu}(X_1,\dots,X_{n-s}),$$
where $\phi_{k,\bu}(X_1,\dots,X_{n-s})=\Phi_{n-s+1,\bu}(1,X_1,\dots,X_{n-s})$
for all $k$.
To prove that the ideal generated by $\bA$ is radical, we prove that
at any point as described above, the Jacobian matrix of $\bA$ with
respect to $X_1,\dots,X_n$ has full rank.

Let thus $\bu$, $\bj$ and $\bv$ be as above, let $\bx \in \KKbar{}^n$
be the corresponding point in $V(\bA)$, and consider equations
$(a_1,\dots,a_s)$ first. Each such equation is a product of linear forms
such as $a_i=\prod_{k=1}^{\gamma_i} \mu_{i,k}$, with $\mu_{i,u_i}(\bx)=0$.
Since the coefficients $\mu_{i,k,\ell}$ are chosen generically, for
$i=1,\dots,s$ and $k \ne u_i$, $\mu_{i,k}(\bx)$ is non-zero; as a
result, in the local ring at $\bx$, the polynomials $(a_1,\dots,a_s)$
are equal (up to units) to the linear forms
$(\mu_{1,u_1},\dots,\mu_{s,u_s})$.

Next, we consider the $p$-minors of $\mL$; in what follows, we 
write $\bx'=(x_1,\dots,x_{n-s})$. Our starting point is that due to 
the genericity of the coefficients $\lambda_{j,k,\ell}$, since 
$$\lambda_{j_1,v_1,\bu}=\cdots=\lambda_{j_{n-s},v_{n-s},\bu}=0$$
only admits $\bx'$ as a solution,
none of the other linear forms $\lambda_{j,k,\bu}$ vanishes at $\bx'$.
Equivalently, none of the other linear forms $\lambda_{j,k}$ vanishes at $\bx$.

Recall that $n=q-p+s+1$, so that $n-s = q-(p-1)$. Hence, there are
exactly $p-1$ columns of $\mL$ not indexed by $\bj=(j_1,\dots,j_{n-s})$; call
them $\bj'=(j'_1,\dots,j'_{p-1})$. We can then consider the 
products
$$ \lambda_{j_1} \lambda_{j'_1} \cdots \lambda_{j'_{p-1}},\dots, \lambda_{j_{n-s}}
\lambda_{j'_1} \cdots \lambda_{j'_{p-1}};$$ each of them (up to a non-zero
constant) is a $p$-minor of $\mL$, so they appear as elements in the
sequence $(a_{s+1},\dots,a_m)$, say as
$(a_{e_1},\dots,a_{e_{n-s}})$. By the remark of the previous
paragraph, in the local ring at $\bx$, up to non-zero constants, these
polynomials are respectively equal to the linear forms
$\lambda_{j_1,v_1},\dots,\lambda_{j_{n-s},v_{n-s}}$.  

To summarize, we have found that the linear equations
$(\mu_{1,u_1},\dots,\mu_{s,u_s})$ and
$(\lambda_{j_1,v_1},\dots,\lambda_{j_{n-s},v_{n-s}})$ belong to the
ideal $\langle \bA \rangle_\m$, where $\m$ is the maximal ideal at
$\bx$. As a result, the Jacobian matrix of $\bA$ must be invertible
at $\bx$, and $\assG_3(0)$ holds.

\medskip

At this stage, we have established all assumptions necessary to apply
Proposition~\ref{prop:degree_fiber}. Since $\bB$ satisfies
$\assA_1,\assA_2$ and $\bA=\bB_0$ satisfies $\assG_1,\assG_2,\assG_3$,
we deduce that the sum of the multiplicities of the isolated solutions
of $\bC=\bB_1$ is at most $c$, where $c$ is the number of solutions of
$\bA$. 

\begin{lemma}\label{lemma:column:c_estimate}
  Under the above assumptions, $c=\gamma_1\cdots \gamma_s
  E_{n-s}(\delta_1, \ldots, \delta_q)$.
\end{lemma}
\begin{proof}
  To estimate $c$, note first that there are $\gamma_1\cdots
  \gamma_s$ choices of $\bu$. For each choice of $\bu$, there are
  $E_{n-s}(\delta_1, \ldots, \delta_q)$ ways to
  choose $\bj$ and $\bv$, where $E_{n-s}$ denotes the elementary
  symmetric polynomial of degree $n-s$.   
\end{proof}
This proves the first part of Proposition~\ref{prop:coldeg}.
We can now inspect assumptions $\assD_1,\dots,\assD_4$, which are
needed to apply the algorithms of
Propositions~\ref{prop:compute_isolated} and~\ref{prop:compute_regular}. For the cost analysis 
below, as in Theorem~\ref{theo:2}, we assume that 
all $\gamma_i$'s and $\delta_j$'s are at least equal to $1$.

\paragraph{Property $\assD_1$.} We know that $\assG_1(0),\assG_2(0),\assG_3(0)$
hold, so we are going to compute a zero-dimensional parametrization of
$V(\bA)$.  We do this by following the description of the solutions of
$\bA$ given in the previous paragraph: for any choice of indices
$\bu$, $\bj$ and $\bv$ as above, the corresponding point $\bx \in
\KKbar{}^n$ in $V(\bA)$ can be found by solving the linear system of
size $n$ given by $(\mu_{1,u_1},\dots,\mu_{s,u_s})$ and
$(\lambda_{j_1,v_1},\dots,\lambda_{j_{n-s},v_{n-s}})$, so in time
$O(n^3)$. We repeat this procedure $c$ times, using a total of $O(c
n^3)$ operations in $\KK$.

Knowing all the points in $V(\bA)$, we can construct a zero-dimensional
parametrization $\scrR_0$ such that $Z(\scrR_0)=V(\bA)$ in time
$O\tilde{~}(c n)$ by means of fast interpolation~\cite[Chapter~10]{GaGe03}.
(Note that for practical purposes, we may modify the algorithm
of Propositions~\ref{prop:compute_isolated} and~\ref{prop:compute_regular} to take into account the
fact that all points in $V(\bA)$ are in $\KK^n$.)

Hence the total cost here is in $O(c n^3)$ operations in $\KK$..

\paragraph{Property $\assD_2$.} Next, we need to determine an upper bound 
$e$ on the degree of the curve $V(J')$, where $J'$ is the union of the
one-dimensional irreducible components of $V(\bB) \subset
\KKbar{}^{n+1}$ whose projection on the $T$-axis is dense.

\begin{lemma}\label{lemma:columndegree:e_estimate}
  Under the above assumptions and notation, $e$ is bounded above by
  $(\gamma_1+1)\cdots(\gamma_s+1) E_{n-s}(\delta_1+1, \ldots, \delta_q+1)$.
\end{lemma}
\begin{proof}
  Let us write $V(\bB)=V(J') \cup V' \cup V''$, where $V''$ is the
  union of the other components of dimension one of $V(\bB)$ and $V''$
  is the union of the components of higher dimension (by $\assA_1$,
  $V(\bB)$ has no isolated point), and let $H$ be a generic hyperplane
  in coordinates $T,X_1,\dots,X_n$. Then, $(V(J') \cup V') \cap V(H)$
  is a finite set consisting of $\deg(V(J')) + \deg(V')$ points,
  whereas $V'' \cap V(H)$ consists only on components of positive
  dimension; these two sets are disjoint. Thus, we can take for $e$
  the number of isolated points of $V(\bB)\cap V(H)$.

The hyperplane $H$ is defined by an equation
$h_0 + h_1 X_1 + \cdots + h_{n}X_{n} + h_{n+1} T=0$. This equation
allows us to rewrite $T$ as
$\eta(X_1,\dots,X_n)=-(h_0 + h_1 X_1 + \cdots +
h_{n}X_{n})/h_{n+1}$;
the points in $V(\bB)\cap V(H)$ are thus in one-to-one correspondence
with the solutions of the system
$(\beta_1,\dots,\beta_s,\beta_{s+1},\dots,\beta_m)$, where
$\beta_i=(1-\eta) a_i + \eta g_i$, for $i=1,\dots,s$, and
$\beta_{s+1},\dots,\beta_m$ are the $p$-minors of the matrix
${\nu}=(1-\eta)\, \mL + \eta \, \mF $.  Now, the polynomials
$(\beta_1,\dots,\beta_s)$ have respective degrees at most
$(\gamma_1+1),\dots,(\gamma_s+1)$, and the column degrees of ${\nu}$
are $\delta_1+1,\dots,\delta_q+1$.

We can then apply Proposition~\ref{prop:degree_fiber}, which shows we
can take for $e$ the integer $(\gamma_1+1)\cdots(\gamma_s+1)
E_{n-s}(\delta_1+1, \ldots, \delta_q+1)$.  
\end{proof}

\paragraph{Property $\assD_3$.} Finally, we need to give an estimate on
the size of a straight-line program that computes the polynomials
$\bB=(b_1,\dots,b_m)$, assuming that we are given a straight-line
program $\Gamma$ of size $\sigma$ that computes polynomials $G=(g_1,\dots,g_s)$ and
the entries of $\mF$.

First, we estimate the complexity of computing the polynomials
$(b_1,\dots,b_s)$. For $i \le s$, the $i$th polynomial $b_i$ is equal
to $(1-T)a_i + T g_i$, where $a_i$ is a product of $\gamma_i$ linear
forms in $n$ variables. This polynomial can be computed in $O(n
\gamma_i)$ operations in $\KK$, hence for a total of $O(n
(\gamma_1+\cdots+\gamma_s))$ operations for $(a_1,\dots,a_s)$, and $O(\sigma+n (\gamma_1+\cdots+\gamma_s))$ for $(b_1,\dots,b_s)$. 

The polynomials $(b_{s+1},\dots,b_m)$ are the $p$-minors of
$\mU=(1-T)\cdot\mL+T\cdot\mF$.  The polynomials $\lambda_1,\dots,\lambda_q$ can
be computed in $O(n (\delta_1+\cdots+\delta_q))$ operations, so that
the entries of $\mU$ can be computed in $O(\sigma +
n(\delta_1+\cdots+\delta_q))$ operations. From that, all $p$-minors of
$\mU$ can be deduced in $O({q \choose p} n^3)$ further steps.  To
summarize, all polynomials in $\bB$ can be computed by a straight-line
program $\Gamma'$ of size $\sigma'=O(\sigma + {q \choose p} n^3 +n(\gamma_1+\cdots+\gamma_s+\delta_1+\cdots+\delta_q))$.

\paragraph{Completing the cost analysis.}
We can then apply Proposition~\ref{prop:compute_isolated},
whose runtime is $\softO(c^5 m n^2  + c(e+c^5) n(\sigma' + n^3))$ operations
in~$\KK$; since $m \le n + {q \choose p}$, this can be simplified as
$$\softO\left (c(e+c^5) n \left(\sigma + {q \choose p} n^3
+n(\gamma_1+\cdots+\gamma_s+\delta_1+\cdots+\delta_q)\right)\right).$$
Since $s \leq n$, $\gamma = \max(\gamma_1, \ldots, \gamma_s)$ and
$\delta = \max(\delta_1, \ldots, \delta_q)$, our bound becomes
$$\softO\left (c(e+c^5) n(\sigma + {q \choose p}n^3+n^2\gamma + nq
  \delta)\right).$$ This can also be rewritten as
$$\softO\left (c(e+c^5 )(\sigma + {q \choose p}n^3+n^2\gamma + nq
  \delta)\right),$$
since one easily checks that $e \ge 2^n$ (because by assumption we
have $\gamma_i\geq 1 $ and $\delta_i\geq 1$), so that
$n \in \softO(e)$. A last factorization shows that the bound can be
simplified to
\[
  \softO\left ({q \choose p}  c(e+c^5 )n^3(\sigma +\gamma + q
    \delta)\right).
\]
Using again  that $n\le\log_2(e)$, we can omit the factor
$n^3$ from the $\softO(\ )$, and we conclude the proof of Proposition~\ref{prop:coldeg}.
The resulting algorithm, called $\mathsf{ColumnDegree}$, is described
herafter.

\begin{algorithm}[h]
\caption{$\mathsf{ColumnDegree}(\Gamma)$}
{\bf Input}: a straight-line program $\Gamma$ of length $\sigma$ that computes 
\begin{itemize}  
\setlength\itemsep{0em}
\item $F \in \KK[X_1, \ldots, X_n]^{p \times q}$ with $\deg(f_{i,j}) \leq \delta_j$ for all $j$ and $p \le q$
\item polynomials $G = (g_1, \ldots, g_s)$ in $\KK[X_1, \ldots, X_n]$, with $n=q-p+s+1$
\end{itemize}
{\bf Output}: a zero-dimensional parametrization of the isolated points of $\VpFG{p}{\mF}{G}$
\begin{enumerate}\setlength\itemsep{0em}
\item for any sequence $\bu=(u_1,\dots,u_s)$, with $u_j \in \{1,\dots,\gamma_j\}$ for all $j$
\begin{enumerate}\setlength\itemsep{0em}
\item for any subsequence $\bj=(j_1,\dots,j_{n-s})$ of $(1,\dots,q)$
\begin{enumerate}\setlength\itemsep{0em}
\item for any sequence $\bv=(v_1,\dots,v_{n-s})$, with $v_k$ in $\{1,\dots,\delta_k\}$ for all $k$
\begin{enumerate}\setlength\itemsep{0em}
 \item compute a zero-dimensional parametrization $\scrR_{\bi,\bj,\bv}$ of the solution of the system 
$$\mu_{1,u_1}=\cdots=\mu_{s,u_s}=\lambda_{j_1,v_1}=\cdots=\lambda_{j_{n-s},v_{n-s}}=0$$

\hfill $\text{\sf{cost:~}} O(cn^3)$, with $c=\gamma_1\cdots\gamma_s E_{n-s}(\delta_1,\dots,\delta_q)$
\end{enumerate}
\end{enumerate}
\end{enumerate}
\item combine all $(\scrR_{\bu,\bj,\bv})_{\bu,\bj,\bv}$ into a zero-dimensional parametrization $\scrR$

  \hfill $\text{\sf{cost:~}} \softO(cn)$

\item construct a straight-line program $\Gamma'$ that computes all polynomials $\bB$

\hfill length of $\Gamma'$ is $\sigma'=O(\sigma + {q \choose p} n^3 + n
(\alpha_1+\cdots+\alpha_p) + n(\gamma_1 + \cdots + \gamma_s))$

\item return $\mathsf{Homotopy}(\Gamma',\scrR)$ 

\hfill $\text{\sf{cost:~}} \softO\left (c^5 m n^2 + c(e+c^5)n (\sigma' + n^3)\right)$, 

\hfill with $e=(\gamma_1+1)\cdots(\gamma_s+1) E_{n-s}(\delta_1+1,\dots,\delta_q+1)$
\end{enumerate}
\label{ColHom}
\end{algorithm}

Finally, to prove Proposition~\ref{prop:coldeg_simple}, we rely on the
algorithm called $\mathsf{ColumnDegree\_simple}$, which differs from
$\mathsf{ColumnDegree}$, only at the last step where Algorithm
$\mathsf{Homotopy\_simple}$ is called instead of $\mathsf{Homotopy}$.
Hence, one applies Proposition~\ref{prop:compute_regular}, which yields a
runtime $\softO(c^2\,m \,n^2  + c\,e\,n (\sigma'+n^2) )$ operations
in $\KK$. Using again $m \leq n+ \binom{q}{p}\leq n\binom{q}{p}$,
$\sigma'=O(\sigma + {q \choose p} n^3 +n(n\gamma+q\delta))$, we
obtain as a bound
\[
  \softO\left (
    {q \choose p}\, c^2\,n^3 +  c\, e\,n (\sigma+ {q \choose p}n^3 \, +n^2\gamma +nq \delta)
  \right ),
\]
which we simplify as
\[
 \softO\left (   {q \choose p}c\, e\,n^4 (\sigma+\gamma+q\delta) \right ),
\]
taking into account that $c \le e$.
Since $e \ge 2^n$, the term $n^4$ can be absorbed in 
the $\softO(\, )$.
This concludes the proof of Proposition~\ref{prop:coldeg_simple}. 

%%%%%%%%%%%%%%%%%%%%%%%%%%%%%%%%%%%%%%%%%%%%%%%%%%%%%%%%%%%%
%%%%%%%%%%%%%%%%%%%%%%%%%%%%%%%%%%%%%%%%%%%%%%%%%%%%%%%%%%%%
%%%%%%%%%%%%%%%%%%%%%%%%%%%%%%%%%%%%%%%%%%%%%%%%%%%%%%%%%%%%

\section{Preliminaries for the row-degree homotopy}\label{sec:prel-row}

In this section, we work with two families of matrices of size $p
\times q$, with $p \le q$, and with entries that are polynomials in
$n=q-p+1$ variables; we prove several properties that will be used 
in our row-degree homotopy algorithm. Let $\balpha=(\alpha_1,\dots,\alpha_p)$
be positive integers. The matrices we consider are
\begin{align}\label{eqdef:type2}
\mM^H= \left( \begin{matrix}
\lambda^H_{1,1} & \lambda^H_{1,2} & \cdots & \lambda^H_{1, q}\\
 \lambda^H_{2,1} &  \lambda^H_{2,2} & \cdots & \lambda^H_{2, q}\\
 \vdots & & & \vdots\\
 \lambda^H_{p,1} &  \lambda^H_{p,2}& \cdots & \lambda^H_{p, q}
\end{matrix} \right),
\end{align}
and matrices of a more specialized kind of the form
\begin{align}\label{eqdef:type1}
\mN^H= \left( \begin{matrix}
\lambda^H_{1,1} & 0 & \cdots & 0 & \lambda^H_{1,p+1} & \cdots & \lambda^H_{1, q}\\
0 & \lambda^H_{2,2} & \cdots & 0 & \lambda^H_{2,p+1} & \cdots & \lambda^H_{2, q}\\
\vdots & \vdots & \ddots & \vdots & \vdots & \ddots & \vdots\\
0 & 0 & \cdots & \lambda^H_{p,p} & \lambda^H_{p,p+1} & \cdots & \lambda^H_{p, q}
\end{matrix} \right),
\end{align}
where the ${}^H$ superscript indicates that all entries are
homogenous.  In both cases, for all $i,j$, the entry $\lambda^H_{i,j}$
is a product of $\alpha_i$ homogeneous linear forms in $n+1$ variables
$X_0,\dots,X_n$ with coefficients in $\KK$ (except when
$\lambda^H_{i,j}$ is explicitly set to zero in the second case), that
is, $\lambda^H_{i,j}=\prod_{k=1}^{\alpha_i} \lambda^H_{i,j,k}$.  

We are interested in describing the projective algebraic sets defined
in $\P^n(\KKbar)$ by the $p$-minors of $\mN^H$ and $\mM^H$ (note that
these minors are all homogeous). In the rest of this section, if
$\mA^H$ is a matrix with polynomial entries that are homogeneous in
$X_0,\dots,X_n$, we use the notation $V_t(\mA^H)$ to denote the
projective set defined by its $t$-minors in $\P^n(\KKbar)$, for any
$t\ge 1$ (we use the same notation for affine algebraic sets in those
cases when the entries of our matrices are polynomials in $X_1,\dots,X_n$;
this should cause no confusion).

\begin{proposition}\label{lemma:appendix}
  For generic choices of the coefficients of the linear forms
  $\lambda^H_{i,j,k}$, the following holds:
  \begin{itemize}
  \item the projective algebraic sets $V_p(\mM^H)$ and $V_p(\mN^H)$
    have no solution at infinity (that is, with $X_0=0$);
  \item the Jacobian matrices of $I_p(\mM^H)$ and $I_p(\mN^H)$ with
    respect to $(X_0,\dots,X_n)$ have rank $n$ at every point of the
    above sets.
\end{itemize}
\end{proposition}

The bulk of this section is devoted to prove this proposition.  Our
strategy is to work all along with linear forms with indeterminate
coefficients, and establish the properties we want in this context.
Explicitly, we prove below properties called
$\assI_2(\balpha,q),\assJ_2(\balpha,q)$ and
$\assI_4(\balpha,q),\assJ_4(\balpha,q),$ which establish the
proposition.  In what follows, for any ring $R$ and any matrix
$\mM \in R^{p\times q}$, if $S$ is a subsequence of $(1,\dots,p)$ and
$T$ a subsequence of $(1,\dots,q)$, $\mM_{S,T}$ is the submatrix of
$\mM$ obtained by keeping rows indexed by $S$ and columns indexed by
$T$. We also call this the $(S,T)$-submatrix of $\mM$.

Let thus ${\cal A}=q(n+1)(\alpha_1 + \cdots + \alpha_p)$; this is the
number of coefficients needed to define homogeneous linear forms
$\lambda^H_{i,j,k}$ in $X_0,\dots,X_n$, for $i=1,\dots,p$,
$j=1,\dots,q$ and $k=1,\dots,\alpha_i$. If needed, we will write
${\cal A}={\cal A}(\balpha,q)$ to make the dependency in $\balpha$ and
$q$ explicit.  Let then $\mathfrak{L}$ be the sequence of ${\cal A}$
indeterminates $\mathfrak{L}=(\mathfrak{l}_{i,j,k,r})$, for $i,j,k$ as
above and $r=0,\dots,n$, and define
$$\mathfrak{l}^H_{i,j,k} = \mathfrak{l}_{i,j,k,0}X_0 + \mathfrak{l}_{i,j,k,1} X_1 +\cdots + \mathfrak{l}_{i,j,k,n} X_n,$$
as well as 
$$\mathfrak{l}^H_{i,j} = \mathfrak{l}^H_{i,j,1} \cdots \mathfrak{l}^H_{i,j,\alpha_i} \in \KK[\mathfrak{L}][\tilde\bX],$$
with $\tilde\bX=(X_0,X_1,\dots,X_n)$. We can then define the
matrix
\begin{align}\label{eq:matM}
\mathfrak{M}^H_{\balpha,q}=\left [\begin{matrix}
\mathfrak{l}^H_{1,1} & \cdots & \mathfrak{l}^H_{1,q}\\
 \vdots & & \vdots\\
\mathfrak{l}^H_{p,1} & \cdots & \mathfrak{l}^H_{p,q}
  \end{matrix}\right ]\in \KK[\mathfrak{L}][\tilde\bX]^{p\times q}.  
\end{align}
Remark that for all $i,j$, the $(i,j)$-th entry of
$\mathfrak{M}^H_{\balpha,q}$ has degree $\alpha_i$ in $\tilde \bX$;
this matrix is thus the ``generic'' model of the matrix $\mM^H$ seen previously.

Given $\Lambda=(\lambda_{i,j,k,r})\in \KKbar{}^{\cal A}$, for any polynomial
$\mathfrak{F}$ in $\KK(\mathfrak{L})[\tilde \bX]$, we write
$\mathfrak{F}(\Lambda,\tilde\bX)$ for the polynomial obtained by
evaluating $\mathfrak{l}_{i,j,k,r}$ at $\lambda_{i,j,k,r}$, for all
indices $i,j,k,r$ as above, as long as no denominator vanishes through
this evaluation; the notation extends to polynomial matrices. More
generally, for a field $\LL$ containing $\KK$, and $\Lambda$ in $\LL^{\cal A}$, the
notation $\mathfrak{F}(\Lambda,\tilde\bX)$ is defined similarly.

Let next ${\cal A}'=n(n+1)(\alpha_1+\cdots+\alpha_p)$; as above, 
we will write ${\cal A}'={\cal A}'(\balpha,q)$ when needed. Let
$\mathfrak{L}'\subset \mathfrak{L}$ be the sequence of ${\cal A}'$
indeterminates $\mathfrak{L}'=(\mathfrak{l}_{i,j,k,r})$, for indices
$i,j,k,r$ as follows: $i$ is in $\{1,\dots,p\}$, $j$ is in
$\{i,p+1,\dots,q\}$, and as previously, $k$ is in
$\{1,\dots,\alpha_i\}$ and $r$ is in $\{0,\dots,n\}$. Remark that the
polynomials $\mathfrak{l}^H_{i,j}$, for $i,j$ as above, are in
$\KK[\mathfrak{L}'][\tilde\bX] \subset \KK[\mathfrak{L}][\tilde \bX]$,
and allow us to define
\begin{align}\label{eq:matMprime}
\mathfrak{N}^H_{\balpha,q}=\left [\begin{matrix} \mathfrak{l}^H_{1,1} & 0 & 0
    &\mathfrak{l}^H_{1,p+1} & \cdots & \mathfrak{l}^H_{1,q}\\ \vdots & \ddots &
    \vdots & \vdots & & \vdots\\ 0&0& \mathfrak{l}^H_{p,p}
    &\mathfrak{l}^H_{p,p+1} & \cdots & \mathfrak{l}^H_{p,q}
  \end{matrix}\right ]\in \KK[\mathfrak{L}'][\tilde\bX]^{p\times q}.
\end{align}
For $\Lambda' \in \KKbar{}^{{\cal A}'}$ and $\mathfrak{F}$ in 
$\KK(\mathfrak{L}')[\tilde\bX]$, the notation
$\mathfrak{F}(\Lambda',\tilde\bX)$ is defined as in the case of
polynomials over $\KK(\mathfrak{L})$ described previously.

\subsection{Setting up the recurrences}
The basic idea behind the proofs below is the following: to prove that
a property such as rank-deficiency holds for a matrix $\mathfrak{M}^H_{\balpha,q}$, we prove
that it holds for a matrix of the form $\mathfrak{N}^H_{\balpha,q}$,
and use an openness property. To prove that property for the
latter matrices, we proceed by induction, relying on the presence of
the left-hand diagonal block. Indeed, for a matrix such as
$\mathfrak{N}^H_{\balpha,q}$ to be rank-deficient at $\tilde\bx \in
\P^n(\KKGpbar)$, at least one of
$\mathfrak{l}^H_{1,1},\dots,\mathfrak{l}^H_{p,p}$ must vanish at
$\tilde\bx$.

Suppose for instance that
$\mathfrak{l}^H_{1,1}(\tilde\bx)=\mathfrak{l}^H_{2,2}(\tilde\bx)=0$,
while all other terms are non-zero. Then, the
$((1,2),(p+1,\dots,q))$-submatrix of
${\mathfrak{N}^H_{\balpha,q}}(\tilde\bx)$ itself must be
rank-deficient.  The constraints
$\mathfrak{l}^H_{1,1}(\tilde\bx)=\mathfrak{l}^H_{2,2}(\tilde\bx)=0$
give us two linear equations, which allow us to eliminate two
coordinates of $\tilde\bx$, say $X_{n-1}$ and $X_n$. We can perform
the corresponding substitution in the above submatrix, and we are left
with a matrix of size $2 \times (n-1)$ that is of the form
$\mathfrak{M}^H_{(\alpha_1,\alpha_2),n-1}(\mathfrak{H},(X_0,\dots,X_{n-2}))$,
with entries depending on $X_0,\dots,X_{n-2}$, for some vector of
coefficients $\mathfrak{H}$ obtained through the elimination of
$X_{n-1}$ and $X_n$. We can then invoke our induction assumption on
the latter matrix.

To formalize this process, for a subsequence
$\bi=(i_1,\dots,i_\kappa)$ of $(1,\dots,p)$, we call the
$(\bi,(p+1,\dots,q))$-submatrix of $\mathfrak{N}^H_{\balpha,q}$ the submatrix of
$\mathfrak{N}^H_{\balpha,q}$ {\em associated} to $\bi$; it consists of the rows of
$\mathfrak{N}^H_{\balpha,q}$ indexed by $\bi$ and columns $p+1,\dots,q$.
For such an $\bi$, we let $R_\bi$ be the set of
all tuples $\br=(r_1,\dots,r_\kappa)$, with $r_1$ in
$\{1,\dots,\alpha_{i_1}\}$, \dots, $r_\kappa$ in
$\{1,\dots,\alpha_{i_\kappa}\}$; for any $k$ in $\{1,\dots,\kappa\}$,
$r_k$ will be the index of the factor $\mathfrak{l}^H_{i_k,i_k,r_k}$
of $\mathfrak{l}^H_{i_k,i_k}$ we cancel. For given $\bi$ and $\br$,
we will let $\mathfrak{L}'_{\bi,\br} \subset \mathfrak{L}'$ be the
indeterminates corresponding to the coefficients of
$\mathfrak{l}^H_{i_1,i_1,r_1},\dots,\mathfrak{l}^H_{i_\kappa,i_\kappa,r_\kappa}$, and
of all entries $\mathfrak{l}^H_{i_1,p+1},\dots,\mathfrak{l}^H_{i_\kappa,q}$
of the submatrix associated to $\bi$ in $\mathfrak{N}^H_{\balpha,q}$.

By Gaussian elimination, we can rewrite the homogeneous linear
equations
$\mathfrak{l}^H_{i_1,i_1,r_1}=\dots=\mathfrak{l}^H_{i_\kappa,i_\kappa,r_\kappa}=0$ as
\begin{align}\label{eq:f_sr}
X_{n-\kappa+1}=\mathfrak{f}_{n-\kappa+1,\bi,\br}(X_0,\dots,X_{n-\kappa}),\dots,X_{n}=\mathfrak{f}_{n,\bi,\br}(X_0,\dots,X_{n-\kappa}),  
\end{align}
for some homogeneous linear forms
$\mathfrak{f}_{n-\kappa+1,\bi,\br},\dots,\mathfrak{f}_{n,\bi,\br}$ of
$(X_0,\dots,X_{n-\kappa})$ with coefficients in
$\KK(\mathfrak{L}'_{\bi,\br})$. Applying this substitution in the
entries of the submatrix of ${\mathfrak{N}^H_{\balpha,q}}$ associated
to $\bi$ gives us the $\kappa \times (n-1)$ matrix
$\mathfrak{M}^H_{\balpha_\bi,n-1}(\mathfrak{H}_{\bi,\br},\tilde\bX')$,
with $\balpha_\bi=(\alpha_{i_1},\dots,\alpha_{i_\kappa})$, whose
entries are products of homogeneous linear forms in
$\tilde\bX'=(X_0,\dots,X_{n-\kappa})$, and where $\mathfrak{H}_{\bi,\br}$ is a
vector of ${\cal A}(\balpha_\bi,n-1)$ elements in $\KK(\mathfrak{L}'_{\bi,\br})$.

The main result we will use in this section is the following lemma,
which summarizes how the above process allows us to describe the
projective zero-set of $t$-minors of $\mathfrak{N}^H_{\balpha,q}$, for
any $t \le p$. This will be the basis of several recursions.
\begin{lemma}\label{lemma:union}
  For $t$ in $\{1,\dots,p\}$, $V_t(\mathfrak{N}^H_{\balpha,q}) \subset \P^{n}(\KKGpbar)$ is the
  union of the sets
 \begin{align}\label{eq:union}
 \left \{(\tilde\bx',\mathfrak{f}_{n-\kappa+1,\bi,\br}(\tilde\bx'),\dots,\mathfrak{f}_{n,\bi,\br}(\tilde\bx')) \mid \tilde\bx' \in
  V_{\kappa-(p-t)}(\mathfrak{M}^H_{\balpha_\bi,n-1}(\mathfrak{H}_{\bi,\br},\tilde\bX')) \subset \P^{n-\kappa}(\KKGpbar)\right \},   
 \end{align}
 for $\bi=(i_1,\dots,i_\kappa)$ of length $\kappa \in \{p-t+1,\dots,\min(p,n-1)\}$ and $\br$ in $R_\bi$,
 and with $\tilde\bX'=(X_0,\dots,X_{n-\kappa})$, together with
 $$\left \{
 (1,\mathfrak{f}_{1,\bi,\br}(1),\dots,\mathfrak{f}_{n,\bi,\br}(1))\right
 \}$$ if $t=p$ and $n \le p$, with $\bi=(i_1,\dots,i_n)$ and $\br$ in $R_\bi$.
\end{lemma}
\noindent We have to write a special case for $t=p$ and $n \le p$ in the last part of the lemma,
since taking $\bi=(i_1,\dots,i_n)$ of length $\kappa=n$
in~\eqref{eq:union} would lead to consider points in $\P^0(\KKGpbar)$.
\begin{proof}
  A point $\tilde\bx \in \P^n(\KKGpbar)$ belongs to
  $V_t(\mathfrak{N}^H_{\balpha,q})$ if and only if some diagonal terms
  of $\mathfrak{N}^H_{\balpha,q}$ vanish at $\tilde\bx$, say
  $\mathfrak{l}^H_{i_k,i_k}(\tilde\bx)=0$ for $k=1,\dots,\kappa$ (all
  other $\mathfrak{l}^H_{i,i}(\tilde\bx)$ being non-zero), and
  if the submatrix of $\mathfrak{N}^H_{\balpha,q}$ associated to
  $\bi=(i_1,\dots,i_\kappa)$ has rank less than $\kappa-(p-t)$ at $\tilde
  \bx$.  In particular, we must have
  $\kappa-(p-t) > 0$, that is, $\kappa \ge p-t+1$.

  For $k=1,\dots,\kappa$, $\mathfrak{l}^H_{i_k,i_k}(\tilde\bx)=0$ if
  and only if there exists $r_k$ in $\{1,\dots,\alpha_{i_k}\}$ such
  that $\mathfrak{l}^H_{i_k,i_k,r_k}(\tilde\bx)=0$. Thus, $\tilde\bx$
  is in $V_t(\mathfrak{N}^H_{\balpha,q})$ if and only if there exists
  a subsequence $\bi=(i_1,\dots,i_\kappa)$  of $(1,\dots,p)$, with
  $\kappa \ge p-t+1$, and $\br=(r_1,\dots,r_\kappa)$ in $R_\bi$ such that
  $\mathfrak{l}^H_{i_1,i_1,r_1}(\tilde\bx)=\cdots=\mathfrak{l}^H_{i_\kappa,i_\kappa,r_\kappa}(\tilde\bx)=0$
  and the submatrix of $\mathfrak{N}^H_{\balpha,q}$ associated to $\bi$
  has rank less than $\kappa-(p-t)$ at $\tilde\bx$.  

  Applying~\eqref{eq:f_sr}, we deduce that the coordinates $(x_0,\dots,x_n)$ 
  of $\tilde \bx$ satisfy
  \begin{align*}
    x_{n-\kappa+1}=\mathfrak{f}_{n-\kappa+1,\bi,\br}(\tilde\bx'),\dots,x_{n}=\mathfrak{f}_{n,\bi,\br}(\tilde\bx'),
  \end{align*}
  with $\tilde\bx'=(x_0,\dots,x_{n-\kappa})$.  In particular, $\kappa
  \le n$, since otherwise this linear system would have no solution
  (recall that the coefficients are algebraically independent
  indeterminates). Remark also that $\tilde\bx'$ is a well-defined
  element of $\P^{n-\kappa}(\KKGpbar)$, that is, it is not identically
  zero, since otherwise $\tilde\bx$ would vanish as well.

  For $\bi=(i_1,\dots,i_\kappa)$ with $\kappa \le n-1$, applying the above
  substitution in the submatrix of $\mathfrak{N}^H_{\balpha,q}$
  associated to $\bi$ (which has size $\kappa \times (n-1)$), the rank
  condition above becomes that
  $\mathfrak{M}^H_{\balpha_\bi,n-1}(\mathfrak{H}_{\bi,\br},\tilde\bX')$
  has rank less than $\kappa-(p-t)$ at $\tilde \bx'$, that is, $\tilde
  \bx'$ is in
  $V_{\kappa-(p-t)}(\mathfrak{M}^H_{\balpha_\bi,n-1}(\mathfrak{H}_{\bi,\br},\tilde\bX'))$.
  In this case, we are done.

  When $\kappa=n$, that is, $\bi=(i_1,\dots,i_n)$ (this can happen
  only if $n \le p$), the linear equations above determine $\tilde\bx$
  entirely; setting $x_0=1$, we obtain
  $x_{1}=\mathfrak{f}_{1,\bi,\br}(1),\dots,x_{n}=\mathfrak{f}_{n,\bi,\br}(1).$
  In this case, the submatrix of $\mathfrak{N}^H_{\balpha,q}$
  associated to $\bi$ has size $n \times (n-1)$. Using the
  specialization of the coefficients that sets the off-diagonal entry
  to $0$ and the $i$th diagonal entries to $X_0^{\alpha_i}$,
  $i=1,\dots,n-1$, we see that its evaluation at $\tilde\bx$ has rank
  $n-1$; as a result $\mathfrak{N}^H_{\balpha,q}$ has rank $p-1$ at
  $\tilde \bx$. Thus, we need to take $\kappa=n$ into account only if
  $t=p$, that is, if we are interested in the maximal minors; in this
  case, we have to take into account the point $\left \{
  (1,\mathfrak{f}_{1,\bi,\br}(1),\dots,\mathfrak{f}_{n,\bi,\br}(1))\right
  \}$.
\end{proof}

\subsection{Solutions with higher rank defect} 
We discuss here the case $t=p-1$.  We take parameters
$\balpha=(\alpha_1,\dots,\alpha_p)$ and $q$, with $2 \le p \le q$, and
we write ${\cal A}={\cal A}(\balpha,q)$ and ${\cal A}'={\cal A}'(\balpha,q)$; we will establish
the following properties.
\begin{description}[leftmargin=*]
\item[$\assI_1(\balpha,q).$] The projective algebraic set
  $V_{p-1}(\mathfrak{M}^H_{\balpha,q})\subset\P^n(\KKGbar)$ is empty.
\item[$\assJ_1(\balpha,q).$] The projective algebraic set
  $V_{p-1}(\mathfrak{N}^H_{\balpha,q})\subset\P^n(\KKGpbar)$ is empty.
\end{description}
The first step of the proof is to establish that for $\balpha$ and $q$
as above, $\assJ_1(\balpha,q)$ implies $\assI_1(\balpha,q)$. Let us
thus fix $\balpha$ and $q$.  Assumption $\assJ_1(\balpha,q)$ implies
that $V_{p-1}(\mathfrak{N}^H_{\balpha,q}(\Lambda',\tilde\bX))$ is
empty for a generic $\Lambda'$ in~$\KKbar{}^{{\cal A}'}$.  We will prove that
$V_{p-1}(\mathfrak{M}^H_{\balpha,q}(\Lambda,\tilde\bX))$ is empty for
a generic $\Lambda$ in~$\KKbar{}^{\cal A}$, which in turn establishes
$\assI_1(\balpha,q)$.

Consider the ideal $I_{p-1}(\mathfrak{M}^H_{\balpha,q})$ in the polynomial
ring $\KK[\mathfrak{L},\tilde\bX]$ in ${\cal A}+n+1$ variables. This ideal
defines an algebraic set $Z_{\balpha,q}$ in $\KKbar{}^{\cal A} \times
\P^n(\KKbar)$, and we let $\Delta_{\balpha,q} \subset \KKbar{}^{\cal A}$ be its
projection on the first factor: this is the set of all $\Lambda$ such that
$V_{p-1}(\mathfrak{M}^H_{\balpha,q}(\Lambda,\tilde\bX))$ is not empty. Because the source is
a projective space, $\Delta_{\balpha,q}$ is closed (so its complement is
open), and we just have to verify that it is not equal to the whole
$\KKbar{}^{\cal A}$. This follows readily from property $\assJ_1(\balpha,q)$,
which proves that generic matrices of the form
$\mathfrak{N}^H_{\balpha,q}(\Lambda',\tilde\bX)$ do not belong to $\Delta_{\balpha,q}$, so 
$\assI_1(\balpha,q)$ holds.

We finish the proof by induction. We first take $p=q$ and consider
$\assJ_1(\balpha,q)$.  In this case, $n=1$ and
$\mathfrak{N}^H_{\balpha,q}$ is a diagonal matrix, whose diagonal
entries are products of linear forms in $(X_0,X_1)$ with indeterminate
coefficients. Hence, no pair of entries $\mathfrak{N}^H_{\balpha,q}$
have any common solution in $\P^1(\KKGpbar)$, so the rank of
$\mathfrak{N}^H_{\balpha,q}$ is at least $p-1$ at any $\tilde\bx \in
\P^1(\KKGpbar)$. As a result, $\assJ_1(\balpha,p)$ holds, and so does
$\assI_1(\balpha,p)$, by the claim in the previous paragraph.

Consider next a pair $(\balpha,q)$, with
$\balpha=(\alpha_1,\dots,\alpha_p)$ and $2 \le p < q$, and suppose
that $\assI_1(\balpha',q')$ holds for all $(\balpha',q')$ with
$\balpha'=(\alpha'_1,\dots,\alpha'_{p'})$, $2 \le p' \le q'$, $p' \le p$ and $q'
< q$; we prove that $\assJ_1(\balpha,q)$ holds (as above, this will
also imply $\assI_1(\balpha,q)$).

Take $t=p-1$ in Lemma~\ref{lemma:union}. Then, the parameters
$(\kappa-(p-t),\balpha_\bi,n-1)$ used in each
expression~\eqref{eq:union} are of the form
$(\kappa-1,\balpha_\bi,n-1)$, with $2 \le \kappa \le \min(p,n-1)$.
Since the ${\cal A}(\balpha_\bi,n-1)$ entries of $\mathfrak{H}_{\bi,\br}$ are algebraically
independent over $\KK$, $\KK(\mathfrak{H}_{\bi,\br})$ is isomorphic to 
$\KK(\lambda_{u,j,k,r})$, for $u=1,\dots,\kappa$, 
$j=1,\dots,n-1$, $k=1,\dots,\alpha_{i_u}$ 
and $r=0,\dots,n-\kappa$, so that $V_{\kappa-1}(\mathfrak{M}^H_{\balpha_\bi,n-1}(\mathfrak{H}_{\bi,\br},\tilde\bX'))$
has the same cardinality as 
$V_{\kappa-1}(\mathfrak{M}^H_{\balpha_\bi,n-1})$.
As a result, since $\balpha_\bi$ has length $\kappa\ge 2$, and since we also
have $\kappa \le n-1$, $\kappa \le p$ and $n-1 < q$, we can apply the
induction hypothesis and deduce that all
$V_{\kappa-1}(\mathfrak{M}^H_{\balpha_\bi,n-1}(\mathfrak{H}_{\bi,\br},\tilde\bX'))$ appearing
in Lemma~\ref{lemma:union} are empty. This in turn implies that
$V_{p-1}(\mathfrak{N}^H_{\balpha,q})$ is empty, as claimed.

\subsection{Solutions at infinity} Next, we focus on the case $t=p$.
We take parameters $\balpha=(\alpha_1,\dots,\alpha_p)$ and $q$, with
$1 \le p \le q$, and we write ${\cal A}={\cal A}(\balpha,q)$ and
${\cal A}'={\cal A}'(\balpha,q)$; then, we prove the following
properties.
\begin{description}[leftmargin=*]
\item[$\assI_2(\balpha,q).$] The projective algebraic set
  $V_p(\mathfrak{M}^H_{\balpha,q}) \subset \P^n(\KKGbar)$ has no point
  satisfying $X_0=0$.
\item[$\assJ_2(\balpha,q).$] The projective algebraic set
  $V_p(\mathfrak{N}^H_{\balpha,q}) \subset \P^n(\KKGpbar)$ has no point
  satisfying $X_0=0$.
\end{description}
In particular, this implies that these sets are finite. We will prove
these properties as we did in the previous paragraph; the first step
is thus to establish that for $\balpha$ and $q$ as above,
$\assJ_2(\balpha,q)$ implies $\assI_2(\balpha,q)$.

Let us thus fix $\balpha$ and $q$, and assume that
$\assJ_2(\balpha,q)$ holds. We prove that
$V_p(\mathfrak{M}^H_{\balpha,q}(\Lambda,\tilde\bX))$ has no point at
infinity for a generic $\Lambda$ in $\KKbar{}^{\cal A}$; this will imply
$\assI_2(\balpha,q)$. Consider the ideal generated by
$I_{p}(\mathfrak{M}^H_{\balpha,q})$ and $X_0$ in the polynomial ring
$\KK[\mathfrak{L},\tilde\bX]$ in ${\cal A}+n+1$ variables. This ideal defines
an algebraic set $Z'_{\balpha,q}$ in $\KKbar{}^{\cal A} \times \P^n(\KKbar)$,
and we let $\Delta'_{\balpha,q} \subset \KKbar{}^{\cal A}$ be its projection
on the first factor: this is thus the set of all $\Lambda$ in
$\KKbar{}^{\cal A}$ such that
$V_p(\mathfrak{M}^H_{\balpha,q}(\Lambda,\tilde\bX))$ has a point at
infinity. Because the source is a projective space, $\Delta'_{\balpha,q}$ is
closed (so its complement is open), and we just have to verify that it
is not equal to the whole $\KKbar{}^{\cal A}$. This follows from property
$\assJ_2(\balpha,q)$, which implies that matrices of the form
$\mathfrak{N}^H_{\balpha,q}(\Lambda',\tilde\bX)$, for generic
$\Lambda'$ in $\KKbar{}^{{\cal A}'}$, do not belong to $\Delta'_{\balpha,q}$.

Again, we finish the proof by induction. We first take $p=q$, and we
prove that $\assJ_2(\balpha,q)$ holds ($\assI_2(\balpha,q)$ will follow,
by the previous paragraph). In this case, $n=1$ and
$\mathfrak{N}^H_{\balpha,q}$ is a diagonal matrix, whose diagonal entries
are products of homogeneous linear forms in $(X_0,X_1)$ with indeterminate
coefficients. Then, $\mathfrak{N}^H_{\balpha,q}$ has rank less than $p$ at
$\tilde\bx\in\P^1(\KKGpbar)$ if and only if one of the linear factors
of some diagonal term vanishes at $\tilde \bx$. None of these linear
forms has a projective root at infinity, so we are done.

Consider next a pair $(\balpha,q)$, with
$\balpha=(\alpha_1,\dots,\alpha_p)$ and $1 \le p \le q$ and suppose
that $\assI_2(\balpha',q')$ holds for all $(\balpha',q')$ with
$\balpha'=(\alpha'_1,\dots,\alpha'_{p'})$, $1 \le p' \le q'$, $p' \le p$ and $q'
< q$; we prove that $\assJ_2(\balpha,q)$ holds; as above, this will
imply $\assI_2(\balpha,q)$.

Take $t=p$ in Lemma~\ref{lemma:union}. We first deal with the last
contribution, corresponding to $\bi=(i_1,\dots,i_n)$, and thus
$\kappa=n$: by design, the corresponding point is not at infinity. For
the other contributions, the parameters $(\kappa-(p-t),\balpha_\bi,n-1)$
used in~\eqref{eq:union} are of the form $(\kappa,\balpha_\bi,n-1)$, with
$\balpha_\bi$ of length $\kappa \in \{1,\dots, \min(p,n-1)\}$; since all
conditions $1 \le \kappa \le n-1$, $\kappa \le p$ and $n-1 < q$ are satisfied,
we can invoke the induction assumption. Since the coefficients
$\mathfrak{H}_{\bi,\br}$ are algebraically independent, we deduce that
none of the projective sets
$V_\kappa(\mathfrak{M}^H_{\balpha_\bi,n-1}(\mathfrak{H}_{\bi,\br},\tilde\bX'))$ appearing in
Lemma~\ref{lemma:union} has any point with $X_0=0$. As a consequence,
$V_p(\mathfrak{N}^H_{\balpha,q})$ has no point at infinity either, as
claimed.

%%%%%%%%%%%%%%%%%%%%%%%%%%%%%%%%%%%%%%%%%%%%%%%%%%%%%%%%%%%%

\subsection{Refining $\assI_1$} The following  is a strengthening of 
property $\assI_1$ above. That property asserts that for any
$\tilde\bx$ in $\P^n(\KKGbar)$, the $p \times q$ matrix
$\mathfrak{M}^H_{\balpha,q}(\tilde\bx)$ has rank at least $p-1$, so
that there exists a non-zero $(p-1)$-minor in this matrix.  We claim
that actually, each $(p-1)\times q$ submatrix of
$\mathfrak{M}^H_{\balpha,q}(\tilde\bx)$ has rank $p-1$.

To rephrase this, consider $\balpha=(\alpha_1,\dots,\alpha_p)$ and
$q$, with $1 \le p \le q$, together with a matrix
$\mathfrak{m}^H_{\balpha,q}$, built as $\mathfrak{M}^H_{\balpha,q}$
before, but using products of homogeneous linear forms in $(n-1)+1=q-p+1$
variables $X_0,\dots,X_{n-1}$, instead of $n+1$ variables
$X_0,\dots,X_n$. Such a matrix takes the form
\begin{align}\label{eq:matM2}
\mathfrak{m}^H_{\balpha,q}=\left [\begin{matrix}
\mathfrak{g}^H_{1,1} & \cdots & \mathfrak{g}^H_{1,q}\\
 \vdots & & \vdots\\
\mathfrak{g}^H_{p,1} & \cdots & \mathfrak{g}^H_{p,q}
  \end{matrix}\right ]\in \KK[\mathfrak{G}][X_0,\dots,X_{n-1}]^{p\times q},
\end{align}
with 
$$\mathfrak{g}^H_{i,j,k} = \frak{g}_{i,j,k,0}X_0 + \frak{g}_{i,j,k,1} X_1 +\cdots + \frak{g}_{i,j,k,n-1} X_{n-1},$$
and
$$\mathfrak{g}^H_{i,j} = \mathfrak{g}^H_{i,j,1} \cdots
\mathfrak{g}^H_{i,j,\alpha_i} \in \KK[\mathfrak{G}][X_0,\dots,X_{n-1}],$$
where $\mathfrak{G}=(\mathfrak{g}_{i,j,k,\ell})$ are indeterminates,
for $i=1,\dots,p$, $j=1,\dots,q$, $k=1,\dots,\alpha_i$ and
$\ell=0,\dots,n-1$; we let ${\cal B}=q n (\alpha_1+\cdots +\alpha_p)$ be the total
number of coefficients $\mathfrak{g}_{i,j,k,\ell}$ involved. In this context, the
following property could be proved by induction as in the other cases,
but a direct proof is available.

\begin{description}[leftmargin=*]
\item[$\assI_3(\balpha,q).$] The projective algebraic set
  $V_p(\mathfrak{m}^H_{\balpha,q}) \subset \P^{n-1}(\KKCbar)$ is empty.
\end{description}
To prove this property, take $\balpha=(\alpha_1,\dots,\alpha_p)$ and
$q$ as above. If $q=p$, we have $n=1$, so the $(i,j)$ entry of
$\mathfrak{m}^h_{\balpha,q}$ has the form
$\mathfrak{g}_{i,j,1,0}\cdots\mathfrak{g}_{i,j,\alpha_i,0}
X_0^{\alpha_i}$; hence, the determinant of this matrix is non-zero,
and the claim follows.

We can thus suppose $q > p$, so that $q-1 \ge p$.  Then, the
$((1,\dots,p),(1,\dots,q-1))$-submatrix of $\mathfrak{m}^H_{\balpha,q}$ is
of the form $\mathfrak{M}^H_{\balpha,q-1}$, with entries depending on
${\cal A}(\balpha,q-1)$ parameters.  Let $(c_i)_{i \in I}$ be the $p$-minors
of $\mathfrak{m}^H_{\balpha,q}$ built by taking $p-1$ of the first $q-1$
columns of $\mathfrak{m}^H_{\balpha,q}$, together with its last column.
Any such minor can be expanded along the last column as $c_i =
\mathfrak{g}^H_{1,q} c_{i,1} + \cdots + \mathfrak{g}^H_{p,q}
c_{i,p}$, where $\mathfrak{g}^H_{1,q},\dots,\mathfrak{g}^H_{p,q}$ are
the entries of the last column, and $c_{i,1},\dots,c_{i,p}$
are $(p-1)$-minors from $\mathfrak{M}^H_{\balpha,q-1}$. Remark that
$(c_{i,j})_{i \in I, 1 \le j \le p}$ are {\em all} $(p-1)$-minors
of $\mathfrak{M}^H_{\balpha,q-1}$ (if $p=1$, we have $I=\{1\}$ and
$c_1=\mathfrak{g}^H_{1,q}$, with $c_{1,1}=1$).

By $\assI_2(\balpha,q-1)$, we deduce that
$V_p(\mathfrak{M}^H_{\balpha,q-1}) \subset \P^{n-1}(\KKCbar)$ is
finite. For all other points $\tilde\bx$ in $\P^{n-1}(\KKCbar)$,
$\mathfrak{M}^H_{\balpha,q-1}$ has full rank $p$ at $\tilde\bx$, and
thus so does $\mathfrak{m}^H_{\balpha,q}$. Hence, we can focus on the
points in $V_p(\mathfrak{M}^H_{\balpha,q-1})$.  Consider a point
$\tilde\bx$ in this set; in particular, by $\assI_2(\balpha,q-1)$, we
can take its first coordinate $x_0$ equal to $1$. Using
$\assI_1(\balpha,q-1)$, together with our remark on the $(p-1)$-minors
of $\mathfrak{M}^H_{\balpha,q-1}$, we deduce that not all minors
$(c_{i,j})_{i \in I, 1 \le j \le p}$ vanish at $\tilde\bx$. Suppose
thus that $c_{i_0,j_0}(\tilde\bx) \ne 0$; we prove that
$c_{i_0}(\tilde\bx) \ne 0$, which is enough to conclude.

Let us split the ${\cal B}$ indeterminates $\mathfrak{G}$ into
$\mathfrak{G}_1$ and $\mathfrak{G}_2$, where $\mathfrak{G}_1$ has
cardinality ${\cal B}_1={\cal A}(\balpha,q-1)$ and corresponds to the coefficients
used in the entries
$\mathfrak{g}^H_{1,1},\dots,\mathfrak{g}^H_{p,q-1}$ in
$\mathfrak{M}^H_{\balpha,q}$, and $\mathfrak{G}_2$ of cardinality
${\cal B}_2={\cal B}-{\cal B}_1$ stands for the coefficients of the entries
$\mathfrak{g}^H_{1,q},\dots,\mathfrak{g}^H_{p,q}$ in the last column
of $\mathfrak{m}^H_{\balpha,q}$.  Let us further
write $$c_{i_0}(\tilde\bx)= \mathfrak{g}^H_{1,q}(\tilde\bx)
c_{i_0,1}(\tilde\bx) + \cdots + \mathfrak{g}^H_{p,q}(\tilde\bx)
c_{i_0,p}(\tilde\bx).$$ Since $V_p(\mathfrak{M}^H_{\balpha,q-1})$ is
finite, the coordinates of $\tilde\bx$ are algebraic over
$\KK(\mathfrak{G}_1)$.  Thus, since $x_0=1$, the polynomial
$\mathfrak{g}^H_{j_0,q}(\tilde\bx)\in
\overline{\KK(\mathfrak{G}_1)}[\mathfrak{G}_2]$ admits
$\mathfrak{g}_{j_0,q,1,0}\cdots \mathfrak{g}_{j_0,q,\alpha_{j_0},0}$ as a
specialization, by setting to zero all coefficients
$\mathfrak{g}_{j_0,q,k,\ell}$, for $k=1,\dots,\alpha_{j_0}$ and
$\ell=1,\dots,n-1$ (remark that these coefficients belong to $\mathfrak{G}_2$).  For $j \ne
j_0$, $\mathfrak{g}^H_{j,q}(\tilde\bx)\in
\overline{\KK(\mathfrak{G}_1)}[\mathfrak{G}_2]$ admits $0$ as a
specialization, by setting to zero all coefficients
$\mathfrak{g}_{j,q,k,\ell}$, for $k=1,\dots,\alpha_j$ and
$\ell=0,\dots,n-1$ (again, these coefficients belong to $\mathfrak{G}_2$).

The coefficients $c_{i_0,j}(\tilde\bx)$ are algebraic over
$\KK(\mathfrak{G}_1)$, so that $c_{i_0}(\tilde\bx)$ is in
$\overline{\KK(\mathfrak{G}_1)}[\mathfrak{G}_2]$. By the previous 
discussion, it admits
$$ \mathfrak{g}_{j_0,q,1,0}\cdots \mathfrak{g}_{j_0,q,\alpha_{j_0},0} c_{i_0,j_0}(\tilde\bx)$$ as a
specialization, which is non-zero. Thus,  $c_{i_0}(\tilde\bx)$ 
is non-zero, as claimed.

\subsection{Multiplicity of the solutions} 
The following is the last property we prove for matrices 
$\mathfrak{M}^H_{\balpha,q}$ and $\mathfrak{N}^H_{\balpha,q}$.
Again, we take parameters $\balpha=(\alpha_1,\dots,\alpha_p)$ and $q$,
with $1 \le p \le q$, and we write ${\cal A}={\cal A}(\balpha,q)$ and
${\cal A}'={\cal A}'(\balpha,q)$; we will establish the following.
\begin{description}[leftmargin=*]
\item[$\assI_4(\balpha,q).$] The Jacobian matrix of
  the $p$-minors of $\mathfrak{M}^H_{\balpha,q}$ with respect to
  $\tilde\bX=(X_0,\dots,X_n)$ has rank $n$ at all points in
  $V_p(\mathfrak{M}^H_{\balpha,q})$.
\item[$\assJ_4(\balpha,q).$] The Jacobian matrix of
  the $p$-minors of $\mathfrak{N}^H_{\balpha,q}$ with respect to
  $\tilde\bX=(X_0,\dots,X_n)$ has rank $n$ at all points in
  $V_p(\mathfrak{N}^H_{\balpha,q})$.
\end{description}
As for other proofs involving both $\mathfrak{M}^H_{\balpha,q}$ and
$\mathfrak{N}^H_{\balpha,q}$, we first show that $\assJ_4(\balpha,q)$
implies $\assI_4(\balpha,q)$.

We fix $\balpha$ and $q$, and we assume that $\assJ_4(\balpha,q)$
holds. Consider the ideal of the polynomial ring
$\KK[\mathfrak{L},\tilde\bX]$ in ${\cal A}+n+1$ variables generated by the
$p$-minors of $\mathfrak{M}^H_{\balpha,q}$, together with the
$n$-minors of the Jacobian matrix of these equations with respect to
$(X_0,\dots,X_n)$. This ideal defines an algebraic set
$Z''_{\balpha,q}$ in $\KKbar{}^{\cal A} \times \P^n(\KKbar)$, and we let
$\Delta''_{\balpha,q} \subset \KKbar{}^{\cal A}$ be its projection on the
first factor. By construction, for $\Lambda$ in
$\KKbar{}^{\cal A}-\Delta''_{\balpha,q}$, the Jacobian matrix of
$M_p(\mathfrak{M}^H_{\balpha,q}(\Lambda,\tilde\bX))$ has rank $n$ at
any $\tilde\bx$ in
$V_p(\mathfrak{M}^H_{\balpha,q}(\Lambda,\tilde\bX))$. As before,
because the source is a projective space, $\Delta''_{\balpha,q}$ is closed
(so its complement is open), and we just have to verify that it is not
equal to the whole $\KKbar{}^{\cal A}$. This follows from property
$\assJ_4(\balpha,q)$, which proves that generic matrices of the form
$\mathfrak{N}^H_{\balpha,q}$ do not belong to $\Delta''_{\balpha,q}$.

Again, we finish the proof by induction. We first take $p=q$, and we
prove that $\assJ_4(\balpha,q)$ holds ($\assI_4(\balpha,q)$ will
follow, by the previous paragraph). In this case, $n=1$ and
$\mathfrak{N}^H_{\balpha,q}$ is a diagonal matrix, whose diagonal
entries are products of homogeneous linear forms
$\mathfrak{l}^H_{i,i}$ depending on $(X_0,X_1)$ and with indeterminate
coefficients. The ideal $I_p(\mathfrak{N}^H_{\balpha,q})$ is generated
by the product of the terms $\mathfrak{l}^H_{i,i}$, which admits no
repeated factors; the conclusion follows.

Consider next a pair $(\balpha,q)$, with
$\balpha=(\alpha_1,\dots,\alpha_p)$ and $1 \le p \le q$ and suppose
that $\assI_4(\balpha',q')$ holds for all $(\balpha',q')$ with
$\balpha'=(\alpha'_1,\dots,\alpha'_{p'})$, $1 \le p' \le q'$, $p' \le
p$ and $q' < q$; we prove that $\assJ_4(\balpha,q)$ holds; this will
imply $\assI_4(\balpha,q)$.

We take $t=p$ in the formula of Lemma~\ref{lemma:union}, and we first
deal with the terms in~\eqref{eq:union}.  Thus, we choose a
subsequence $\bi=(i_1,\dots,i_\kappa)$ of $(1,\dots,p)$, with $1 \le \kappa\le
\min(p,n-1)$, and indices $\br=(r_1,\dots,r_\kappa)$, with $ 1\le r_k \le
\alpha_{i_k}$ for all $k$. We prove that the Jacobian matrix of the $p$-minors of $\mathfrak{N}^H_{\balpha,q}$
with respect to $\tilde\bX=(X_0,\dots,X_n)$ has rank $n$ at all points $\tilde\bx=(x_0,\dots,x_n)$ of
$V_p(\mathfrak{N}^H_{\balpha,q})$  such that
$\tilde\bx'=(x_0,\dots,x_{n-\kappa})$ is in
$V_\kappa(\mathfrak{M}^H_{\balpha_\bi,n-1}(\mathfrak{H}_{\bi,\br},\tilde\bX')) \subset
\P^{n-\kappa}(\KKGpbar)$, and such that
\begin{align}\label{eq:subsX}
  x_{n-\kappa+1}=\mathfrak{f}_{n-\kappa+1,\bi,\br}(\tilde\bx'),\dots,x_{n}=\mathfrak{f}_{n,\bi,\br}(\tilde\bx').
\end{align}
By Lemma~\ref{lemma:union}, taking all such $\tilde\bx$ into account,
for all $\bi$ and $\br$, will cover all points in
$V_p(\mathfrak{N}^H_{\balpha,q})$, up to the exception of those points
obtained from $\kappa=n$, which will admit a simpler treatment.
For simplicity, we continue the proof with $\bi=(1,\dots,\kappa)$, so
that we have $\balpha_\bi=(\alpha_1,\dots,\alpha_\kappa)$.  

We are going to exhibit some polynomials that belong to
$I_p(\mathfrak{N}^H_{\balpha,q})$, for which we can control the rank
of the Jacobian at $\tilde\bx$. First, we prove that for $i$ in
$\{1,\dots,\kappa\}$ and $r$ in $\{1,\dots,\alpha_i\}-\{r_i\}$, as
well as $i$ in $\{\kappa+1,\dots,p\}$ and $r$ in
$\{1,\dots,\alpha_i\}$, the value $\mathfrak{l}^H_{i,i,r}(\tilde\bx)$
is non-zero.  We subdivide the indeterminates $\mathfrak{L}'$ into
$\mathfrak{L}'_{\bi,\br}$ and $\mathfrak{L}''_{\bi,\br}$, where
$\mathfrak{L}'_{\bi,\br}$ corresponds to the coefficients involved in
$\mathfrak{l}^H_{i,i,r_i}$, for $i=1,\dots,\kappa$, and in the
submatrix of $\mathfrak{N}^H_{\balpha,q}$ associated to $\bi$, and
$\mathfrak{L}''_{\bi,\br}$ are the other coordinates.  By
$\assI_2(\balpha_s,n-1)$,
$V_\kappa(\mathfrak{M}^H_{\balpha_\bi,n-1}(\mathfrak{H}_{\bi,\br},\tilde\bX'))$
is finite; as a result, since all entries of $\mathfrak{H}_{\bi,\br}$ 
are in $\KK(\mathfrak{L}'_{\bi,\br})$,
all coordinates of $\tilde\bx$ are algebraic
over $\KK(\mathfrak{L}'_{\bi,\br})$. For $i,r$ as above, the
coefficients of the equation
$$\mathfrak{l}^H_{i,i,r} = \frak{l}_{i,i,r,0} X_0+ \frak{l}_{i,i,r,1}
X_1 +\cdots + \frak{l}_{i,i,r,n} X_n$$ are
in  $\KK(\mathfrak{L}''_{\bi,\br})$, thus algebraically independent
over the field of definition of $\tilde\bx$, so that $\mathfrak{l}^H_{i,i,r}(\tilde\bx)$
is non-zero.

\begin{remark}\label{remark:disjoint}
  This implies in particular that the union in Lemma~\ref{lemma:union} is disjoint.
\end{remark}

In the following two paragraphs, assume $\kappa \ge 2$
and take $i$ in $\{1,\dots,\kappa\}$. We can then define
$\bi^*=(1,\dots,i-1,i+1,\dots,\kappa)$,
$\balpha^*=(\alpha_1,\dots,\alpha_{i-1},\alpha_{i+1},\dots,\alpha_\kappa)$,
and we call $\mathfrak{N}^H_i$ the submatrix of
$\mathfrak{N}^H_{\balpha,q}$ associated to $\bi^*$; this is a matrix
with $\kappa-1$ rows (indexed by $\bi^*$ in
$\mathfrak{N}^H_{\balpha,q}$) and $n-1$ columns (of indices $p+1,\dots,q$ 
in $\mathfrak{N}^H_{\balpha,q}$).

We prove that there exists a $(\kappa-1)$-minor $c_i$ of
$\mathfrak{N}^H_i$ such that $c_i(\tilde\bx)\ne 0$.  Let indeed
$\mathfrak{m}^H_i$ be the matrix obtained by applying the
substitution~\eqref{eq:subsX} in $\mathfrak{N}^H_i$. This matrix 
has $\kappa-1$ rows and $n-1$ columns; its entries are products of
linear forms in $(n-\kappa)+1$ variables $X_0,\dots,X_{n-\kappa}$,
with coefficients that are algebraically
independent over $\KK$. We can thus apply $\assI_3(\balpha^*,n-1)$ to
$\mathfrak{m}_{i}^H$, and deduce that this matrix has full rank $\kappa-1$
at $\tilde\bx'$.  Thus, $\mathfrak{N}^H_i$ has rank $\kappa-1$ at
$\tilde\bx$, from which the existence of the minor $c_i$ follows.
If $\kappa=1$, we define $c_1=1$.

We next deduce that for $i$ in $\{1,\dots,\kappa\}$, there exists a
polynomial of the form $b_{i} \mathfrak{l}^H_{i,i,r_i}$ in the ideal
$I_p(\mathfrak{N}^H_{\balpha,q})$, with $b_{i}(\tilde\bx)\ne
0$. Indeed, we consider the $p$-minor of $\mathfrak{N}^H_{\balpha,q}$
obtained by taking the columns $i$,$\kappa+1,\dots,p$, and all
$\kappa-1$ columns in the $(\kappa-1)$-minor $c_i$ (if $\kappa=1$,
there is no need to consider such columns). Using the factorization
$$\mathfrak{l}^H_{i,i} = \beta_i \mathfrak{l}^H_{i,i,r_i},\quad\text{with}\quad
\beta_i=\mathfrak{l}^H_{i,i,1}\cdots \mathfrak{l}^H_{i,i,r_i-1}\mathfrak{l}^H_{i,i,r_i+1}\cdots \mathfrak{l}^H_{i,i,\alpha_i},$$
that minor evaluates to 
$$b_i \mathfrak{l}^H_{i,i,r_i}\quad\text{with}\quad b_i = \beta_i
\mathfrak{l}^H_{\kappa+1,\kappa+1}\cdots \mathfrak{l}^H_{p,p}c_i.$$ Hence, $b_i\,
\mathfrak{l}^H_{i,i,r_i}$ belongs to $I_p(\mathfrak{N}^H_{\balpha,q})$, and by
the discussion of the three previous paragraphs, $b_i(\tilde\bx)\ne 0$,
as claimed. In what follows, we write $b=b_1 \cdots b_\kappa$,
so that $b(\tilde\bx) \ne 0$ and $b\, \mathfrak{l}^H_{i,i,r_i}$ 
is in $I_p(\mathfrak{N}^H_{\balpha,q})$. 
This in turn implies
that all polynomials
$$
b(X_{n-\kappa+1}-\mathfrak{f}_{n-\kappa+1,\bi,\br}(\tilde\bX')),\dots,b(X_{n}-\mathfrak{f}_{n,\bi,\br}(\tilde\bX'))
$$ are in $I_p(\mathfrak{N}^H_{\balpha,q})$ as well.

Similarly, for every $\kappa$-minor $\eta$ of the submatrix of
$\mathfrak{N}^H_{\balpha,q}$ associated to $\bi$, the
polynomial  $\mathfrak{l}^H_{\kappa+1,\kappa+1}\cdots \mathfrak{l}^H_{p,p}\, \eta$ belongs to
$I_p(\mathfrak{N}^H_{\balpha,q})$. Thus, $b \,\eta$ is in $I_p(\mathfrak{N}^H_{\balpha,q})$
as well.

As a result, the polynomial $b\, \eta(\tilde\bX',\mathfrak{f}_{n-\kappa+1,\bi,\br}(\tilde\bX'),\dots,\mathfrak{f}_{n,\bi,\br}(\tilde\bX'))$
belongs to $I_p(\mathfrak{N}^H_{\balpha,q})$. Now, 
$\gamma=\eta(\tilde\bX',\mathfrak{f}_{n-\kappa+1,\bi,\br}(\tilde\bX'),\dots,\mathfrak{f}_{n,\bi,\br}(\tilde\bX'))$
is one of the $\kappa$-minors of
$\mathfrak{M}^H_{\balpha_\bi,n-1}(\mathfrak{H}_{\bi,\br},\tilde\bX')$,
and all $\kappa$-minors of this matrix are obtained this way.
To summarize, we have proved that 
$$b\, \mathfrak{l}^H_{1,1,r_1},\dots,b\,
\mathfrak{l}^H_{\kappa,\kappa,r_\kappa} \quad\text{and}\quad b\,
\gamma, \text{~for all $\kappa$-minors $\gamma$ of
  $\mathfrak{M}^H_{\balpha_\bi,n-1}(\mathfrak{H}_{\bi,\br},\tilde\bX')$}$$
are in $I_p(\mathfrak{N}^H_{\balpha,q})$, with $b(\tilde\bx) \ne
0$. The Jacobian matrix of these polynomials at $\tilde\bx$ is, up to
the non-zero constant $b(\tilde\bx)$, equal to that of
$\mathfrak{l}^H_{1,1,r_1},\dots,
\mathfrak{l}^H_{\kappa,\kappa,r_\kappa}$ (which is simply a matrix of
constants), and of all $\kappa$-minors $\gamma$. Using our induction
assumption, we know that the Jacobian matrix of the ideal of
$\kappa$-minors $\gamma$ with respect to $\tilde\bX'$ has rank
$n-\kappa$ at $\tilde\bx'$. As a result, the larger Jacobian matrix of
all equations above has rank $n$ at $\tilde\bx$, as claimed.

It remains to deal with the case $\kappa=n$, for $n \le p$; as above,
we may simplify the discussion by assuming that $\bi=(1,\dots,n)$. In
this case, the discussion is simpler: proceeding as above, but dealing
only with the polynomials
$\mathfrak{l}^H_{1,1},\dots,\mathfrak{l}^H_{n,n}$, we obtain the fact
that equations of the form $b\, \mathfrak{l}^H_{1,1,r_1},\dots,b\,
\mathfrak{l}^H_{n,n,r_n}$ belong to $I_p(\mathfrak{N}^H_{\balpha,q})$,
with $b(\tilde\bx) \ne 0$. The conclusion follows directly.

%%%%%%%%%%%%%%%%%%%%%%%%%%%%%%%%%%%%%%%%%%%%%%%%%%%%%%%%%%%%

\subsection{An algorithm}

We conclude this section with an algorithm
$\mathsf{RowDegreeDiagonal}$, that applies the decomposition in
Lemma~\ref{lemma:union}, in the case $t=p$, to non-homogeneous
matrices. Indeed, while homogeneity is used at several steps in the
proof (and will be needed again when we apply this result), our main
algorithm deals with matrices without a homogeneous structure.  Thus
we will consider a matrix $\mN$ as in~\eqref{eqdef:type1}, but with
$X_0=1$. Explicitly, we have
\begin{align}\label{eqdef:type1aff}
\mN= \left( \begin{matrix}
\lambda_{1,1} & 0 & \cdots & 0 & \lambda_{1,p+1} & \cdots & \lambda_{1, q}\\
0 & \lambda_{2,2} & \cdots & 0 & \lambda_{2,p+1} & \cdots & \lambda_{2, q}\\
\vdots & \vdots & \ddots & \vdots & \vdots & \ddots & \vdots\\
0 & 0 & \cdots & \lambda_{p,p} & \lambda_{p,p+1} & \cdots & \lambda_{p, q}
\end{matrix} \right),
\end{align}
where for all $i,j$, $\lambda_{i,j}$ is the product of $\alpha_i$
linear forms $(\lambda_{i,j,k})_{1 \le k \le \alpha_i}$ with
coefficients in $\KK$, in variables $X_1,\dots,X_n$.  By
Proposition~\ref{lemma:appendix}, we deduce that for a generic choice
of the coefficients of these linear forms,
$V_p(\mN) \subset \KKbar{}^n$ is a finite set, whose structure is
given by Lemma~\ref{lemma:union}. Besides, oberve that all points of
$V_p(\mN)$ are simple and isolated (this is an immediate consequence of
the second assertion of Proposition~\ref{lemma:appendix}).

Algorithm $\mathsf{RowDegreeDiagonal}$ below takes as input the linear
forms $(\lambda_{i,j,k})$ and computes a zero-dimensional
parametrization of $V_p(\mN)$. In the algorithm, we assume the
existence of a subroutine $\mathsf{RowDegree\_simple}(\Gamma)$ which
takes as input a straight-line program $\Gamma$ that computes a
polynomial matrix $\mF$ and a system of equations $G$, and solves
Problem~\eqref{problem3} for this input using a row-degree homotopy. We
give such an algorithm in the next section. We denote by
$T_{\rm row}(\sigma,\bgamma,\balpha,q)$ the time spent by
$\mathsf{RowDegree\_simple}(\Gamma)$ on input a straight-line program
of length $\sigma$ that computes $\mF$ with row degrees
$\balpha=(\alpha_1,\dots,\alpha_p)$ and $q$ columns, and
$G=(g_1,\dots,g_s)$ of degrees $\bgamma=(\gamma_1,\dots,\gamma_s)$.

We will be use here the particular case of subroutine
$\mathsf{RowDegree\_simple}$ where the input matrix has the form
\begin{align}\label{eqdef:type2aff}
  \mM= \left( \begin{matrix}
    \lambda_{1,1} & \lambda_{1,2} & \cdots & \lambda_{1, q}\\
    \lambda_{2,1} &  \lambda_{2,2} & \cdots & \lambda_{2, q}\\
    \vdots & & & \vdots\\
    \lambda_{p,1} &  \lambda_{p,2}& \cdots & \lambda_{p, q}
  \end{matrix} \right),
\end{align}
where for all $i,j$, $\lambda_{i,j}$ is the product of $\alpha_i$ (non
necessarily homogeneous) linear forms $(\lambda_{i,j,k})_{1 \le k \le
  \alpha_i}$ in $n$ variables $X_1,\dots,X_n$. In this case, each
entry $\lambda_{i,j}$ can be computed in $O(n \alpha_i)$ operations in
$\KK$, so that the whole matrix $\mM$ can be computed by a
straight-line program of length $O(nq(\alpha_1+\cdots+\alpha_p))$. In
this case there are no additional equations $G$, so we denote the cost
of Algorithm {\sf RowDegree\_simple} for such input by
\begin{align}
T_{M,{\rm row}}(\balpha,q)=T_{\rm row}(nq(\alpha_1+\cdots+\alpha_p),(),\balpha,q).
\end{align}
We conclude this section with the cost 
analysis of Algorithm $\mathsf{RowDegreeDiagonal}$.
\begin{algorithm}[!t]
\caption{$\mathsf{RowDegreeDiagonal}((\lambda_{i,j,k})_{i,j,k})$}
{\bf Input}: linear forms $(\lambda_{i,j,k})_{i,j,k}$ making up 
the entries of 
$\mN \in \KK[X_1, \ldots, X_n]^{p \times q}$ as in~\eqref{eqdef:type1aff}, with $p \leq q$ and $n = q-p+1$\\
{\bf Output}: a zero-dimensional parametrization $\scrR$ of $V_p(\mN)$
\begin{enumerate}
\item for any subsequence $\bi = (i_1, \ldots, i_\kappa)$ of $(1, \ldots, p)$ with $1 \leq \kappa \leq\min(n-1,p)$
  \begin{enumerate}
  \item for any sequence $\br = (r_1, \ldots, r_\kappa)$, with $r_k$ in 
    $\{1,\dots,\alpha_k\}$ for all $k$
    \begin{enumerate}
    \item apply Gaussian elimination to the system 
      $\lambda_{i_1,i_1,r_1}=\dots=\lambda_{i_\kappa,i_\kappa,r_\kappa}=0$
      to rewrite $(X_{n-\kappa+1}, \ldots, X_n)$ as linear forms $(f_{j,\bi,\br})_{n-\kappa+1 \le j \le n}$ in $(X_1, \ldots, X_{n-\kappa})$. 

\hfill $\text{\sf{cost:~}} O(\sum_{\bi,\br} n^3)$
    \item\label{step:constSLP} construct a straight-line program $\Gamma_{\bi,\br}$ that computes the matrix $\mM_{\bi,\br} \in \KK[X_1, \dots, X_{n-\kappa}]^{\kappa \times (n-1)}$ obtained
      by substituting $(f_{j,\bi,\br})_{n-\kappa+1 \le j \le n}$ into $\mN_{\bi,(p+1, \ldots, q)}$. The length of $\Gamma_{\bi,\br}$ is $O(\kappa n(\alpha_{i_1}+\cdots+\alpha_{i_\kappa}))$.

\hfill $\text{\sf{cost:~}} O(\sum_{\bi,\br} (\alpha_{i_1} + \cdots + \alpha_{i_\kappa}) n^3)$
    \item $\scrR_{\bi,\br}' \gets$ $\mathsf{RowDegree\_simple}(\Gamma_{\bi,\br})$  (points have coordinates $(X_1, \ldots, X_{n-\kappa})$)

\hfill $\text{\sf{cost:~}} \sum_{\bi,\br} T_{M,{\rm row}}((\alpha_{i_1},\dots,\alpha_{i_\kappa}), n-1)$
    \item\label{step:substdiag} deduce $\scrR_{\bi,\br}$ from $\scrR_{\bi,\br}'$ by adding the expressions for $(X_{n-\kappa+1}, \ldots, X_n)$

\hfill      $\text{\sf{cost:~}} O(\sum_{\bi,\br}  \rc_{\bi,\br} n^2)$
  \end{enumerate}
  \end{enumerate}

\item if $n \le p$, for any subsequence $\bi = (i_1, \ldots, i_n)$ of $(1, \ldots, p)$
  \begin{enumerate}
  \item for any sequence $\br = (r_1, \ldots, r_n)$, with $r_k \in \{1,\dots,\alpha_k\}$ for all $k$
    \begin{enumerate}
    \item let $\bx_{\bi,\br}$ be the solution of the system $\lambda_{i_1,i_1,r_1}=\dots=\lambda_{i_n,i_n,r_n}=0$

\hfill      $\text{\sf{cost:~}} O(\sum_{\bi,\br}n^3)$
    \item create a zero-dimensional parametrization $\scrR_{\bi,\br}$ such that $Z(\scrR_{\bi,\br})=\{\bx_{\bi,\br}\}$

\hfill      $\text{\sf{cost:~}} O(\sum_{\bi,\br}n)$
  \end{enumerate}
\end{enumerate}
\item combine all $(\scrR_{\bi,\br})_{\bi,\br}$ into the output $\scrR$

\hfill  $\text{\sf{cost:~}} \softO(\sum_{\bi,\br}  \rc_{\bi,\br} n)$
\end{enumerate}
\label{Row}
\end{algorithm}

\begin{lemma}\label{lemma:rowdegreediagonal}
  Let $S_n(\alpha_1,\dots,\alpha_p)$ be the degree $n$ complete
  symmetric function of $(\alpha_1,\dots,\alpha_p)$.  The total cost of
  $\mathsf{RowDegreeDiagonal}((\lambda_{i,j,k})_{i,j,k})$
  is $$\sum_{\substack{\bi=(i_1,\dots,i_\kappa)\\ \kappa\le\min(n-1,p)}}
  \alpha_{i_1} \cdots \alpha_{i_\kappa} T_{M,{\rm
      row}}((\alpha_{i_1},\dots,\alpha_{i_\kappa}), n-1) + \softO\left (n^3(\rc + S_n(\alpha_1,\dots,\alpha_p))  \right),$$
  where $\rc$ is the cardinality of $V_p(\mN)$.
\end{lemma}
\begin{proof}
  The cost reported at each step in the pseudo-code is the total amount
  of time spent there, over all iterations (the sums in the first loop
  are for $\kappa \le \min(n-1,p)$, the ones in the second loop for
  $\kappa=n$ if $n \le p$). Several steps are straightforward to
  analyze; we briefly comment on a few others.

  Step~\ref{step:constSLP} uses the linear forms
  $(f_{j,\bi,\br})_{n-\kappa+1 \le j \le n}$ to construct a
  straight-line program $\Gamma_{\bi,\br}$ that computes the entries
  of $\mM_{\bi,\br}$. This is done by computing the coefficients
  of the linear forms in $(X_1,\dots,X_{n-\kappa})$ obtained after 
  substitution. Each linear form requires a matrix-vector product 
  with a matrix of size $(n-\kappa) \times n$, for $O(n^2)$ operations,
  whence a total of $O((\alpha_{i_1} + \cdots +\alpha_{i_\kappa})n^3)$
  for all entries.

  Step~\ref{step:substdiag} consists in adding $\kappa$ coordinates
  $(X_{n-\kappa+1},\dots,X_n)$ to a zero-dimensional parametrization
  in variables $X_1,\dots,X_{n-\kappa}$, where
  $(X_{n-\kappa+1},\dots,X_n)$ are known as linear forms
  $(f_{j,\bi,\br})_{n-\kappa+1 \le j \le n}$ in $(X_1, \ldots,
  X_{n-\kappa})$: this is done by means of a matrix product in size
  $(\kappa \times n-\kappa)$ by $(n-\kappa \times \rc_{\bi,\br})$, where
  $\rc_{\bi,\br}$ is the cardinality of $V_\kappa(\mM_{\bi,\br})$, for
  $\bi=(i_1,\dots,i_\kappa)$.  The cost is thus $O(\rc_{\bi,\br}n^2 )$;
  the sum of these costs is thus $O(\rc n^2 )$, since the sum of all
  $\rc_{\bi,\br}$ is equal to $\rc$ by Remark~\ref{remark:disjoint}.

  The combination in the last step is done by fast Chinese Remaindering,
  in quasi-linear time $\softO(\sum_{\bi,\br} \rc_{\bi,\br}  n)$,
  which is $\softO(\rc n)$. Thus, the total runtime is 
  \begin{multline*}
 \sum_{\substack{\bi=(i_1,\dots,i_\kappa)\\ \br=(r_1,\dots,r_\kappa)\\\kappa\le\min(n-1,p)}}
T_{M,{\rm row}}((\alpha_{i_1},\dots,\alpha_{i_\kappa}), n-1)\\
+\softO\left (\rc n^2  + \sum_{\substack{\bi=(i_1,\dots,i_\kappa)\\ \br=(r_1,\dots,r_\kappa)\\\kappa\le\min(n-1,p)}}
  (\alpha_{i_1} + \cdots +\alpha_{i_\kappa}) n^3 
  +  \sum_{\substack{\bi=(i_1,\dots,i_n)\\ \br=(r_1,\dots,r_n)}} n^3\right).    
  \end{multline*}
  The costs reported in the sums do not depend on $\br$, so that 
  this can be rewritten as 
  \begin{multline*}
 \sum_{\substack{\bi=(i_1,\dots,i_\kappa)\\ \kappa\le\min(n-1,p)}}
  \alpha_{i_1}  \cdots \alpha_{i_\kappa}T_{M,{\rm row}}((\alpha_{i_1},\dots,\alpha_{i_\kappa}), n-1)\\
+\softO\left ( \rc n^2 + \sum_{\substack{\bi=(i_1,\dots,i_\kappa)\\ \kappa\le\min(n-1,p)}}
    \alpha_{i_1}  \cdots \alpha_{i_\kappa} (\alpha_{i_1} + \cdots +\alpha_{i_\kappa}) n^3 
  +  \sum_{\substack{\bi=(i_1,\dots,i_n)}   }\alpha_{i_1}  \cdots \alpha_{i_n} n^3\right).    
  \end{multline*}
  The final simplification comes from noting that $\sum_{\bi}
  \alpha_{i_1} \cdots \alpha_{i_\kappa}(\alpha_{i_1} + \cdots
  +\alpha_{i_\kappa})$, for $\bi$ a subsequence of $(1,\dots,p)$ of
  length $\kappa\le\min(n-1,p)$, is bounded from above by
  $S_n(\alpha_1,\dots,\alpha_p)$. The same holds for the second
  sum (which is empty if $n >p$).
\end{proof}

%%%%%%%%%%%%%%%%%%%%%%%%%%%%%%%%%%%%%%%%%%%%%%%%%%%%%%%%%%%%
%%%%%%%%%%%%%%%%%%%%%%%%%%%%%%%%%%%%%%%%%%%%%%%%%%%%%%%%%%%%
%%%%%%%%%%%%%%%%%%%%%%%%%%%%%%%%%%%%%%%%%%%%%%%%%%%%%%%%%%%%

\section{The row-degree homotopy}\label{sec:rowdegree}

We now give algorithms to solve Problems~\eqref{problem2}
and~\eqref{problem3} whose runtime will depend on the row-degrees of the
input matrix $\mF$. These algorithms are more complex than the ones in
Section~\ref{sec:columndegree}, due to their recursive nature. This
boils down to the fact that the start system we use for the homotopy
must itself be solved by means of several homotopies of smaller size,
along the lines of the discussion in the previous section.

Again, we are given a matrix
$\mF =[f_{i,j}]\in \KK[X_1,\dots,X_n]^{p \times q}$ and polynomials
$G=(g_1,\dots,g_s)$ in $\KK[X_1,\dots,X_n]$, with $p \leq q$ and
$n = q-p+s+1$, and we want to compute the isolated points (or the
simple points) of $\VpFG{p}{\mF}{G}$, with
$$\VpFG{p}{\mF}{G} = \{\bx \in \KKbar{}^n \mid \mathrm{rank}(\mF({\bx})) < p
\text{~and~} g_1(\bx)=\cdots=g_s(\bx)=0\}.$$

We are now interested in designing algorithms for computing the
isolated points or the simple points of $\VpFG{p}{\mF}{G}$
whose cost depends on the row degrees
$\alpha_1=\rdeg(\mF,1),\dots,\alpha_p=\rdeg(\mF,p)$; with this
notation, $\deg(f_{i,j}) \leq \alpha_i$ holds for all $i,j$. As in
Section~\ref{sec:columndegree}, we write
$\gamma_1=\deg(g_1),\dots,\gamma_s=\deg(g_s)$ and we let
$\alpha = \max(\alpha_1, \ldots, \alpha_p)$ and
$\gamma = \max(\gamma_1, \ldots, \gamma_s)$.
We start by stating our first result on computing the isolated
points of $\VpFG{p}{\mF}{G}$. Recall in what follows that 
where $S_{n-s}$ is the complete homogeneous symmetric 
function of degree $n-s$.

\begin{proposition}\label{prop:rowdegree}
  Suppose that the matrix $\mF \in \KK[X_1,\dots,X_n]^{p \times q}$
  and the polynomials $G=(g_1,\dots,g_s)$ in $\KK[X_1,\dots,X_n]$ are
  given by a straight-line program of length $\sigma$.  Then, the multiplicities of
  the isolated points of $\VpFG{p}{\mF}{G}$ are at most
  $\rc=\gamma_1\cdots\gamma_s S_{n-s}(\alpha_1, \ldots, \alpha_p)$.
  
  Assume that  all $\gamma_i$'s and $\alpha_j$'s are at least equal to $1$, 
  and let   $\re=(\gamma_1+1)\cdots(\gamma_s+1) S_{n-s}(\alpha_1+1, \ldots,
  \alpha_p+1)$, $\alpha = \max(\alpha_1, \ldots, \alpha_p)$ and
  $\gamma = \max(\gamma_1, \ldots, \gamma_s)$. Then, there
  exists a randomized algorithm that computes the isolated points of
  $\VpFG{p}{\mF}{G}$ using
  $$\softO\left( {q \choose p} \rc (\re+\rc^5 )  (\sigma + \gamma+ p \alpha)  \right)$$
  operations in $\KK$.
\end{proposition}

We state now the complexity result for computing the simple
points of $\VpFG{p}{\mF}{G}$.

\begin{proposition}\label{prop:rowdegree_simple}
  Reusing the notations introduced above, there
  exists a randomized algorithm that computes the simple
  points of $\VpFG{p}{\mF}{G}$ using
  $$\softO\left({q \choose p}  \rc \re       (\sigma + \gamma+ p \alpha)  \right)$$
  operations in $\KK$.
\end{proposition}

These propositions complete the proofs of Theorems~\ref{theo:1},
\ref{theo:2} and~\ref{theo:3}. 

%%%%%%%%%%%%%%%%%%%%%%%%%%%%%%%%%%%%%%%%%%%%%%%%%%%%%%%%%%%%

\subsection{Setting up the homotopy}

We are again going to rely on the algorithm of
Section~\ref{sec:homotopy}. As in Section~\ref{sec:columndegree}, we
let $\bC=(c_1,\dots,c_s,c_{s+1}\dots,c_m)$ be such that
$(c_1,\dots,c_s)=(g_1,\dots,g_s)$ and $(c_{s+1},\dots,c_m)$ are the
$p$-minors of $\mF$. Our main concern is to design a sequence of
polynomials $\bB=(b_1,\dots,b_s,\dots,b_m)$ in $\KK[T,\bX]$ such that
$\bC=\bB_1$, such that we can solve efficiently the system
$\bA=\bB_0$, and such that $\bB$ has the same degree profile as our
target system $\bC$.

The polynomials $(b_1,\dots,b_s)$ are defined as in Section~\ref{sec:columndegree}, letting $a_i$ be a product of $\gamma_i$ linear forms
$\mu_{i,k}$ with randomly chosen coefficients, of the form
\begin{align}\label{eqdef:ai}
a_i=\prod_{k=1}^{\gamma_i} \mu_{i,k},\quad\text{with}\quad
\mu_{i,k}=\mu_{i,k,0} + \sum_{\ell=1}^n \mu_{i,k,\ell}X_\ell
\end{align}
and writing $b_i=(1-T)a_i + T g_i$ for $i=1,\dots,s$. The difference
will lie in the construction of the start matrix used in the homotopy. The
construction  presented in
Section~\ref{sec:columndegree} does not carry over if we want to take
row degrees into account. Instead, we use a deformation that cancels
out many off-diagonal terms; following the construction in the
previous section, we define $\mN$ as in~\eqref{eqdef:type1aff}, that is
\begin{align*}
\mN = \left( \begin{matrix}
\lambda_{1,1} & 0 & \cdots & 0 & \lambda_{1,p+1} & \cdots & \lambda_{1, q}\\
0 & \lambda_{2,2} & \cdots & 0 & \lambda_{2,p+1} & \cdots & \lambda_{2, q}\\
\vdots & \vdots & \ddots & \vdots & \vdots & \ddots & \vdots\\
0 & 0 & \cdots & \lambda_{p,p} & \lambda_{p,p+1} & \cdots & \lambda_{p, q}
\end{matrix} \right), 
\end{align*} 
where for all $i,j$, $\lambda_{i,j}$ is a product of $\alpha_i$ linear
forms with random coefficients in $\KK$, of the form
$$\lambda_{i,j}= \prod_{k=1}^{\alpha_i}\lambda_{i,j,k},
\quad\text{with}\quad
\lambda_{i,j,k} =\lambda_{i,j,k,0} + \sum_{\ell=1}^n \lambda_{i,j,k,\ell}X_\ell.
$$ Then, we define $(b_{s+1},\dots,b_m)$ as the
$p$-minors of $\mU=(1-T)\cdot\mN + T\cdot \mF$, and $\bB=(b_1,\dots,b_m)$.
The polynomials $(a_{s+1},\dots,a_m)$ are defined as the $p$-minors of
$\mN$, so that $\bA=\bB_0$; on the other hand, we also have
$\bC=\bB_1$.  Our next step is to prove that all assumptions of
Propositions~\ref{prop:degree_fiber}~\ref{prop:compute_isolated}
and~\ref{prop:compute_regular}
are satisfied for $\bB$ and $\bA=\bB_0$, as long as the coefficients 
of $a_1,\dots,a_s$ and $\mN$ are chosen generically.

\paragraph{Properties $\assA_1$ and $\assA_2$}
These follow from Proposition~\ref{prop:KH1H2}.

\paragraph{Property $\assG_1(0)$.} We have to prove that for $i=1,\dots,m$,
$\deg_\bX(b_i)=\deg_\bX(a_i)$. We already established it in
Section~\ref{sec:columndegree} for indices $i=1,\dots,s$. For
$i=s+1,\dots,m$, we can readily see that the degree of $b_i$ in $\bX$ is at most
$\alpha_1 + \cdots + \alpha_p$, so it suffices to prove that the
degree of all $p$-minors $(a_{s+1},\dots,a_m)$ of $\mN$ is $\alpha_1 +
\cdots + \alpha_p$.

Indeed, any $p$-minor of $\mN$ is of the form $\lambda_{i_1,i_1}
\cdots \lambda_{i_\kappa,i_\kappa} \zeta$, for some sequence
$\bi=(i_1,\dots,i_\kappa) \subset (1,\dots,p)$ of length $\kappa \in
\{0,\dots,p\}$ and some $(p-\kappa)$-minor $\zeta$ of $\mN_{\bi
,(p+1,\dots,q)}$.  Since the entries of $\mN_{\bi
,(p+1,\dots,q)}$ are products of linear form with randomly
chosen coefficients $(\lambda_{i,j,k,\ell})$, for a generic choice of
these coefficients, the determinant $\zeta$ has degree
$\sum_{i' \notin \bi} \alpha_{i'}$. As a result, the corresponding
$p$-minor of $\mN$ has degree $\alpha_1 + \cdots + \alpha_p$, as
claimed.

\paragraph{Property $\assG_2(0)$.} Next, we prove that the system $\bA=\bB_0$ has no solution 
at infinity. As in Section~\ref{sec:columndegree}, we introduce a homogenization
variable $X_0$, and we consider the system $\bA^H=(a_1^H,\dots,a_s^H,\dots,a_m^H)$ obtained 
by homogenizing all equations in $\bA$. Thus we have
$$a_i^H=\prod_{k=1}^{\gamma_i} \mu^H_{i,k} \quad\text{with}\quad \mu^H_{i,k}=\mu_{i,k,0}X_0 + \sum_{\ell = 1}^{n}\mu_{i,k,\ell}X_\ell$$
for $i=1,\dots,s$, whereas $a_{s+1}^H,\dots,a_m^H$ are the $p$-minors of the matrix
\begin{align*}
\mN^H = \left( \begin{matrix}
\lambda^H_{1,1} & 0 & \cdots & 0 & \lambda^H_{1,p+1} & \cdots & \lambda^H_{1, q}\\
0 & \lambda^H_{2,2} & \cdots & 0 & \lambda^H_{2,p+1} & \cdots & \lambda^H_{2, q}\\
\vdots & \vdots & \ddots & \vdots & \vdots & \ddots & \vdots\\
0 & 0 & \cdots & \lambda^H_{p,p} & \lambda^H_{p,p+1} & \cdots & \lambda^H_{p, q}
\end{matrix} \right),  
\end{align*}
where $\lambda^H_{i,k}$ is the homogenization of
$\lambda_{i,j}$. (This latter property requires genericity of the
coefficients of the linear forms $\lambda^H_{i,k}$; it is enough that 
each $p$-minor of $\mN$ have degree $\alpha_1 + \cdots + \alpha_p$.)

The solutions of $\bA^H$ in $\P^n(\KKbar)$ are found by first solving
the equations $(a^H_1,\dots,a^H_s)$. As in Section~\ref{sec:columndegree}, all
$a_i^H$ are products of linear forms, so any solution of
$(a^H_1,\dots,a^H_s)$ is obtained by setting some of these linear forms
to zero (at least one for each $i=1,\dots,s$). We choose indices $\bu=(u_1,\dots,u_s)$, with
$u_1\in\{1,\dots,\gamma_1\}$, \dots, $u_s\in\{1,\dots,\gamma_s\}$, and
we solve
$$\mu^H_{i,u_i}=0, \quad \text{~that is,~} \quad \mu_{i,u_i,0}X_0 + \sum_{\ell = 1}^{n}\mu_{i,u_i,\ell}X_\ell =0,$$ for $i=1,\dots,s$.
In what follows, we fix such an $\bu$.
Then, for a generic choice of coefficients $\mu_{i,k,\ell}$, these equations
are equivalent to
$$X_{n-s+1}=\Phi_{n-s+1,\bu}(X_0,\dots,X_{n-s}),\dots,X_{n}=\Phi_{n,\bu}(X_0,\dots,X_{n-s}),$$
for some homogeneous linear forms $\Phi_{n-s+1,\bu},\dots,\Phi_{n,\bu}$.
After applying this substitution, for all $i,j$,
$\mN^H$ can be rewritten as 
\begin{align*}
 \mN^H_\bu = \left( \begin{matrix}
\lambda^H_{1,1,\bu} & 0 & \cdots & 0 & \lambda^H_{1,p+1,\bu} & \cdots & \lambda^H_{1, q,\bu}\\
0 & \lambda^H_{2,2,\bu} & \cdots & 0 & \lambda^H_{2,p+1,\bu} & \cdots & \lambda^H_{2, q,\bu}\\
\vdots & \vdots & \ddots & \vdots & \vdots & \ddots & \vdots\\
0 & 0 & \cdots & \lambda^H_{p,p,\bu} & \lambda^H_{p,p+1,\bu} & \cdots & \lambda^H_{p, q,\bu}
\end{matrix} \right),
\end{align*}
with
$$\lambda^H_{i,j,\bu}=\prod_{k=1}^{\alpha_i}\lambda^H_{i,j,k,\bu},
\quad\text{and}\quad \lambda^H_{i,j,k,\bu}=\sum_{\ell =
  0}^{n-s}\lambda_{i,j,k,\ell}X_\ell + \sum_{\ell =
  n-s+1}^{n}\lambda_{i,j,k,\ell}
\Phi_{\ell,\bu}(X_0,\dots,X_{n-s}).$$ 

Remark that the entries of $\mN^H_\bu$ are products of homogeneous
linear forms in $(n-s)+1$ variables $(X_0,\dots,X_{n-s})$, so that
this matrix has the form seen in~\eqref{eqdef:type1}. As a result, for
a generic choice of the coefficients $\mu_{i,k,\ell}$ and
$\lambda_{i,j,k,\ell}$, the first item in Lemma~\ref{lemma:appendix}
implies that there is no projective solution to $I_p(\mN^H_\bu)$
satisfying $X_0=0$. Taking into account all possible choices of $\bu$,
we deduce that there is no projective solution to $\bA^H$ satisfying
$X_0=0$, and  $\assG_2(0)$ is proved.

\paragraph{Property $\assG_3(0)$.} Finally, we have to prove that the Jacobian
matrix of $\bA$ has full rank $n$ at any point in $V(\bA) \subset
\KKbar{}^n$. Let thus $\bx=(x_1,\dots,x_n)$ be in $V(\bA)$; in
particular, $\tilde \bx=(1,x_1,\dots,x_n)$ is a projective solution of
$\bA^H$.  Thus, there exists $\bu=(u_1,\dots,u_s)$ as above such that
$$x_{n-s+1}=\phi_{n-s+1,\bu}(x_1\dots,x_{n-s}),\dots,x_{n}=\phi_{n,\bu}(x_1,\dots,x_{n-s}),$$
where $\phi_{k,\bu}(X_1,\dots,X_{n-s})=\Phi_{k,\bu}(1,X_1,\dots,X_{n-s})$ for 
all $k$, and such that
and $\mN^H_\bu$ has rank less than $p$ at $\tilde\bx'=(1,x_1,\dots,x_{n-s})$.  The second item of
Lemma~\ref{lemma:appendix} shows that the Jacobian matrix of
$M_p(\mN^H_\bu)$ with respect to $X_0,\dots,X_{n-s}$ has rank $n-s$
at $\tilde \bx'$. Since the first coordinate of $\tilde\bx'$ is non-zero,
and all generating polynomials of $I_p(\mN^H_\bu)$ are homogeneous,
Euler's relation implies that the Jacobian matrix of $M_p(\mN_\bu)$
with respect to $X_1,\dots,X_{n-s}$
has full rank $n-s$ at $\bx'=(x_1,\dots,x_{n-s})$, where  
\begin{align}\label{eqdef:Lu}
 \mN_\bu = \left( \begin{matrix}
\lambda_{1,1,\bu} & 0 & \cdots & 0 & \lambda_{1,p+1,\bu} & \cdots & \lambda_{1, q,\bu}\\
0 & \lambda_{2,2,\bu} & \cdots & 0 & \lambda_{2,p+1,\bu} & \cdots & \lambda_{2, q,\bu}\\
\vdots & \vdots & \ddots & \vdots & \vdots & \ddots & \vdots\\
0 & 0 & \cdots & \lambda_{p,p,\bu} & \lambda_{p,p+1,\bu} & \cdots & \lambda_{p, q,\bu}
\end{matrix} \right),
\end{align}
with
$$\lambda_{i,j,\bu}=\prod_{k=1}^{\alpha_i}\lambda_{i,j,k,\bu},
\quad\text{and}\quad \lambda_{i,j,k,\bu}=\lambda_{i,j,k,0}+\sum_{\ell
  = 1}^{n-s}\lambda_{i,j,k,\ell}X_\ell + \sum_{\ell =
  n-s+1}^{n}\lambda_{i,j,k,\ell}
\phi_{\ell,\bu}(X_1,\dots,X_{n-s}).$$ We now prove
that the Jacobian matrix of $\bA$ with respect to $X_1,\dots,X_n$ has
full rank at $\bx$.

The first step is similar to what we did in Section~\ref{sec:columndegree}.  For
$i=1,\dots,s$, $a_i$ is a product of linear forms of the form
$a_i=\prod_{k=1}^{\gamma_i} \mu_{i,k}$, with $\mu_{i,u_i}(\bx)=0$.
Since the coefficients $\mu_{i,k,\ell}$ are chosen generically, for
$i=1,\dots,s$ and $k \ne u_i$, $\mu_{i,k}(\bx)$ is non-zero; as a
result, in the local ring at $\bx$, the polynomials $(a_1,\dots,a_s)$
are equal (up to units) to the linear forms
$(\mu_{1,u_1},\dots,\mu_{s,u_s})$. This further implies that
\begin{align}\label{eq:subst}
X_{n-s+1}-\phi_{n-s+1,\bu}(X_1,\dots,X_{n-s}),\dots,X_{n}-\phi_{n,\bu}(X_1,\dots,X_{n-s})
\end{align}
belong to the ideal generated by $(a_1,\dots,a_s)$ in the local 
ring at $\bx$.

Next, we consider the $p$-minors $(a_{s+1},\dots,a_m)$ of $\mN$. Let
$\zeta \in\KK[X_1,\dots,X_n]$ be a $p$-minor of $\mN$, and let
$\zeta_\bu \in \KK[X_1,\dots,X_{n-s}]$ be the polynomial obtained
after applying the substitution in~\eqref{eq:subst} in $\mN$. Since
$\zeta$ and all polynomials in~\eqref{eq:subst} are in the ideal
$\langle \bA \rangle \cdot \mathcal{O}_\bx$, the polynomial
$\zeta_\bu$ is in this ideal as well. Now, note that $\zeta_\bu$ is a
$p$-minor of $\mN_\bu$ as defined in~\eqref{eqdef:Lu}, and that all
its $p$-minors are obtained this way. We pointed out above that the
Jacobian matrix of these equations with respect to $X_1,\dots,X_{n-s}$
has full rank $n-s$ at $\bx'$. As a result, taking all $\zeta_\bu$ into account,
together with the equations in~\eqref{eq:subst}, we obtain 
a family of polynomials in $\langle \bA \rangle \cdot \mathcal{O}_\bx$ 
whose Jacobian matrix has rank $n$ at $\bx$, and $\assG_3(0)$ is proved.

\medskip

In view of the previous paragraphs,
we can then apply Proposition~\ref{prop:degree_fiber}. Since $\bB$
satisfies $\assA_1,\assA_2$ and $\bA=\bB_0$ satisfies
$\assG_1,\assG_2,\assG_3$, we deduce that the sum of the
multiplicities of the isolated solutions of $\bC=\bB_1$ is at most
$\rc$, where $\rc$ is the number of solutions of $\bA$. Our next step
is to establish the value of $\rc$. This is done in Corollary~\ref{coro:complete}
below, which proves the first claim in Proposition~\ref{prop:rowdegree}.
Recall below that $S_t$ is the degree $t$ complete symmetric
function, for $t \ge 0$.

\begin{lemma}
  Let $\balpha=(\alpha_1,\dots,\alpha_p)$ be positive integers.  and
  let $S_{t}(\alpha_1,\dots,\alpha_p)$ be the complete symmetric
  function of degree $t$ in $\alpha_1,\dots,\alpha_p$. For generic $p
  \times q$ matrices $\mN$ as in~\eqref{eqdef:type1aff} or $\mM$ as
  in~\eqref{eqdef:type2aff}, with entries in $t=q-p+1$ variables,
$V_p(\mN)$ and $V_p(\mM)$ have
  cardinality $S_{t}(\alpha_1,\dots,\alpha_p)$.
\end{lemma}
\begin{proof}
  First, let us show that if the claim holds for $\mN$ in size $p
  \times q$, it holds for $\mM$ as well. To this effect, we set up a
  homotopy between $\mN$ and $\mM$, by considering the matrix
  $(1-T)\cdot\mN + T\cdot\mM$.  The discussion in the previous paragraphs shows
  that (for generic choices of the coefficients) this matrix satisfies
  Properties $\assA_1$ and $\assA_2$, together with
  $\assG_1(0),\assG_2(0),\assG_3(0)$.  We claim that
  $\assG_1(1),\assG_2(1),\assG_3(1)$ hold as well: the degree property
  in $\assG_1(1)$ is proved as we did for $\assG_1(0)$, and
  $\assG_2(1),\assG_3(1)$ are restatements of
  Lemma~\ref{lemma:appendix}. As a result, we can apply
  Proposition~\ref{prop:degree_fiber} to the specializations of
  $(1-T)\cdot\mN + T\cdot\mM$ at $T=0$ and $T=1$, and conclude that $V_p(\mN)$
  and $V_p(\mM)$ have the same cardinality, for generic choices of
  the coefficients of~$\mN$~and~$\mM$.

  We finish the proof by induction. If $p=q$, then $t=1$, $\mN$ is diagonal, and
  its determinant has degree $\alpha_1 + \cdots + \alpha_p =
  S_1(\alpha_1,\dots,\alpha_p)$, so our claim holds for $\mN$ (and
  thus for $\mM$). Suppose now that the claim is true  for all $p'\le p$ and all $q' < q$ with $p'
  \le q'$ and for all choices
  of degrees $(\alpha_1,\dots,\alpha_{p'})$. Following Algorithm~{\sf RowDegreeDiagonal}
  (which is essentially a restatement of Lemma~\ref{lemma:union}) and Remark~\ref{remark:disjoint} 
  (which states that the corresponding union is disjoint), we obtain
  $$
  |V_p(\mN)| = \sum_{\substack{\bi=(i_1,\dots,i_\kappa)\\\br=(r_1,\dots,r_\kappa)}}  |V_\kappa(\mM_{\bi,\br})|,
  $$ for all subsequences $\bi=(i_1,\dots,i_\kappa)$ of length $\kappa
  \in \{1,\dots,\min(t-1,p)\}$ and $\br=(r_1,\dots,r_\kappa)$, with $r_k
  \in \{1,\dots,\alpha_k\}$ for all $k$, and where 
  matrix $\mM_{\bi,\br}$
  has $\kappa \le p$ rows and $t-1 < q$ columns, with row degrees $(\alpha_{i_1},\dots,\alpha_{i_p})$;
  in particular, we can apply our induction assumption to such matrices.
  In addition, if $t \le p$, we should take into account one extra point 
  for each subsequence $(i_1,\dots,i_t)$ of $(1,\dots,p)$. Altogether, 
  we obtain 
  $$
  |V_p(\mN)| = \sum_{\substack{\bi=(i_1,\dots,i_\kappa),\\\br=(r_1,\dots,r_\kappa)}}  S_{t-\kappa}(\alpha_{i_1},\dots,\alpha_{i_\kappa}),
  $$
  for $\kappa \in \{1,\dots,\min(t,p)\}$, since $S_0=1$.
  For any given $\bi=(i_1,\dots,i_\kappa)$, there are $\alpha_{i_1}\cdots \alpha_{i_\kappa}$ 
  choices of indices $\br$, so that we have
  $$
  |V_p(\mN)| = \sum_{\bi=(i_1,\dots,i_\kappa)} \alpha_{i_1}\cdots \alpha_{i_\kappa} S_{t-\kappa}(\alpha_{i_1},\dots,\alpha_{i_\kappa}),
  $$
  for $\bi=(i_1,\dots,i_\kappa)$ subsequence of $(1,\dots,p)$ with $\kappa \in \{1,\dots,\min(t,p)\}$.
  The latter sum is precisely $S_t(\alpha_1,\dots,\alpha_p)$, so we are done.
\end{proof}

\begin{corollary}\label{coro:complete}
  For a generic choice of coefficients $\mu_{i,k,\ell}$ and
  $\lambda_{i,j,k,\ell}$, the cardinality $\rc$ of the algebraic set
  $V(\bA)$ is $\gamma_1 \cdots \gamma_s S_{n-s}(\alpha_1,\dots,\alpha_p)$.
\end{corollary}
\begin{proof}
  For a sequence $\bu=(u_1,\dots,u_s)$ as above, let $V_\bu$ be the
  subset of $V(\bA)$ consisting of all those points $\bx$ such that
  $\mu_{i,u_i}(\bx)=0$ for all $i$. Remark first that the sets $V_\bu$
  are (generically) pairwise disjoint: we pointed out above that for
  $\bx$ in $V_\bu$, any index $i$ and any $k \ne u_i$,
  $\mu_{i,k}(\bx)$ is non-zero.
  
  Let us thus fix $\bu=(u_1,\dots,u_s)$. The cardinality of $V_\bu$ is
  equal to the number of points in $V_p(\mN_\bu)$; this is a
  polynomial matrix of size $p \times q$, with entries that are
  products of linear forms in $n-s=q-p+1$ variables and with row
  degrees $\alpha_1,\dots,\alpha_p$. The previous lemma then shows
  that for any $\bu$, for generic choices of the coefficients, $V_\bu$
  has cardinality $S_{n-s}(\alpha_1,\dots,\alpha_p)$; the conclusion
  follows.
\end{proof}

%%%%%%%%%%%%%%%%%%%%%%%%%%%%%%%%%%%%%%%%%%%%%%%%%%%%%%%%%%%%

\subsection{Towards the homotopy algorithms}

Since $\assA_1$, $\assA_2$, $\assG_1(0)$, $\assG_2(0)$ and
$\assG_3(0)$ hold, we are going to apply
Proposition~\ref{prop:compute_regular} to first compute the simple
points in $\VpFG{p}{\mF}{G}$. Indeed, the resulting algorithm
$\mathsf{RowDegree\_simple}$ is used by the
$\mathsf{RowDegreeDiagonal}$ of the previous section, which itself
will be used to compute the isolated points of $\VpFG{p}{\mF}{G}$, by
means of Proposition~\ref{prop:compute_isolated}.

In the following paragraphs, we discuss the required properties
$\assD_1,\assD_2,\assD_3$ in this context. In what follows, we assume
that we are given a straight-line program $\Gamma$ of length $\sigma$
that computes the input matrix $\mF$ and the input equations $G$.
Besides, we also assume that all $\gamma_i$'s and $\alpha_j$'s are at
least equal to $1$.

\paragraph{Property $\assD_1$.}
To perform the homotopy, we need a zero-dimensional parametrization of
$V(\bA)$. We now describe how to obtain it; the process is based on
Algorithm {\sf RowDegreeDiagonal} given in the previous section, and
makes up the first two steps in Algorithm ${\sf RowDegree\_simple}$.

As a preliminary, we construct a straight-line program $\Delta$ that
computes the entries of~$\mN$: for all $i,j$, $\Delta$ computes and multiplies
the values of the $\alpha_i$ linear forms involved in $\lambda_{i,j}$
using $O(n \alpha_i)$ step. Its total length is $\sigma_\mN=O(n^2
(\alpha_1+\cdots+\alpha_p))$, which is $O(n^2 p \alpha)$,
with $\alpha = \max(\alpha_1, \ldots, \alpha_p)$.

Then, for any sequence $\bu=(u_1,\dots,u_s)$, with $u_j$ in
$\{1,\dots,\gamma_j\}$ for all $j$, we start by solving the equations
$\mu_{1,u_1} = \cdots = \mu_{s,u_s}=0$, to express
$(X_{n-s+1},\dots,X_n)$ as linear forms
$(\phi_{n-s+1,\bu},\dots,\phi_{n,\bu})$ in $(X_1,\dots,X_{n-s})$; this
takes a total of $O(\gamma_1 \cdots \gamma_s n^3 )$ operations in $\KK$.

From this, we deduce a straight-line program $\Delta_\bu$ that
computes the entries of matrix $\mN_\bu$ from~\eqref{eqdef:Lu}: it
simply consists in $\Delta$, to which we add $O(n^2)$ operations that
evaluate $(\phi_{n-s+1,\bu},\dots,\phi_{n,\bu})$. Given $\Delta_\bu$, we 
can then apply Algorithm {\sf RowDegreeDiagonal} to compute 
a zero-dimensional parametrization $\scrR'_\bu$ of $V_p(\mN_\bu)$.
The number of points $\rc$ in the output is $S_{n-s}(\alpha_1,\dots,\alpha_p)$
(Corollary~\ref{coro:complete}),
so by Lemma~\ref{lemma:rowdegreediagonal}, Algorithm {\sf RowDegreeDiagonal} takes time 
\begin{eqnarray}
T=\sum_{\substack{\bi=(i_1,\dots,i_\kappa)\\ \kappa\le\min(n-s-1,p)}}
\alpha_{i_1} \cdots \alpha_{i_\kappa} T_{M,{\rm row}}((\alpha_{i_1},\dots,\alpha_{i_\kappa}), n-s-1)
+
\softO( S_{n-s}(\alpha_1,\dots,\alpha_p) n^3).  \label{eqdef:T}
\end{eqnarray}
Since there are $\gamma_1 \cdots \gamma_s$ choices of $\bu$, 
the total cost is $\gamma_1 \cdots \gamma_s T$.

The next stage consists in adding to each $\scrR'_\bu$, which involves
only variables $(X_1,\dots,X_{n-s})$, the expressions of
$(X_{n-s+1},\dots,X_n)$ obtained from
$(\phi_{n-s+1,\bu},\dots,\phi_{n,\bu})$. As in the analysis of
Algorithm {\sf RowDegreeDiagonal}, the total runtime is $O(\gamma_1
\cdots \gamma_s S_{n-s}(\alpha_1,\dots,\alpha_p)  n^2)=O(\rc n^2)$. Finally, we
combine the resulting parametrizations $(\scrR_\bu)_\bu$ into a single
parametrization $\scrR$ using Chinese Remaindering, in time
$\softO(\gamma_1 \cdots \gamma_s  S_{n-s}(\alpha_1,\dots,\alpha_p) n)=\softO(\rc n)$.

Altogether, the overall time spent in computing the zero-dimensional
parametrization $\scrR$ of $V(\bA)$ is 
\begin{multline}\label{eq:VbA}
\gamma_1 \cdots \gamma_s T +\softO( \rc n^3)  \\  = \gamma_1 \cdots \gamma_s\sum_{\substack{\bi=(i_1,\dots,i_\kappa)\\ \kappa\le\min(n-s-1,p)}}
\alpha_{i_1} \cdots \alpha_{i_\kappa} T_{M,{\rm row}}((\alpha_{i_1},\dots,\alpha_{i_\kappa}), n-s-1)
+\softO( \rc n^3).
\end{multline}

\paragraph{Property $\assD_2$.} Next, we need to determine an upper bound 
$\re$ on the degree of the curve $V(J')$, where $J'$ is the union of
the one-dimensional irreducible components of $V(\bB) \subset
\KKbar{}^{n+1}$ whose projection on the $T$-axis is dense. Proceeding
as in Lemma~\ref{lemma:columndegree:e_estimate} of
Section~\ref{sec:columndegree}, we can take for $\re$ the integer
$(\gamma_1+1)\cdots(\gamma_s+1) S_{n-s}(\alpha_1
+1,\dots,\alpha_p+1)$.

\paragraph{Property $\assD_3$.} Finally, we need to give an estimate
on the size of a straight-line program that computes the polynomials
$\bB=(b_1,\dots,b_m)$, assuming that we are given a straight-line
program $\Gamma$ of size $\sigma$ that computes polynomials
$G=(g_1,\dots,g_s)$ and the entries of $\mF$. We already defined a
straight-line program $\Delta$ of size $\sigma_\mN$ that computes all
entries of $\mN$; for an extra $O({q \choose p} n^3)$ operations, we
can compute all entries of $\mU=(1-T)\cdot\mN+T\cdot\mF$ and all
$p$-minors $(b_{s+1},\dots,b_m)$ of this matrix.  Adding an extra
$O(n(\gamma_1 + \cdots + \gamma_s))\in O(n^2 \gamma)$ operations
(with $\gamma = \max(\gamma_1, \ldots, \gamma_s)$), we can also compute
all polynomials $(a_1,\dots,a_s)$, and thus $(b_1,\dots,b_s)$.

Altogether, we have obtained a straight-line program $\Gamma'$ that
computes $\bB=(b_1,\dots,b_m)$ using
$\sigma'=\sigma + \sigma_\mN + O({q \choose p} n^3 + n^2 \gamma)=\sigma +O( {q \choose p} n^3 + n^2 p \alpha+ n^2\gamma)$
operations. 

%%%%%%%%%%%%%%%%%%%%%%%%%%%%%%%%%%%%%%%%%%%%%%%%%%%%%%%%%%%%

\subsection{Algorithm $\mathsf{RowDegree\_simple}$}

Algorithm $\mathsf{RowDegree\_simple}$ that we deduce from the above
discussion is given hereafter. We indicate in the pseudo-code the
arithmetic costs for intermediate steps.

\begin{algorithm}[!t]
\caption{$\mathsf{RowDegree\_simple}(\Gamma)$}
{\bf Input}:  a straight-line program $\Gamma$ of length $\sigma$ that computes
 $F \in \KK[X_1, \ldots, X_n]^{p \times q}$ with $\deg(f_{i,j}) \leq \alpha_i$
and $G = (g_1, \ldots, g_s)$ in $\KK[X_1, \ldots, X_n]$ with $p \leq q$, $n = q-p+s+1$\\
{\bf Output}: a zero-dimensional parametrization of the isolated points of $\VpFG{p}{\mF}{G}$
\begin{enumerate}
\item construct a straight-line program $\Delta$ that computes
  $\mN \in \KK[X_1,\dots,X_n]^{p \times q}$ as
  in~\eqref{eqdef:type1aff}
  
\hfill length of $\Delta$ is $O(n^2 p \alpha)$
\item for any sequence $\bu=(u_1,\dots,u_s)$, with $u_j \in \{1,\dots,\gamma_j\}$ for all $j$
\begin{enumerate}
\item apply Gaussian elimination to the system of linear forms
  $\mu_{1,u_1}=\cdots=\mu_{s,u_s}=0$ given at \eqref{eqdef:ai} to
  rewrite $(X_{n-s+1}, \ldots, X_n)$ as linear forms
  $(\phi_{k,\bu})_{n-s+1 \le k \le n}$ in $(X_1,\dots,X_{n-s})$
  
\hfill  $\text{\sf{cost:~}} O(\gamma_1 \cdots \gamma_s n^3)$

\item construct a straight-line program $\Delta_{\bu}$ that computes the matrix 
  $\mN_\bu  \in \KK[X_1, \dots, X_{n-s}]^{p \times q}$ obtained
  by substituting $(\phi_{k,\bu})_{n-s+1 \le k \le n}$ into $\mN$
  
  \hfill length of $\Delta_{\bu}$ is $O(n^2 p\alpha)$

\item $\scrR'_\bu \gets \mathsf{RowDegreeDiagonal}(\Gamma_\bu)$ (points have coordinates $(X_1,\dots,X_{n-s})$
  
\hfill  $\text{\sf{cost:~}} \gamma_1 \cdots \gamma_s T$, for $T$ as in~\eqref{eqdef:T}

\item deduce $\scrR_\bu$ from $\scrR'_\bu$ by adding the expressions for $(X_{n-s+1},\dots,X_n)$
  
\hfill  $\text{\sf{cost:~}} O( \rc n^2)$, with $\rc= \gamma_1 \cdots \gamma_s S_{n-s}(\alpha_1,\dots,\alpha_p)$
\end{enumerate}
\item combine all $\scrR_\bu$ into $\scrR$

  \hfill $\text{\sf{cost:~}} \softO( \rc n)$

\item construct a straight-line program $\Gamma'$ that computes all polynomials $\bB$

  \hfill length of $\Gamma'$ is $\sigma'=O(\sigma + {q \choose p} n^3 + n^2 p \alpha + n^2\gamma)$

\item\label{step:final:rowdegreesimple} return $\mathsf{Homotopy\_simple}(\Gamma',\scrR)$

  \hfill $\text{\sf{cost:~}} \softO( \rc^2 m n^2 + \rc \re n (\sigma'  + n^2))$, 
with $\re=(\gamma_1+1) \cdots (\gamma_s+1) S_{n-s}(\alpha_1+1,\dots,\alpha_p+1)$
\end{enumerate}
\label{RowHom_simple}
\end{algorithm}

All cost estimates were given in the previous subsection and are
summarized in~\eqref{eq:VbA}, save for that of the last step. To
estimate its complexity, we apply
Proposition~\ref{prop:compute_regular}, which gives a runtime of
$\softO(\rc^2 m n^2 + \rc \re n (\sigma' + n^2))$ operations in~$\KK$
for the cost of calling the homotopy subroutine at the last step of
Algorithm {\sf RowDegree\_simple}.  Now, we write
$\sigma'+ n^2 = \sigma +O( {q \choose p} n^3 + n^2 p \alpha + n^2
\gamma)$,
for which we use the upper bound
${q \choose p} n^3(\sigma + p\alpha + \gamma)$ (recall
$\alpha=\max(\alpha_1, \ldots, \alpha_p)$ and
$\gamma=\max(\gamma_1, \ldots, \gamma_s)$).  This gives the upper
bound
\[
\softO\left (\rc^2 m n^2 + \rc\re n {q \choose p} n^3(\sigma + p\alpha + \gamma )\right ).
\]
Using the inequalities $\rc \le \re$ and $m\leq n + {q \choose p} \le n {q \choose p}$,
we see that the second term in the sum is dominant.
Thus, the bound for the cost of Algorithm {\sf Homotopy\_simple} becomes
\[
\softO\left ( \rc\re {q \choose p} n^4(\sigma + p\alpha + \gamma) \right ).
\]
Hence, the total cost of the algorithm is
\begin{align*}
  T_{\rm row}(\sigma,\bgamma,\balpha,q)= \gamma_1 \cdots\gamma_s T +
  \softO\left (  {q \choose p} n^4 \rc  \re ( \sigma + p \alpha +\gamma )  \right ),  
\end{align*}
with $T$ as in~\eqref{eqdef:T}. Since $\re \ge 2^n$ (because
$\alpha_i\geq 1$ and $\gamma_i\geq 1$ by assumption), this becomes
\begin{align}\label{eq:recT}
  T_{\rm row}(\sigma,\bgamma,\balpha,q)= \gamma_1 \cdots\gamma_s T +
  \softO\left (  {q \choose p}  \rc  \re ( \sigma + p \alpha +\gamma )  \right ).
\end{align}

This will now allow us to give an estimate on $T_{M,{\rm row}}$ by
solving a few recurrence relations. Recall that $T_{M,{\rm row}}$ describes
the case where $s=0$, so that $\gamma_1 \cdots \gamma_s=1$, and $\mM$
is a $p \times q$ input matrix as in~\eqref{eqdef:type2aff}. In this
case, we can take $\sigma=O((q-p) q(\alpha_1 + \cdots + \alpha_p)) \in
O((q-p)pq\alpha)$; following our convention in the previous section,
the runtime $T_{\rm row}(\sigma,(),(\alpha_1,\dots,\alpha_p),q)$ is
then written $T_{M,{\rm row}}((\alpha_1,\dots,\alpha_p),q)$.

\begin{lemma}\label{lemma:TMrow}
  One can take
  $$T_{M,{\rm row}}((\alpha_1,\dots,\alpha_p),q)= \softO\left ({q
    \choose p} S_{q-p+1}(\alpha_1,\dots,\alpha_p) S_{q-p+1}(\alpha_1+1,\dots,\alpha_p+1) pq \alpha \right), $$ with
 $\alpha =\max(\alpha_1,\dots,\alpha_p)$.
\end{lemma}
\begin{proof}
  Taking into account that $\gamma=1$, Equation~\eqref{eq:recT},
  combined with the definition of $T$ in~\eqref{eqdef:T}, gives the
  recursion
  \begin{multline}
  T_{M,{\rm row}}((\alpha_1,\dots,\alpha_p),q)=
\sum_{\substack{\bi=(i_1,\dots,i_\kappa)\\ \kappa\le\min(q-p,p)}}
\alpha_{i_1} \cdots \alpha_{i_\kappa} T_{M,{\rm row}}((\alpha_{i_1},\dots,\alpha_{i_\kappa}), q-p)
\\ + \softO\Big(
 {q \choose p} S_{q-p+1}(\alpha_1,\dots,\alpha_p) S_{q-p+1}(\alpha_1+1,\dots,\alpha_p+1) p q  \alpha )
\Big);
  \end{multline}
notice that 
a factor $(q-p)$ disappeared from the last term, since
it can be absorbed in the logarithmic factors in the $\softO(\ )$.
Let us rewrite the second summand as 
$S_{q-p+1}(\alpha_1,\dots,\alpha_p) C((\alpha_1,\dots,\alpha_p),q)$,
with
\[
C((\alpha_1,\dots,\alpha_p),q)= \softO\Big(
 {q \choose p} S_{q-p+1}(\alpha_1+1,\dots,\alpha_p+1) p q  \alpha 
\Big).
\]
This term is at its maximum at
  the root of the recursion tree.  Thus, we can find an upper bound on
  $T_{M,{\rm row}}$ by finding a solution to the recurrence
  \begin{align}\label{rec:TrowC}
  T_{M,{\rm row}}((\alpha_1,\dots,\alpha_p),q)&=
\sum_{\substack{\bi=(i_1,\dots,i_\kappa)\\ \kappa\le\min(q-p,p)}}
\alpha_{i_1} \cdots \alpha_{i_\kappa} T_{M,{\rm row}}((\alpha_{i_1},\dots,\alpha_{i_\kappa}), q-p)\\
&+ S_{q-p+1}(\alpha_1,\dots,\alpha_p) K,
  \end{align}
  for some constant $K$, and replacing $K$ by $
 \softO\Big(
 {q \choose p} S_{q-p+1}(\alpha_1+1,\dots,\alpha_p+1) p q  \alpha 
\Big). $  Now, a quick induction shows
that the   solution of~\eqref{rec:TrowC} satisfies 
$$T_{M,{\rm row}} \le   (q-p+1) S_{q-p+1} 
(\alpha_1,\dots,\alpha_p)K,$$  and the conclusion follows.
\end{proof}

We can then take the expression given in this lemma, and combine it
with the definition of $T$ given in~\eqref{eqdef:T}. Using  the fact that
$n-s-1 = q-p$, we have
\begin{align*}
T&=\sum_{\substack{\bi=(i_1,\dots,i_\kappa)\\ \kappa\le\min(q-p,p)}}
\alpha_{i_1} \cdots \alpha_{i_\kappa} T_{M,{\rm row}}((\alpha_{i_1},\dots,\alpha_{i_\kappa}), q-p)
+
\softO( S_{q-p+1}(\alpha_1,\dots,\alpha_p) n^3).
\end{align*}
Using the previous lemma,
we obtain that a term such as $T_{M,{\rm row}}((\alpha_{i_1},\dots,\alpha_{i_\kappa}),q-p)$
is
$$
 \softO\left ({q-p \choose \kappa} 
S_{q-p+1-\kappa}(\alpha_{i_1},\dots,\alpha_{i_\kappa}) 
S_{q-p+1-\kappa}(\alpha_{i_1}+1,\dots,\alpha_{i_\kappa}+1) \kappa (q-p) \alpha \right ).$$
As in the proof of the previous lemma,  we rewrite this expression by
factoring out the first complete function, as
$S_{q-p+1-\kappa}(\alpha_{i_1},\dots,\alpha_{i_\kappa}) D(\alpha_{i_1},\dots,\alpha_{i_\kappa},p,q),$ with
$$
D(\alpha_{i_1},\dots,\alpha_{i_\kappa},p,q)= \softO\left ({q-p \choose \kappa} 
 S_{q-p+1-\kappa}(\alpha_{i_1}+1,\dots,\alpha_{i_\kappa}+1) \kappa (q-p) \alpha \right ).$$
Now, we use the fact that (for the values of $\kappa$ that show up
in the sum), we have
\begin{align*}
{q-p \choose \kappa} & \le {q \choose p}\\
S_{q-p+1-\kappa}(\alpha_{i_1}+1,\dots,\alpha_{i_\kappa}+1) & \le S_{q-p+1}(\alpha_{1}+1,\dots,\alpha_{p}+1).
\end{align*}
Thus, 
$$
D(\alpha_{i_1},\dots,\alpha_{i_\kappa},p,q)= \softO\left ( 
{q \choose p}  S_{q-p+1}(\alpha_{1}+1,\dots,\alpha_{p}+1) p(q-p)\alpha\right ),$$
independently of the choice of $\alpha_{i_1},\dots,\alpha_{i_\kappa}$.
The sum in the definition of $T$ becomes
\[
\left(\sum_{\substack{\bi=(i_1,\dots,i_\kappa)\\ \kappa\le\min(q-p,p)}}
\alpha_{i_1} \cdots \alpha_{i_\kappa} 
S_{q-p+1-\kappa}(\alpha_{i_1},\dots,\alpha_{i_\kappa})\right)
\softO\left ( 
{q \choose p}  S_{q-p+1}(\alpha_{1}+1,\dots,\alpha_{p}+1) p(q-p)\alpha\right ),
\]
or equivalently
\[
\softO\left ( 
{q \choose p}  S_{q-p+1}(\alpha_{1},\dots,\alpha_{p}) S_{q-p+1}(\alpha_{1}+1,\dots,\alpha_{p}+1)p(q-p)\alpha \right ).
\]
The value of $T$ we infer from this is
\[
\softO\left ( 
{q \choose p}  S_{q-p+1}(\alpha_{1},\dots,\alpha_{p}) S_{q-p+1}(\alpha_{1}+1,\dots,\alpha_{p}+1)p(q-p)\alpha 
+  S_{q-p+1}(\alpha_{1},\dots,\alpha_{p}) n^3\right ).
\]
We inject this value in the runtime
analysis~\eqref{eq:recT}. Terms 
such as $(q-p)$ or $n^3$ are polylogarithmic in $\re$; removing them,
 the first-hand term $\gamma_1 \cdots \gamma_s T$ in~\eqref{eq:recT} is then
bounded above by the second one, so that the runtime is simply
\begin{align}\label{eq:recTT}
  T_{\rm row}(\sigma,\bgamma,\balpha,q)=  \softO\left (  {q \choose p}  \rc  \re ( \sigma + p \alpha +\gamma )  \right ).
\end{align}
This establishes Proposition~\ref{prop:rowdegree_simple}.

%%%%%%%%%%%%%%%%%%%%%%%%%%%%%%%%%%%%%%%%%%%%%%%%%%%%%%%%%%%%

\subsection{Algorithm $\mathsf{RowDegree}$}

Algorithm $\mathsf{RowDegree}$ is similar to
$\mathsf{RowDegree\_simple}$: the only difference consists in calling
Algorithm $\mathsf{Homotopy}$ from
Proposition~\ref{prop:compute_isolated} at the last step
\eqref{step:final:rowdegreesimple}, instead of
$\mathsf{Homotopy\_simple}$.

The cost of Algorithm $\mathsf{Homotopy}$ is
$\softO(\rc^5 m n^2  + \rc(\re+\rc^5) n(\sigma' + n^3))$.
Using the facts that $\sigma'= 
\sigma +O( {q \choose p} n^3 + n^2 p \alpha+ n^2\gamma)$,
and that $n$ is in $\softO(\re)$,
we rewrite this as 
$\softO(\rc^5 m n^2  + \rc(\re+\rc^5 ) {q \choose p}(\sigma  +  p \alpha+ \gamma))$.
Then, we use the inequality $m \le  {q \choose p} n$,
which gives $\rc^5 m n^2 \le (\re+\rc^5)  {q \choose p} n^3$;
hence the first term can be neglected, and the runtime of 
$\mathsf{Homotopy}$ is thus
\[
\softO\left( \rc(\re+\rc^5 ) {q \choose p}(\sigma  +  p \alpha+ \gamma)\right).
\]
The costs of all other steps are the same as those in
$\mathsf{RowDegree\_simple}$, and the analysis in the previous section
shows that can be neglected. As a result, the bound given above 
holds for the whole algorithm, and Proposition~\ref{prop:rowdegree}
is proved.

\paragraph*{Acknowledgments.} J.D. Hauenstein is supported by Sloan
Research Fellowship BR2014-110 TR14 and NSF grant
ACI-1460032. \'E. Schost is supported by an NSERC Discovery
Grant. M. Safey El Din and T.X. Vu are supported by the
ANR-17-CE40-0009 GALOP project and the GAMMA project funded by
PGMO/FMJH.

\vspace{-0.5cm}

\bibliographystyle{plain} \bibliography{roadmap}

\end{document}

%% file: macros.tex
\def\softO{\ensuremath{{O}{\,\tilde{ }\,}}}

\def\P {\ensuremath{\mathbb{P}}}
\def\Q {\ensuremath{\mathbb{Q}}}

\def\KKbar {\ensuremath{\overline{\mathbf{K}}}}
\def\KKCbar {\ensuremath{\overline{\mathbf{K}(\mathfrak{G})}}}
\def\KKGbar {\ensuremath{\overline{\mathbf{K}(\mathfrak{L})}}}
\def\KKGpbar {\ensuremath{\overline{\mathbf{K}(\mathfrak{L}')}}}

\def\LL {\ensuremath{\mathbf{L}}}
\def\LLbar {\ensuremath{\overline{\mathbf{L}}}}

\def\KK {\ensuremath{\mathbf{K}}}
\def\SS {\ensuremath{\mathbf{S}}}

\def\d {\ensuremath{\mathbf{d}}}

\def\mA{\ensuremath{{A}}}
\def\mU{\ensuremath{{U}}}
\def\mL{\ensuremath{{L}}}
\def\mF{\ensuremath{{F}}}
\def\mG{\ensuremath{{G}}}
\def\mH{\ensuremath{{H}}}

\def\mM{\ensuremath{{M}}}
\def\mN{\ensuremath{{N}}}

\def\jac{\ensuremath{{\rm jac}}}

\def\scrS{\ensuremath{\mathscr{S}}}

\def\scrR{\ensuremath{\mathscr{R}}}

\def\f {\ensuremath{\mathbf{f}}}

\def\X {\ensuremath{\mathbf{X}}}

\def\m {\ensuremath{\mathfrak{m}}}

\def\bh{\mbox{\boldmath$h$}}

\def\bell{\mbox{\boldmath$\ell$}}

\def\KKbar {\ensuremath{\overline{\mathbf{K}}}}
 
\DeclareBoldMathCommand{\bM}{M}
\DeclareBoldMathCommand{\bD}{D}
\DeclareBoldMathCommand{\bB}{B}
\DeclareBoldMathCommand{\bA}{A}
\DeclareBoldMathCommand{\bC}{C}
\DeclareBoldMathCommand{\bX}{X}
\DeclareBoldMathCommand{\bi}{i}
\DeclareBoldMathCommand{\bk}{k}
\DeclareBoldMathCommand{\bY}{Y}
\DeclareBoldMathCommand{\bn}{n}
\DeclareBoldMathCommand{\bbb}{b}
\DeclareBoldMathCommand{\d}{d}
\DeclareBoldMathCommand{\bc}{c}
\DeclareBoldMathCommand{\bj}{j}
\DeclareBoldMathCommand{\be}{e}
\DeclareBoldMathCommand{\bh}{h}
\DeclareBoldMathCommand{\br}{r}
\DeclareBoldMathCommand{\bu}{u}
\DeclareBoldMathCommand{\bv}{v}
\DeclareBoldMathCommand{\bx}{x}
\DeclareBoldMathCommand{\by}{y}
\DeclareBoldMathCommand{\btheta}{\vartheta}
\DeclareBoldMathCommand{\c}{c}
\DeclareBoldMathCommand{\f}{f}
\DeclareBoldMathCommand{\g}{g}
\DeclareBoldMathCommand{\h}{h}
\DeclareBoldMathCommand{\x}{x}
\DeclareBoldMathCommand{\bell}{\ell}
\DeclareBoldMathCommand{\bbeta}{\beta}
\DeclareBoldMathCommand{\balpha}{\alpha}
\DeclareBoldMathCommand{\bgamma}{\gamma}

\def\rdeg{\ensuremath{\mathrm{rdeg}}}
\def\cdeg{\ensuremath{\mathrm{cdeg}}}

\def\assH{\ensuremath{\mathsf{A}}}
\def\assA{\ensuremath{\mathsf{B}}}
\def\assG{\ensuremath{\mathsf{C}}}
\def\assD{\ensuremath{\mathsf{D}}}
\def\assI{\ensuremath{\mathsf{J}}}
\def\assJ{\ensuremath{\mathsf{K}}}

\newcommand{\VpF}[2]{V_{#1}(#2)}
\newcommand{\VpFG}[3]{V_{#1}(#2,#3)}

\def\rc{\ensuremath{{c'}{}}}
\def\re{\ensuremath{{e'}{}}}